%% file: main.tex
\documentclass[11pt]{article}
\usepackage[utf8]{inputenc}
\usepackage{mystyle}

\begin{document}

\title{\huge Spectral Ranking Inferences based on General Multiway Comparisons\thanks{
					Emails: \texttt{jqfan@princeton.edu}, \texttt{ZHL318@pitt.edu}, \texttt{weichenw@hku.hk}, and \texttt{mengxiny@wharton.upenn.edu}.   The research was supported by ONR grant N00014-22-1-2340 and NSF grants DMS-2210833, DMS-2053832, and DMS-2052926. The work described in this paper was also partially supported by a grant from the Research Grants Council of the Hong Kong Special Administrative Region, China (Project No. 27307623).}} 

\author{Jianqing Fan\qquad Zhipeng Lou\qquad Weichen Wang\qquad Mengxin Yu}


\maketitle

\begin{abstract}
This paper studies the performance of the spectral method in the estimation and uncertainty quantification of the unobserved preference scores of compared entities in a general and more realistic setup. Specifically, the comparison graph consists of hyper-edges of possible heterogeneous sizes, and the number of comparisons can be as low as one for a given hyper-edge.  Such a setting is pervasive in real applications, circumventing the need to specify the graph randomness and the restrictive homogeneous sampling assumption imposed in the commonly-used Bradley-Terry-Luce (BTL) or Plackett-Luce (PL) models. 
Furthermore, in scenarios where the BTL or PL models are appropriate, we unravel the relationship between the spectral estimator and the Maximum Likelihood Estimator (MLE). We discover that a two-step spectral method, where we apply the optimal weighting estimated from the equal weighting vanilla spectral method, can achieve the same asymptotic efficiency as the MLE. 
Given the asymptotic distributions of the estimated preference scores, we also introduce a comprehensive framework to carry out both one-sample and two-sample ranking inferences, applicable to both fixed and random graph settings. It is noteworthy that this is the first time effective two-sample rank testing methods have been proposed. 
Finally, we substantiate our findings via comprehensive numerical simulations and subsequently apply our developed methodologies to perform statistical inferences for statistical journals and movie rankings.
\end{abstract}

\input{intro}

\input{theory}



\newpage
\bibliographystyle{ims}
\bibliography{dynamic}

\newpage
\appendix
\input{appendix_spectral}

\end{document}

%% file: intro.tex
\section{Introduction}

Rank aggregation is crucial in various applications, including web search \citep{dwork2001rank,wang2016learning}, primate intelligence experiments \citep{johnson2002bayesian},
assortment optimization \citep{aouad2018approximability,chen2020dynamic}, recommendation systems \citep{baltrunas2010group, li2019estimating}, sports ranking \citep{massey1997statistical, turner2012bradley}, education \citep{avery2013revealed, caron2014bayesian}, voting \citep{plackett1975analysis, mattei2013preflib}, and instruction tuning used in the recent popular large language model ChatGPT \citep{ouyang2022training}. Therefore, it becomes an essential problem in many fields, such as psychology, econometrics, education, operation research, statistics, machine learning, artificial intelligence, etc. 

Luce~\citep{Luce1959} introduced the celebrated~{Luce's axiom of choice}. Let $p(i\given A)$ be the probability of selecting item $i$ over all other items in the set of alternatives $A$. According to the axiom, when comparing two items $i$ and $j$ in any sets of alternatives $A$ containing both $i$ and $j$, the probability of choosing $i$ over $j$ is unaffected by the presence of other alternatives in the set. In specific, the axiom postulates that
$$
{\frac{\PP(i \text{ is preferred in } A)}{\PP(j \text{ is preferred in } A)} = \frac{\PP(i \text{ is preferred in } \{i,j\})}{\PP(j \text{ is preferred in } \{i,j\})}\,}.
$$
This assumption gives rise to a unique parametric choice model, the Bradley-Terry-Luce (BTL) model for pairwise comparisons, and the Plackett-Luce (PL) model for $M$-way rankings $M\ge 2$.

In this paper,  we consider a collection of $n$ items whose true ranking is determined by some unobserved preference scores $\theta_{i}^*$ for $i=1,\cdots,n$.  In this scenario, the BTL model assumes that an individual or a random event ranks item $i$ over $j$ with probability {$\PP(\text{item $i$ is preferred over item $j$}) = e^{\theta_i^*} / (e^{\theta_i^*} + e^{\theta_j^*})$}. The Plackett-Luce model is an expanded version of pairwise comparison, which allows for a more comprehensive $M$-way full ranking, as initially described in \cite{plackett1975analysis}. { This model takes individual preferences into account when ranking a selected subset of items with size $M<\infty$ (among all $n$ items)}, which we represent as $i_1 \succ \dots \succ i_M$. Think of this full ranking as $M-1$ distinct events where $i_1$ is favored over the set ${i_1, \dots, i_M}$, followed by $i_2$ being favored over the set ${i_2,\dots,i_M}$, and so on.
The PL model calculates the probability of a full ranking $i_1 \succ \dots \succ i_M$ using the formula:
$$
    \PP(i_1 \succ \dots \succ i_M) = \prod_{j=1}^{M-1} \bigg[e^{\theta_{i_j}^*} / \sum_{k=j}^M e^{\theta_{i_k}^*}\bigg].
$$
Next, we will give a brief introduction to the literature that has made progress on model estimation and uncertainty quantification for the BTL and the PL models over the parametric model.

\subsection{Related literature}

A series of papers studied model estimation or inference based on the BTL or PL models. In the case of the Bradley-Terry-Luce model, its theoretical characteristics were solidified through a minorization-maximization algorithm, as outlined by  \cite{hunter2004mm}.  Additionally, \cite{negahban2012iterative}  developed an iterative rank aggregation algorithm called \emph{Rank Centrality} (spectral method), which can recover the BTL model's underlying scores at an optimal $\ell_2$-statistical rate. Subsequently, \cite{chen2015spectral} used a two-step methodology (spectral method followed by MLE) to examine the BTL model in a context where the comparison graph is based on the Erdős-Rényi model, where every item pair is assumed to have a probability $p$ of being compared, and once a pair is connected, it will be compared for $L$ times. 
Subsequently, under a similar setting with~\cite{chen2015spectral}, \cite{chen2019spectral} investigated the optimal statistical rates for recovering the underlying scores, demonstrating that regularized maximum likelihood estimation (MLE) and the spectral method are both optimal for retrieving top-$K$ items when the conditional number is a constant. They derived $\ell_2$- as well as $\ell_{\infty}$-rates for the unknown underlying preference scores in their study. Furthermore, \cite{chen2022partial} extended the investigation of \cite{chen2019spectral} to the partial recovery scenarios and improved the analysis to un-regularized MLE.

Expanding beyond simple pairwise comparisons, researchers also explored ranking issues through $M$-way comparisons, where $M\ge 2$. The Plackett-Luce model and its variations serve as prominent examples in this line of study, as evidenced by numerous references \citep{plackett1975analysis,guiver2009bayesian,cheng2010label,hajek2014minimax,maystre2015fast,szorenyi2015online,jang2018top}.
In particular, \cite{jang2018top} investigated the Plackett-Luce model within the context of a uniform hyper-graph, where a tuple with size $M$ is compared with probability $p$ and once a tuple is connected or compared in the hyper-graph, it will be compared for $L$ times. By dividing $M$-way comparison data into pairs, they employ the spectral method to obtain the $\ell_{\infty}$- statistical rate for the underlying scores. Additionally, they presented a lower bound for sample complexity necessary to identify the top-$K$ items under the Plackett-Luce model. In a more recent development, under the same model setting, \cite{fan2022ranking} enhanced the findings of \cite{jang2018top}, focusing solely on the top choice. Rather than splitting the comparison data into pairwise comparisons, they applied the Maximum Likelihood Estimation (MLE) and matched the sample complexity lower bound needed to recover the top-\emph{K} items. This contrasts with \cite{jang2018top}, which requires a significantly denser comparison graph or a much larger number of comparisons for their results to hold.


The aforementioned literature primarily concentrated on the non-asymptotic statistical consistency for estimating item scores. However, the results of limiting distribution for ranking models remained largely unexplored. Only a limited number of findings on asymptotic distributions of estimated ranking scores exist under the Bradley-Terry-Luce model, whose comparison graphs are sampled from the Erdős-Rényi graph with connection probability $p$ and each observed pair has the same number of comparisons $L$. For example, \cite{simons1999asymptotics} established the asymptotic normality of the BTL model's maximum likelihood estimator when all comparison pairs are entirely conducted (i.e., $p=1$). \cite{han2020asymptotic} expanded these findings for dense, but not fully connected, comparison graphs (Erdős-Rényi random graphs) where $p \gtrsim n^{-1/10}$. Recently, \cite{liu2022lagrangian} introduced a Lagrangian debiasing approach to derive asymptotic distributions for ranking scores, accommodating sparse comparison graphs with $p \asymp 1/n$ but necessitating comparison times $L \gtrsim n^2$. Furthermore, \cite{gao2021uncertainty} employed a ``leave-two-out'' technique to determine asymptotic distributions for ranking scores, achieving optimal sample complexity and allowing $L=O(1)$ in the sparse comparison graph setting (i.e., $p \asymp 1/n$ up to logarithm terms). \cite{fan2022uncertainty} extended upon existing research by incorporating covariate information into the BTL model. By introducing an innovative proof, they 
presented the MLE's asymptotic variance with optimal sample complexity when $p \asymp 1/n$ and $L=\cO(1)$. In the sequel, \cite{fan2022ranking} broadened the asymptotic results of the BTL model to encompass Plackett-Luce (PL) models with $M\ge 2$ again with optimal sample complexity. They also developed a unified framework to construct rank confidence intervals, requiring the number of comparisons $L\gtrsim \textrm{poly}(\log n)$. Moreover, recently, \cite{han2023unified} further extended the settings in \cite{fan2022ranking} by investigating the asymptotic distribution of the MLE, where the comparisons are generated from a mixture of Erdős-Rényi graphs with different sizes or a hypergraph stochastic block model.  

{Finally, we discuss related literature of an important application of our framework: assortment optimization~
\citep{Talluri2004, Rusmevichientong2012robust, Vulcano2012, Davis2014assortment, Zhang2020assortment, chen2020dynamic, Chen2023robust, Shen2023} which is of great importance in revenue management. 
Specifically, in their settings, each product (including the no-purchase alternative) is also associated with an unknown customer preference score, which can characterize the customers' choice behavior over a set of offered products. Based on the consistently estimated preference scores and the available profit information of each product, various efficient algorithms have been proposed to determine the optimal assortment under different kinds of practical constraints~\citep{Talluri2004, Rusmevichientong2010dynamic, Gallego2014constrained, Sumida2021revenue}. Moreover, uncertainty quantification of the estimated preference scores also enables statistical inference on the properties of the optimal assortment~\citep{Shen2023}. }


\subsection{Motivations and Contributions}
In this section, we discuss our motivation and problem settings and compare our results with previous literature in different aspects, namely the comparison graph, the connection between the spectral method and MLE, and ranking inferences.

\subsubsection{Comparison graph} 
\label{general_graph}
Previous studies on parametric ranking models mostly require comparison graphs derived from a specific random graph. For example, in the endeavor of understand multiway comparisons (e.g. \cite{jang2016top, fan2022ranking}), it is typically assumed that the comparisons are generated explicitly from a homogeneous graph with a known distribution.
This assumption may pose a challenge in certain contexts. Although practical applications with homogeneous comparisons exist, it is sometimes unrealistic to presume that all comparisons are generated from a known homogeneous distribution. In fact, there are more cases where we see heterogeneous comparisons, where some are compared more often than others and the comparison graph generation process is unknown. We will present Example~\ref{intro_exmp1} as our motivation.

\begin{example}
\label{intro_exmp1} 
There is a sequence of customers buying goods. For the $l$-th customer, according to her preference, her reviewed products are $A_l$ (a choice set), and she finally chose product $c_l\in A_l$ (her top choice). Then the total datasets presented to us is $\{(A_l,c_l)\}_{l=1}^{D}.$ 
If all choice sets are presented to the customers with the same probability and the reviewed number of items are of the same size (such as pairwise comparisons), we say the comparison graph is homogeneous. But if some choice sets, possibly with different sizes, are presented more often to the customers or are chosen arbitrarily based on the customers' preference profiles, then this heterogeneous comparison graph cannot be well approximated by a given random graph model. That is, the comparisons may not follow, say, the BTL or PL models with Erdős-Rényi types of uniform graphs. 
 
\end{example}
\begin{figure}[]
    \centering
       \begin{tabular}{cc}
   \hskip-30pt \includegraphics[height=3.4in,width=3.3in]{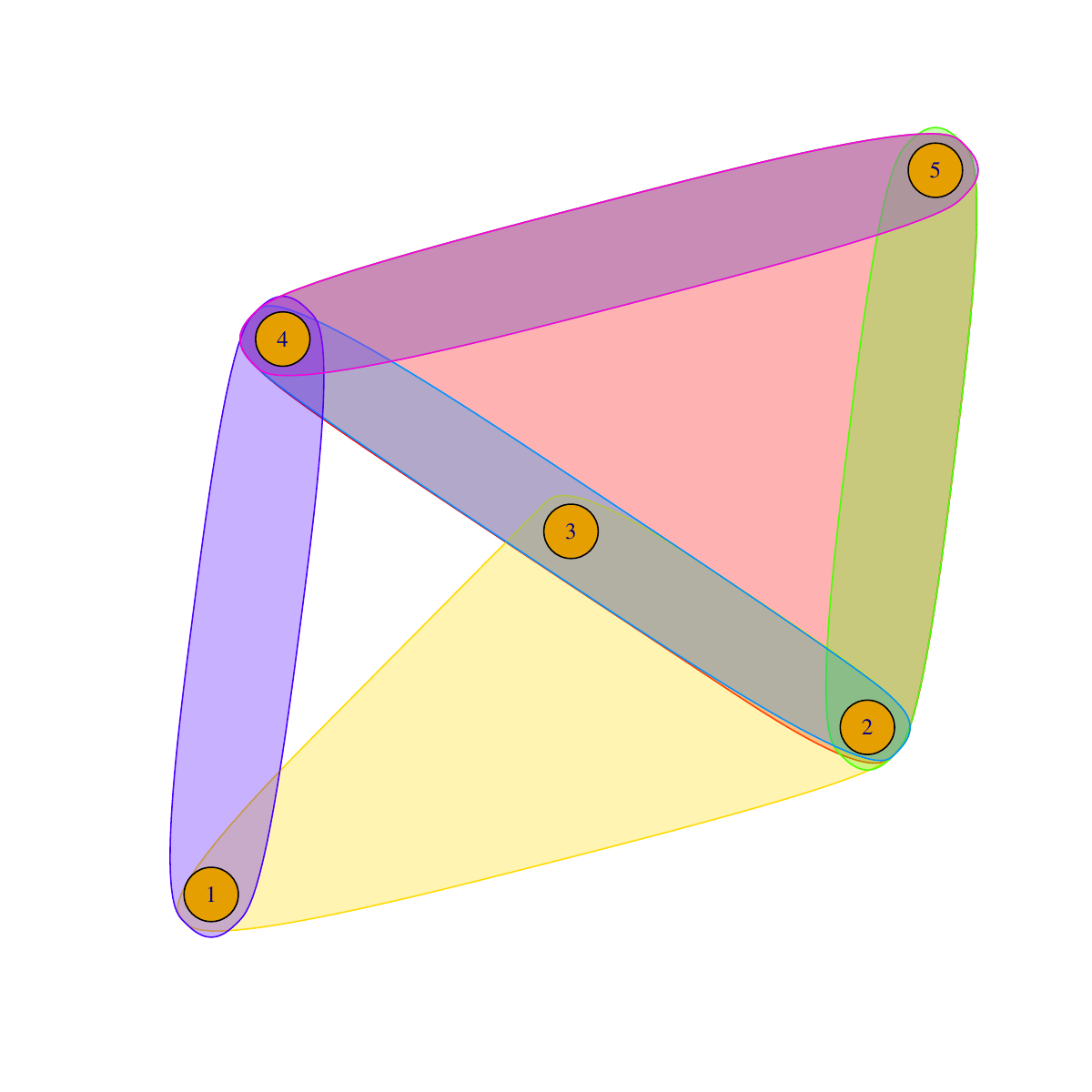}
   &\hskip-10pt \includegraphics[height=3.4in,width=3.5in]{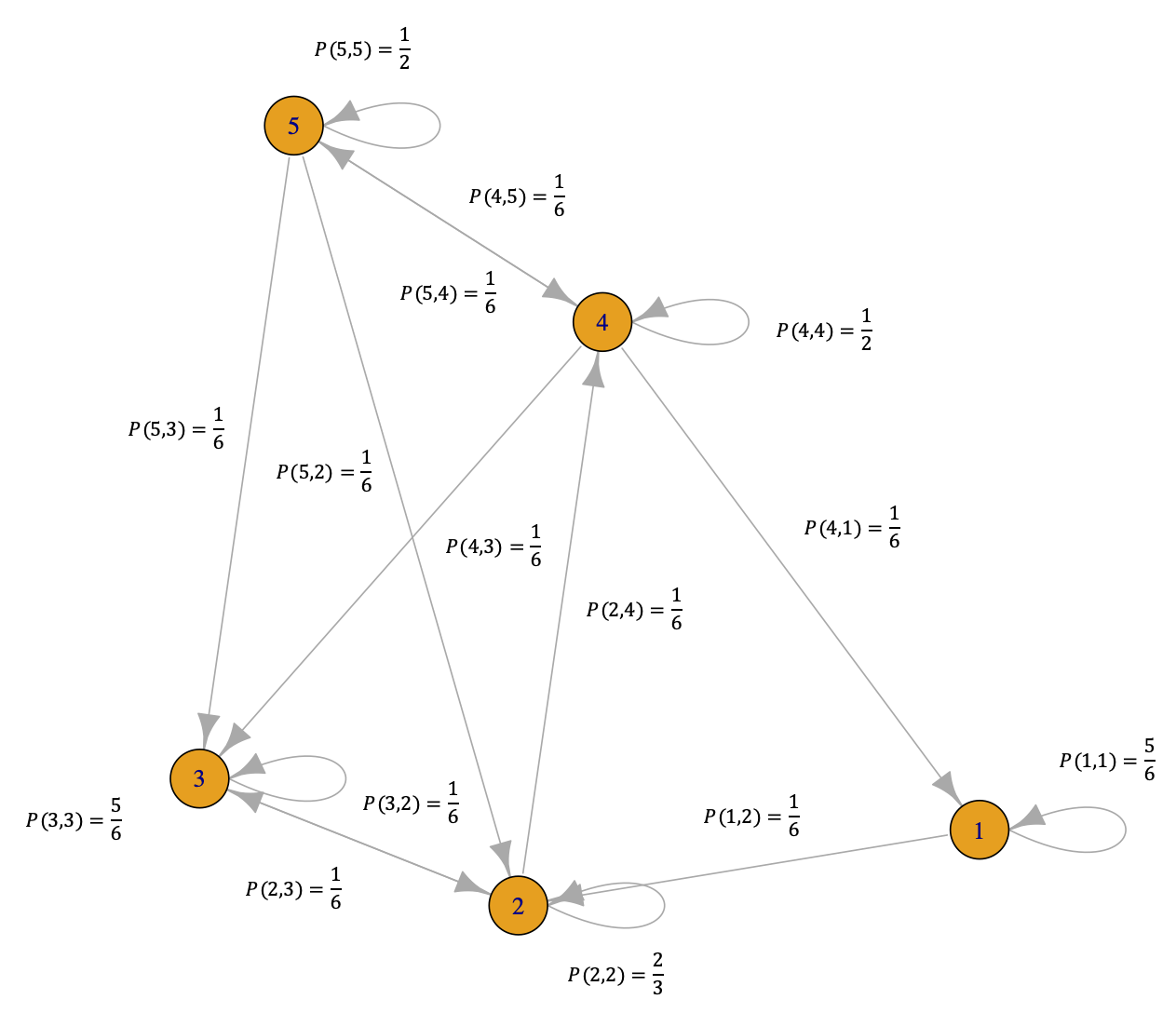}
   \end{tabular}
    \caption{A simple illustration of a collection of $\{c_l, A_l\}_{\ell=1}^{D}$ with $(c_1,A_1)=(3,\{2,3,4,5\})$, $(c_2,A_2)=(2,\{1,2,3\})$, $(c_3,A_3)=(2,\{2,5\})$, $(c_4,A_4)=(4,\{4,5\})$, $(c_5,A_5)=(4,\{2,4\})$, $(c_6,A_6)=(1,\{1,4\})$, $(c_7,A_7)=(5,\{4,5\}).$ Left panel illustrates the choice sets $\{A_l\}_{\ell=1}^7,$ where all nodes inside each $A_l$ are surrounded by an open area with the same color. The right panel presents the comparison-induced Markov transition matrix, whose computation is detailed in Section \ref{sec:mle_con}. A directed edge from $i$ to $j,(j\neq i)$ exists if and only if 
    $i,j,(i\neq j)$ are compared in some $A_l$ and $j$ is the winner ($c_l=j$). 
    }
    \label{fig:hyper}
\end{figure}


{The heterogeneous comparison scheme in Example~\ref{intro_exmp1} above is applicable across a wide range of practical scenarios. For instance, it covers the typical setup of the assortment optimization, wherein the no-purchase alternative is also included in the choice set $A_{l}$.} We give a toy example with $5$ products in Figure~\ref{fig:hyper}, where sizes of $A_l$ are among $\{2,3,4\}$. Due to the heterogeneity,
it is unrealistic to assume $A_l$ is of the same size or sampled from an explicit random graph. However, most of the previous works focused on statistical properties under certain ad-hoc random graphs, most commonly the Erdős-Rényi type of uniform graphs, e.g., \cite{chen2015spectral,chen2019spectral,jang2016top,gao2021uncertainty,liu2022lagrangian,fan2022ranking,han2023unified}. 
One interesting piece of research that indeed worked with the fixed comparison graph is \cite{shah2015estimation}, which explored the optimality of MLE in $\ell_2$ estimation with pairwise comparisons. Following this work, \cite{li2022ell} further discussed the optimality of MLE in $\ell_\infty$ error. Still, little has been known about the inference results for general fixed comparison graphs.

In this paper, we focus on the setting of a general fixed comparison graph, where we circumvent the modeling for the complicated graph-generating process. 
Specifically, we study statistical estimation and uncertainty quantification of the preference scores over a general, heterogeneous choice set via the spectral method. 
In addition, we also study the theoretical performance of the spectral method when we do have a homogeneous random comparison graph, and compare the results. Our results require slightly stronger conditions to handle fixed graph settings since we need to make sure each item has enough information to be ranked. 

For the general setting, we denote the choice preference for the $l$-th comparison as $(c_l, A_l)$, wherein $A_l$ signifies the set of choices with heterogeneous size, which can be either fixed or random, and $c_l \in A_l$ represents the most preferred item in $A_l$. Hence, in the $l$-th comparison, we understand that $c_l$ outranks all other elements within $A_l$. As such, our broadest comparison data is symbolized as $\mathcal D = \{l | (c_l, A_l)\}$. The associated collection of choice sets is denoted as $\mathcal G = \{A_l | l \in \mathcal D\}.$ This framework also contains the Plackett-Luce model as a special case, if we treat the PL model as $M-1$ selections over $M-1$ choice sets. Under mild conditions, we manage to obtain optimal statistical rates for our spectral estimator conditional on the comparison graphs and specify the explicit form of its asymptotic distribution. This gives an efficient solution when one encounters heterogeneous comparison graphs. In addition, since the graph is fixed or conditioned upon, it is not necessary for us to repeat each comparison for $L \geq 1$ times. We can even accommodate situations where a choice set is chosen and compared for just a single time, which is true in many practical applications. 

\subsubsection{Connection between the spectral estimator and the MLE}

Our general setting, as introduced in Section \ref{general_graph}, encompasses homogeneous random comparison graphs as a particular instance. The bulk of prior research has centered around the evaluation of the Maximum Likelihood Estimator (MLE) or the spectral method when applied to homogeneous comparisons \citep{chen2015spectral,chen2019spectral,chen2022partial,fan2022uncertainty,fan2022ranking}. Both are effective methods for examining ranking issues within the context of the BTL or PL model. Hence, an interesting question arises: What is the relationship between the MLE and the spectral method?

A handful of studies have offered insights into this question. For instance, \cite{maystre2015fast} identified a link between the MLE and spectral method in multiway comparisons, where the spectral estimator aligns with the MLE through the iterative updating of specific weighting functions in constructing the spectral estimator. This connection was only limited to the first order in the sense that the paper only concerns the convergence property.
Furthermore, \cite{gao2021uncertainty} demonstrated that the asymptotic variance of the spectral method exceeds that of the MLE in pairwise comparisons using the BTL model. However, this discrepancy arises from their choice of a suboptimal weighting function for the spectral estimator.

In our paper, we leverage the homogeneous random comparison graph case (as it is the most popularly studied setting in many previous articles) to illustrate that by employing certain optimally estimated information weighting functions, the asymptotic variance of the spectral estimator matches that of the MLE with multiway comparisons in the PL model. Therefore, the MLE could be considered a ``two-step'' spectral method, where in the first step we consistently estimate the unobserved preference scores, and in the second step we use the proper weighting in the spectral method to achieve an efficient estimator.
It is also noteworthy that we achieve the optimal sample complexity over the sparsest sampling graph up to logarithmic terms.

\subsubsection{Ranking inferences: one sample vs two samples}
As another contribution, we also study several rank-related inference problems. We have the following motivating example:
 \begin{example} 
 First, we consider the one-sample inference problem. Consider a group of candidate items $\{1,\cdots,n\}$ and one observed dataset on their comparisons; we are interested in
 \begin{itemize}
 \item Building the confidence intervals for the ranks of certain items $\{r_1,\cdots, r_m\}$ of our interest.
 \item Testing whether a given item $m$ is in the top-\emph{K} set, which includes $K$ best items.
  \end{itemize}
 Second, we consider the two-sample inference problem. For two groups of datasets of the same list of items $\{1,\cdots,n\}$, we are interested in
 \begin{itemize}
     \item Testing whether the rank of a certain item $m$ is preserved across these two samples (e.g. groups or time periods).
     \item Testing whether the top-\emph{K} sets have changed or not.
 \end{itemize}
\end{example}



Rankings are ubiquitous in real-world applications, for instance, in the evaluation of universities, sports teams, or web pages. However, most of these rankings provide only first-order information, presenting the results without offering any measure of uncertainty. For example, under the BTL model, when two items have equivalent underlying scores, there's a $50\%$ probability of one being ranked higher than the other. Thus, rankings between these two items could be unreliable due to their indistinguishable underlying scores, emphasizing the necessity for confidence intervals in rankings.

Given these critical considerations, our study offers a comprehensive framework that efficiently solves the problems outlined in Example 1.2 over heterogeneous comparison graphs. Additionally, our approach enhances the sample complexity of several previous works. For instance, when restricting our general framework to homogeneous random comparison graphs, regarding the one-sample inference, \citet{liu2022lagrangian} required $L \gtrsim n^2$ to carry out the statistical inference, while~\citet{fan2022ranking} further improved this requirement to $L \gtrsim \textrm{poly}(\log n)$. In our paper, our framework can allow $L=\cO(1)$ and even $L=1$. Furthermore, two-sample ranking inference, which can be widely applied in real-world scenarios like policy evaluation, treatment effect comparison, change point detection, etc., has not been previously studied. Our paper also introduces a general framework for studying the two-sample ranking inference problems, again offering optimal sample complexity. 

\subsubsection{{Theoretical contributions}}
{
We build up our theoretical analyses based on some previously developed techniques from \cite{gao2021uncertainty} and \cite{chen2019spectral}. However, our proofs have the following novelty:
In those two papers and other papers with a random comparison graph, graph randomness and ranking outcome randomness are typically intertwined in the analysis. We will separate them and reveal the proper quantities of interest to summarize the information in a fixed comparison graph. 
We study the connection between these quantities and the spectral ranking performance, and provide sufficient conditions under the fixed graph for valid ranking inference. This theoretical attempt has not previously been seen in the literature. In addition, all our analyses allow varying comparison sizes and an arbitrary number of repetitions of each comparison set. This significantly broadens the applicability of our proposed methodology, as in practice, a lot of ranking problems contain non-repeating comparisons of different numbers of items. We also work on the theory when we also have graph randomness. We realize that the homogeneity of sampling each comparison tuple can lead to more relaxed assumptions. With more relaxed conditions, we clearly show where and how we can achieve an improved performance guarantee (see the assumption and proof of Theorem \ref{thm_consist_normality_rand}). This part highlights the difference between fixed and random graphs and provides more theoretical insights into the role of graph randomness in spectral ranking. 
}

\subsection{Roadmap and notations}

In Section \ref{sec:algo}, we set up the model and introduce the spectral ranking method. Section \ref{sec:ranking_estimation} is dedicated to the examination of the asymptotic distribution of the spectral estimator, based on fixed 
comparison graphs and random graphs with the PL model, respectively. Within the same section, we also introduce the framework designed for the construction of rank confidence intervals and rank testing statistics for both one-sample and two-sample analysis. Section \ref{sec:ranking_inference} details the theoretical guarantees for all proposed methodologies. Sections \ref{sec:numerical}  and \ref{sec:realdata} contain comprehensive numerical studies to verify theoretical results and two real data examples to illustrate the usefulness of our ranking inference methods. Finally, we conclude the paper with some discussions in Section \ref{sec:discussion}. All the proofs are deferred to the appendix. 

Throughout this work, we use $[n]$ to denote the index set $\{1,2,\cdots,n\}.$ For any given vector $\xb\in \RR^{n}$ and $q\ge 0$, we use $\|\xb\|_{q}=(\sum_{i=1}^{n}|x_i|^q)^{1/q}$ to represent the vector $\ell_q$ norm.  For any given matrix $\Xb\in \RR^{d_1\times d_2}$, we use $\|\cdot\|$ to denote the spectral norm of $\Xb$ and  write $\Xb\succcurlyeq 0$ or $\Xb\preccurlyeq 0$ if $\Xb$ or $-\Xb$ is positive semidefinite. For event $A$, $1(A)$ denotes an indicator which equals $1$ if $A$ is true and $0$ otherwise. For two positive sequences $\{a_n\}_{n\ge 1}$, $\{b_n\}_{n\ge 1}$, we write $a_n=\cO(b_n)$ or $a_n\lesssim b_n$ if there exists a positive constant $C$ such that $a_n / b_n\le C$ and we write $a_n=o(b_n)$ if $a_n/b_n\rightarrow 0$. In addition, $\cO_p(\cdot)$ and $o_p(\cdot)$ share similar meanings as $\cO(\cdot)$ and $o(\cdot),$ respectively, but these relations hold asymptotically with probability tending to 1. Similarly we have $a_n=\Omega(b_n)$ or $a_n\gtrsim b_n$ if $a_n/b_n\ge c$ with some constant $c>0$. We use $a_n=\Theta(b_n)$ or $a_n\asymp b_n$ if $a_n=\cO(b_n)$ and $a_n=\Omega(b_n)$.  For two random variables $A_n, B_n$, if we write $A_n\approx B_n$, it holds that $A_n-B_n=o(1)$ with probability goes to $1$. Given $n$ items, we use $\theta_i^*$ to indicate the underlying preference score of the $i$-th item. 
Define $r : [n] \to [n]$ as the rank operator on the $n$ items, which maps each item to its population rank based on the preference scores. We write the rank of the $i$-th item as $r_i$ or $r(i)$. By default, we consider ranking from the largest score to the smallest score.

%% file: theory.tex
\section{Multiway Comparison Model and Spectral Ranking} \label{sec:algo}
We first introduce a general discrete choice model, which encompasses the classical Plackett-Luce model as well as fixed comparison graph scenario. 

\subsection{Discrete choice model} 
We assume there are $n$ items to be ranked. According to Luce's choice axiom~\citep{Luce1959}, the preference scores of a given group of $n$ items can be parameterized as a vector $(\theta_{1}^{*}, \ldots, \theta_{n}^{*})^{\top}$ 
such that $\PP(i \textrm{ wins among } A)=e^{\theta_i^*}/(\sum_{k\in A}e^{\theta_k^*})$ for any choice set $A$ and item $i\in A$. Since the parameters are only identifiable up to a location shift, without loss of generality, we assume $\sum_{i = 1}^{n} \theta_{i}^{*} = 0$ for identification. 
We consider the general comparison model, where we are given a collection of comparisons and outcomes $\{(c_l,A_l)\}_{l\in \cD}$. Here $c_l$ denotes the selected item over the choice set $A_l$ with probability $e^{\theta_{c_l}^*}/(\sum_{k\in A_l}e^{\theta_k^*})$.
 \begin{remark} 
 \label{rem2.1}
 	This general comparison model contains many well-known special cases. 
 	\begin{itemize}
 		\item For the Bradley-Terry-Luce (BTL) model, it is easy to set $A_l$ as the pair being compared every time. If each pair is compared for $L$ times independently, we just need to write the outcomes as $(c_l,A_l)$ and re-index $l$.
 		\item For the Plackett-Luce (PL) model, we have obtained the full ranking of a choice set $B=\{i_1,\cdots,i_B\}$. The probability of observing a certain ranking becomes 
 {  \begin{align*}
   \PP(i_1 >i_2 >\cdots>i_B)&=\PP( i_1 \textrm{ wins among } C_1\given C_1=B)\cdot \PP( i_2 \textrm{ wins among } C_2\given C_2=B\{-i_1\})\cdots\\&\qquad  \cdot\PP(i_{B-1} \textrm{ wins among } C_{B-1}\given C_{B-1}=B\{-i_1,\cdots,-i_{B-2}\})\\&\qquad =\frac{e^{\theta_{i_1}^*}}{\sum_{j=1}^B e^{\theta_{i_j}^*}}\cdot \frac{e^{\theta_{i_2}^*}}{\sum_{j=2}^B e^{\theta_{i_j}^*}}\cdots \frac{e^{\theta_{i_{B-1}}^*}}{\sum_{j=B-1}^B e^{\theta_{i_j}^*}},
   \end{align*}
   where $C_i,i\ge 1$ is the $i$-th comparison set and $B\{-i_1,\cdots,-i_M\}$ denotes the set of remained items after removing $\{i_1,\cdots,i_M\}.$}
 		These comparison results can also be decomposed into the comparisons: $$\{(i_1, B),(i_2,B\{-i_1\}),\cdots,(i_{B-1}, B\{-i_1,\cdots,-i_{B-2}\})\}.$$
        \item With fixed comparison graphs, $\{A_l, l \in \cD\}$ are given and hence have no randomness, so the comparison results in $c_l$ are assumed independent. {In contrast, with a random comparison graph, such as in the PL model, $A_l$ may be dependent. For instance, $(\theta_{i_1}, B)$ and $(\theta_{i_2},B\{-i_1\})$ are dependent as $B\{-i_1\}$ depends on the winner of the first comparison $i_1$.} Therefore, we have to explicitly lay out the random process assumption for comparison generation in order to study the theoretical properties in a case-by-case manner. 
        \end{itemize}
 	\end{remark}

Recall that the general comparison data is denoted as $\{ (c_l, A_l)\}_{l \in \cD}$. The corresponding collection of choice sets is $\mathcal G = \{A_l | l \in \mathcal D\}$. When we only have pairwise comparisons, $|A_l|=2$ and $\mathcal G$ represents the set of all edges we have compared. But in a general setting, $A_l$ can have different sizes and we denote $M = \max_{l \in \mathcal D} |A_l|<\infty$ as the maximal size of the comparison hyper-graph edge. Also if we have $L$ independent comparisons of the same comparison hyper-edge $A_l$, we use different $l$ to indicate the comparison. So in $\mathcal G$, the hyper-edge $A_l$ may be compared for multiple times with different outcomes $c_l$.

Throughout this paper, we consider using the spectral method on the preference data based on multiway comparisons. We will first focus on the fixed comparison graph and then consider commonly-used random comparison graph structures. Notice that no matter whether the comparison graph $\mathcal G$ is fixed or random, our spectral ranking methodology will be conditional on $\mathcal G$, which is observed in practice. The underlying model for generating $\mathcal G$ can be very general: it can be given, or random based on the Erdős–Rényi random graph with the same probability $p$, or more generally induced from some other comparison rules, which could even cause some $A_l$ dependent. For example, if we view each comparison of the PL model as $M-1$ pairwise comparisons involving top 1 vs top 2, top 2 vs top 3, ..., top $M-1$ vs top $M$. Then the resulting comparison data, denoted as $(c_l=i_k, A_l=\{i_k, i_{k+1}\})$ for $k=1,\dots, M-1$, are dependent (even the definition of $A_l$ depends on the complete comparison result).

\subsection{Spectral ranking}
\label{sec:mle_con} 
In the spectral method, we formally define a Markov chain, denoted as $M=(S,P)$. Here, $S$ signifies the collection of $n$ states corresponding to the $n$ items to be compared, represented as vertices of a directed comparison graph. And $P$ constitutes the transition matrix defined below. This matrix oversees transitions amongst the states by representing whether any two particular states within $S$ are capable of being connected via non-zero transition probabilities.

Define two index sets $\mathcal W_j, \mathcal L_i$ for comparisons, with $j$ as the winner and $i$ as the loser:
\begin{align*}
\mathcal W_j = \{l \in \cD| j \in A_l, c_l = j\}\,, \quad\quad
\mathcal L_i = \{l \in \cD | i \in A_l, c_l \ne i\}\,.
\end{align*}
So their intersection for $i\ne j$ gives all situations where $i,j$ are compared and $j$ wins, i.e., $\mathcal W_j \cap \mathcal L_i = \{l \in \cD | i,j \in A_l, c_l = j\}$. Define the transition matrix $P$ with transition probability 
\[
P_{ij} = \begin{cases}
      \frac{1}{d} \sum_{l \in \mathcal W_j \cap \mathcal L_i} \frac{1}{f(A_l)}\,, & \text{if $i\ne j$,}\\
      1 - \sum_{k: k\ne i} P_{ik}\,, & \text{if $i = j$.}
    \end{cases}  
\]
Here, $d$ is chosen to be large enough so that the diagonal element is non-negative, but not too large to give enough transition probability. When the comparison graph is random, we choose $d$ to make non-negative diagonal elements with probability approaching $1$ by studying the concentration inequality for $\sum_{k: k\ne i} P_{ik} $. Here, $f(A_l)>0$ is a weighting function to encode the total information in the $l$-th comparison. A natural choice is $f(A_l) = |A_l|$ giving more weight to hyper-edges with a smaller number of compared items. We will discuss later the optimal choice of $f(\cdot)$.  

When $i\ne j$, $P_{ij}$ can also be written as 
\begin{align*}
    P_{ij} = \frac{1}{d} \sum_{l \in \mathcal{D}} 1(i, j \in A_{l}) 1(c_{l} = j) \frac{1}{f(A_{l})}. 
\end{align*}
Conditioning on $\mathcal G$, the population transition is
\begin{align*}
    P_{ij}^* = E[P_{ij}|\mathcal G] = \begin{cases}
      \frac1d \sum\limits_{l \in \mathcal D} 1(i, j \in A_{l}) \frac{e^{\theta_j^*}}{\sum_{u\in A_l} e^{\theta_u^*}} \frac{1}{f(A_l)} \,, & \text{if $i\ne j$,}\\
      1 - \sum_{k: k\ne i} P_{ik}^*\,, & \text{if $i = j$.}
    \end{cases}  
\end{align*}
Let 
\[
\pi^* = (e^{\theta_1^*}, \dots, e^{\theta_n^*}) / \sum_{k=1}^n e^{\theta_k^*}\,.
\]
Note that both $\sum_{u\in A_l} e^{\theta_u^*}$ and $f(A_l)$ in the denominator are symmetric with respect to $\theta_i^*,\theta_j^*$ as long as both $i,j$ belong to $A_l$. So we have $P_{ij}^* \pi_i^* = P_{ji}^* \pi_j^*$. This is the so-called detailed balance that leads to $\pi^*$ being the stationary measure of the Markov chain with the above population transition for any $f(\cdot)$. That is ${\pi^*}^{\top} P^* = {\pi^*}^{\top}$, namely, $\pi^*$ is the top-left eigenvector of $P^*.$ 

The spectral method estimates $\pi^*$ by using the stationary measure $\hat\pi$ of the empirical transition $P$, namely, 
\[
\hat\pi^{\top} P = \hat\pi^{\top}\,. 
\]
Note that if we consider the directed graph induced by $P$ to be strongly connected, this implies that the Markov chain it generates will be ergodic, which ensures the existence of a unique stationary distribution $\hat\pi$ as defined above. 
Consider the toy example in Figure \ref{fig:hyper}, if we naively choose $f(\cdot)=1$ as a constant weighting function, it is not hard to count the times that $j$ beats $i$ and fill the value into the transition matrix $P_{ij}$ (divided by $d$). In the right panel, the transition probabilities are calculated with $d=6$ to guarantee the self-loop transition happens with a positive probability. For this $P$, the stationary distribution is $\hat\pi = (0.199, 0.531, 0.796, 0.199, 0.066)^\top$, meaning that the estimated ranking of preference scores of the 5 products are $3 \succ 2 \succ 1 = 4 \succ 5$.

Finally given the identifiability condition of $1^\top \theta^* = 0$, we can estimate $\theta_i^*$ by 
\begin{equation}
\tilde\theta_i := \log \hat\pi_i - \frac1n \sum_{k=1}^n \log \hat\pi_k \,. \label{eq:spectrla_estimator}
\end{equation}
It is worth mentioning that the spectral estimator is easier to compute in practice, by only requiring one-step eigen-decomposition.  Indeed, we need only the eigenvector that corresponds to the largest eigenvalue, which can even be computed very fast by the power method.  In comparison, the MLE is typically computationally heavier in terms of data storage and step size determination during the implementation of the gradient descent algorithm.


\section{Ranking Inference Methodology}
\label{sec:ranking_estimation}
In this section, we study the inference methodology for the spectral estimator for the underlying scores $\{\theta_i^*\}_{i\in [n]}$ of $n$ items. To be specific, we need to establish the statistical convergence rates and asymptotic normality for $\tilde\theta_i$.

\subsection{Uncertainty quantification with fixed comparison graph}
\label{Section_uncertainty_quantification_fixed}
For the estimation of $\pi^*$, we use the following two approximations, which we will justify later to be accurate enough so as not to affect the asymptotic variance. Let us first focus on our intuition. Firstly, we have
\[
\hat\pi_i = \frac{\sum_{j:j\ne i} P_{ji}\hat\pi_j}{\sum_{j:j\ne i} P_{ij}} \approx \frac{\sum_{j:j\ne i} P_{ji}\pi_j^*}{\sum_{j:j\ne i} P_{ij}} =: \bar\pi_i\,.
\]
Equivalently, 
\begin{equation}
\frac{\hat\pi_i - \pi_i^*}{\pi_i^*} \approx  \frac{\bar\pi_i - \pi_i^*}{\pi_i^*} = \frac{\sum_{j:j\ne i} (P_{ji}\pi_j^* - P_{ij}\pi_i^*)}{\pi_i^* \sum_{j:j\ne i} P_{ij}}\,.
\label{approx1}
\end{equation}
Secondly, the denominator above can be approximated by its expected value so that \eqref{approx1} can further be approximated as
\begin{equation}
J_i^* := \frac{\sum_{j:j\ne i} (P_{ji} e^{\theta_j^*} - P_{ij} e^{\theta_i^*})}{\sum_{j:j\ne i} E[P_{ij}|\mathcal G] e^{\theta_i^*}}  \,,
\label{approx2}
\end{equation}
by using $\pi_i^* \propto e^{\theta_i^*}$.
We will rigorously argue that the asymptotic distributions of $\frac{\hat\pi_i - \pi_i^*}{\pi_i^*}$ and $J_i^*$ are identical. For now, let us look at the asymptotic distribution of $J_i^*$. Obviously, it is mean zero due to the detailed balance: $E[P_{ji}|\mathcal G] \pi_j^* = E[P_{ij}|\mathcal G] \pi_i^*$. The denominator of $J_i^*$ is a constant and can be explicitly written out as follows:
\begin{equation}
\label{definition_tau}
    {\tau_{i}(\theta^{*}) :
    = \sum_{j : j \neq i} E[P_{ij}|\mathcal{G}] e^{\theta_i^*} } 
    = \frac{1}{d} \sum_{l \in \mathcal{D}} 1(i \in A_{l}) \left(1 - \frac{e^{\theta_{i}^{*}}}{\sum_{u \in A_{l}} e^{\theta_{u}^{*}}}\right) \frac{e^{\theta_{i}^{*}}}{f(A_{l})}.  
\end{equation}
Thus, $J_i^*$ can be expressed as 
\begin{equation}\label{eq:J_decomp_fixed}
J_i^* = \frac{1}{d \tau_{i}} \sum_{l\in\mathcal D} \frac{1(i \in A_{l})}{f(A_l)} \bigg( 1(c_l = i) \sum_{u\in A_l, u\ne i} e^{\theta_u^*} - e^{\theta_i^*} \sum_{u\in A_l, u\ne i} 1(c_l = u) \bigg) =: \frac{1}{d} \sum_{l\in\mathcal D} J_{il}(\theta^*)\,,
\end{equation}
where $\tau_{i}$ is short for $\tau_{i}(\theta^{*})$.
Since each $(c_l, A_l)$ is independent in the fixed graph setting (see Remark \ref{rem2.1} for discussions), the variance of $J_i^*$ is
\begin{align}
\mbox{Var}&(J_i^* | \mathcal G) 
 = \frac{1}{d^2 \tau_{i}^2} \sum\limits_{l \in \mathcal D} \frac{1(i \in A_{l})}{f^2(A_l)} \cdot \mbox{Var}\bigg(1(c_l = i) \sum_{u\in A_l, u\ne i} e^{\theta_u^*} - e^{\theta_i^*} \sum_{u\in A_l, u\ne i} 1(c_l = u) \bigg) & \nonumber\\
& = \Bigg({\sum\limits_{l \in \mathcal D} 1(i \in A_{l}) \frac{(\sum_{u\in A_l} e^{\theta_u^*} - e^{\theta_i^*})e^{\theta_i^*}}{f(A_l)^2}}\bigg) \bigg/ {\bigg[\sum\limits_{l \in \mathcal D} 1(i \in A_{l}) \bigg(\frac{\sum_{u\in A_l} e^{\theta_u^*} - e^{\theta_i^*}}{\sum_{u\in A_l} e^{\theta_u^*}}\bigg) \frac{e^{\theta_i^*}}{f(A_l)}\bigg]^2}\, , \label{eq:fan1}
\end{align}

A few important comments are in order. Firstly, 
the function $f$ achieves the minimal variance in \eqref{eq:fan1} when $f(A_l) \propto \sum_{u\in A_l} e^{\theta_u^*}$ due to 
simply applying the Cauchy-Schwarz inequality.
Actually, \cite{maystre2015fast} showed that when $f(A_l) = \sum_{u\in A_l} e^{\theta_u^*}$, spectral method estimator converges to the MLE up to the first order.
Secondly, under the situation of pairwise comparison in the BTL model, each $(c_l, A_l)$ is independent and we assume in $\mathcal D$ each pair $(i,j)$ is either compared for $L$ times (denoted as $\tilde A_{ij} = 1$) or never compared (denoted as $\tilde A_{ij} = 0$). Further assuming $f(A_l) = |A_l| = 2$, we have 
\[
\mbox{Var}(J_i^* | \mathcal G) = \frac{1}{L} \bigg({\sum_{j:j\ne i} \tilde A_{ij} e^{\theta_i^*} e^{\theta_j^*}}\bigg) \bigg/ {\bigg[ \sum_{j:j\ne i} \tilde A_{ij} \frac{e^{\theta_i^*} e^{\theta_j^*}}{e^{\theta_i^*} +e^{\theta_j^*}}\bigg]^2} \,.
\]
This exactly matches with Proposition 4.2 of \cite{gao2021uncertainty}. In addition, if we choose $f(A_l) = \sum_{u\in A_l} e^{\theta_u^*}$, we get the most efficient variance just like the MLE variance given in Proposition 4.1 of \cite{gao2021uncertainty}, which is
\[
\mbox{Var}(J_i^* | \mathcal G) = \bigg(L \cdot \sum_{j:j\ne i} \tilde A_{ij} \frac{e^{\theta_i^*} e^{\theta_j^*}}{(e^{\theta_i^*} +e^{\theta_j^*})^2}\bigg)^{-1} \,.
\]

With the above discussion and computation, after some additional derivations, we come to the conclusion that $\tilde\theta_i - \theta_i^*$ has the same asymptotic distribution as $\frac{\hat\pi_i - \pi_i^*}{\pi_i^*}$ and $J_i^*$. Therefore,
\begin{align*}
   \mbox{Var}(J_i^* | \mathcal G)^{-1/2} (\tilde\theta_i-\theta_i^*)\Rightarrow N(0,1),
\end{align*}
for all $i\in[n]$. Based on this result, we can make inference for $\tilde\theta_i$ and additionally the rank of item $i$ (see Section \ref{sec:onesam_twosided_inf}). The rigorous derivations for this conclusion will be provided in Section \ref{sec:ranking_inference}.

\subsection{Uncertainty quantification for the PL model} \label{Section_uncertainty_quantification_random}


In this section, we consider the case of a random comparison graph, which could potentially lead to dependent $(c_l, A_l)$, unlike the case of the fixed graph in the previous section. Note that since the random comparison graph generation can be arbitrary, we cannot work with each case. As an illustrating example, we consider the classical PL model from the Erdős–Rényi graph here for two reasons. Firstly, this is the most popularly studied random comparison model in the ranking literature \citep{chen2015spectral,chen2022partial,han2020asymptotic,liu2022lagrangian,gao2021uncertainty,fan2022uncertainty,fan2022ranking}. Secondly, we can use the model to verify the uncertainty quantification of the spectral method and compare it with that of the MLE.
It turns out the model is good enough to give us new insights. 
To further simplify the discussion and presentation, we will only focus on $M=3$ in the PL model. Results for general $M$ can be similarly derived with the same conclusion. 


With the PL model, we can write down the specific variance of $J_i^*$. Consider the most natural way to encode a $3$-way comparison. Say one ranks $(i,j,k)$ as $i \succ j \succ k$ where $a \succ b$ means $a$ is better than $b$. Motivated by the likelihood function, which multiplies the probability of selecting $i$ as the best one from $\{i,j,k\}$ and the probability of selecting $j$ next from $\{j,k\}$, we break this complete $3$-way comparison into two dependent comparison data: $(i, \{i,j,k\})$ and $(j,\{j,k\})$. We call this {\it multi-level breaking}, where in the first level comparison of all three items, $i$ is preferred, and in the second level comparison of the remaining items, $j$ is preferred. By doing this, we can naturally link and compare results with the MLE estimator. \cite{azari2013generalized} also proposed other ways of breaking an $M$-way comparison into pairwise comparisons, but different breaking methods will lead to different dependent structures, which we do not intend to analyze one by one in this work. So in the sequel, we only consider multi-way breaking, motivated by the likelihood function, and leave the study of other possible breaking methods to the future. We use $\tilde A_{ijk} = 1$ or $0$ to denote whether $(i,j,k)$ has been compared for $L$ times or is never compared.

Let us work on the multi-level breaking. Now the key difference is that the induced comparison graph $\mathcal G$ cannot be treated as fixed. Instead, we condition on $\tilde{\mathcal G} = \{\tilde A_{ijk}\}$. Similar to \eqref{approx2}, we have 
\begin{equation}
\frac{\bar\pi_i - \pi_i^*}{\pi_i^*} = \frac{\sum_{j:j\ne i} P_{ji}\pi_j^* - P_{ij}\pi_i^*}{\pi_i^* \sum_{j:j\ne i} P_{ij}} \approx 
\frac{\sum_{j:j\ne i} P_{ji} e^{\theta_j^*} - P_{ij} e^{\theta_i^*}}{\sum_{j:j\ne i} {E[P_{ij}|\tilde{\mathcal G}]} e^{\theta_i^*}} =: J_i^* \,.
\label{approx3}
\end{equation}
In the case of the random comparison graph, i.e. conditioning on $\tilde{\mathcal G}$, we obtain
\begin{align}
P_{ij} = \frac{1}{d} \sum_{\ell=1}^L \sum\limits_{k: k \ne i,j} \tilde A_{ijk} Z_{{ijk}}^\ell \,, \label{p_ij}
\end{align}
where $Z_{ijk}^{\ell} = 1(y^{(\ell)}_{k\succ j\succ i}=1) / f(\{i,j\}) + 1(y^{(\ell)}_{j\succ i\succ k}=1) / f(\{i,j,k\})+1(y^{(\ell)}_{j\succ k\succ i}=1) / f(\{i,j,k\})$. Here $y^{(\ell)}_{i_1\succ i_2\succ i_3}$ is a binary variable which equals to $1$ when the event $i_1\succ i_2\succ i_3$ holds under the $\ell$-th comparison among items $\{i_1,i_2,i_3\}$. Essentially, we need to collect all terms induced from the same comparison into one term $Z_{ijk}^\ell$ so that the summation is always over independent terms. 

We lightly abuse the notation of $J_i^*$ although here the expectation is conditioning on $\tilde{\mathcal G}$ instead of $\mathcal G$ used in the fixed graph case.
Note that 
$$
E[Z_{ijk}^{\ell} |\tilde{\mathcal G}] = 
 \frac{ e^{\theta_k^*} e^{\theta_j^*} }{(e^{\theta_i^*} + e^{\theta_j^*}+e^{\theta_k^*})(e^{\theta_i^*}+ e^{\theta_j^*})f(\{i,j\})} + \frac{ e^{\theta_j^*} }{(e^{\theta_i^*} + e^{\theta_j^*}+e^{\theta_k^*})f(\{i,j,k\})}.
$$
Using $\sum_{j\not = k} a_{ijk} = \sum_{j < k} (a_{ijk} + a_{ikj})$, 
the denominator of $J_i^*$ can be  expressed as
\begin{align*}
\tau_{i}^{\diamond}(\theta^*) := & \sum_{j:j\ne i} E[P_{ij}|\tilde{\mathcal G}] e^{\theta_i^*} = \frac{L}{d}  \sum\limits_{j < k: j,k \ne i} \tilde A_{ijk}  e^{\theta_i^*} \bigg( 
\frac{e^{\theta_j^*} e ^{\theta_k^*} }{(e^{\theta_i^*} + e^{\theta_j^*}+e^{\theta_k^*})(e^{\theta_i^*}+ e^{\theta_j^*})f(\{i,j\})}\\
& + \frac{e^{\theta_j^*} e^{\theta_k^*} }{(e^{\theta_i^*} + e^{\theta_j^*}+e^{\theta_k^*})(e^{\theta_i^*} + e^{\theta_k^*})f(\{i, k\})} + \frac{e^{\theta_j^*}+e^{\theta_k^*}}{(e^{\theta_i^*} + e^{\theta_j^*} + e^{\theta_k^*}) f(\{i,j,k\})} 
  \bigg) \,.
\end{align*}
Hence, the expression of $J_i^*$ is given as follows:
\begin{equation} \label{eq:J_decomp_random}
\begin{aligned}
J_i^* & = \frac{1}{\tau_{i}^{\diamond}} \bigg(\sum_{j:j\ne i} P_{ji}e^{\theta_j^*} - P_{ij}e^{\theta_i^*}\bigg) =\frac{1}{d \tau_{i}^{\diamond}} \bigg(\sum_{\ell=1}^{L}\sum_{j:j\neq i}\sum_{k:k\neq i,j}\tilde A_{ijk}(Z_{jik}^{\ell}e^{\theta_j^*}-Z_{ijk}^{\ell}e^{\theta_i^*})\bigg) \\
&=\frac{1}{d \tau_{i}^{\diamond}} \sum_{\ell=1}^{L}\sum_{j<k: j,k\neq i}\tilde A_{ijk}(Z_{jik}^{\ell}e^{\theta_j^*}+Z_{kij}^{\ell}e^{\theta_k^*}-Z_{ijk}^{\ell}e^{\theta_i^*}-Z_{ikj}^{\ell}e^{\theta_i^*})
=: \frac{1}{d} \sum_{\ell=1}^{L}\sum_{j<k: j,k\neq i} J_{ijk\ell}(\theta^*)\,,
\end{aligned}
\end{equation}
where $\tau_{i}^{\diamond}$ is short for $\tau_{i}^{\diamond}(\theta^*)$.
Since each 3-way comparison is independent, it can be shown that the variance of $J_i^*$ is 
\begin{align*}
\mbox{Var}\big(J_i^* | \tilde {\mathcal G} \big) = \frac{L}{d^2 (\tau_{i}^{\diamond})^2} \sum\limits_{j < k: j,k \ne i} \tilde A_{ijk} e^{\theta_i^*}\bigg(\frac{(e^{\theta_j^*} + e^{\theta_k^*}) }{f^2(\{i,j,k\})} + \frac{e^{\theta_j^*} e^{\theta_k^*} }{e^{\theta_i^*}+e^{\theta_j^*}+e^{\theta_k^*}} \Big(\frac{1}{f^2(\{i,k\})} + \frac{1}{f^2(\{i,j\})}\Big) \bigg )\,.
\end{align*}
The essential component is to compute 
$E J_{ijk\ell}(\theta^*)^2$ due to independence and zero-mean.  For a given triplet $(i, j, k)$, there are 6 possible preference outcomes with probabilities governed by the PL model.  Averaging the squared random outcomes over 6 probabilities gives $E J_{ijk\ell}(\theta^*)^2$, which results in the expression above.  We omit the details of these calculations.

Let us consider the simple situation that all $\theta_i^*$ are equal. In this case, if we apply the most efficient weighting function $f(A_l) \propto \sum_{u\in A_l} e^{\theta_u^*}$, that is, $f(\{i,j,k\}) = 3, f(\{i,j\}) = f(\{i,k\}) = 2$, we have $\mbox{Var}(J_i^*| \tilde {\mathcal G}) = 18/(7L)$. However, if we naively choose $f$ as a constant function, we get $\mbox{Var}(J_i^*| \tilde {\mathcal G}) = 8/(3L)$, which is indeed larger. It is also worth noting that when we choose $f(A_l) \propto \sum_{u\in A_l} e^{\theta_u^*}$, the aforementioned variance matches with the variance of MLE in \cite{fan2022ranking}.

With the above formula of $\mbox{Var}(J_i^*| \tilde {\mathcal G})$ for the PL model, we can also conclude that
\begin{align*}
   \mbox{Var}(J_i^*| \tilde {\mathcal G})^{-1/2} (\tilde\theta_i-\theta_i^*)\Rightarrow N(0,1),
\end{align*}
for all $i\in[n]$. 
The rigorous arguments will be introduced in Section \ref{sec:ranking_inference}.

\subsection{Ranking inference: one-sample confidence intervals} 
\label{sec:onesam_twosided_inf}
In numerous practical applications, individuals frequently interact with data and challenges related to rankings. The prevalent approach to utilizing rankings typically revolves around computing preference scores and then showcasing these scores in ranked order. These only provide first-order information on ranks and can not answer many questions, such as: 
\begin{itemize}
\item How do we ascertain with high confidence that an item's true rank is among the top-$3$ (or general $K,K\ge 1$) choices? And how can we establish a set of candidates with high confidence, guaranteeing that the true top-$3$ candidates are not overlooked?
\item How do we analyze if the ranking preferences for a given array of products are consistent in two distinct communities (such as male and female) or the same community but at two different time periods?
\end{itemize}

In sum, there is a need for tools and methodologies that address these and other insightful queries in real-world applications involving rankings, especially when the comparisons are drawn from a general comparison graph.

Within this section, we first present a comprehensive framework designed for the construction of two-sided confidence intervals for ranks. In endeavoring to establish simultaneous confidence intervals for the ranks, an intuitive methodology entails deducing the asymptotic distribution of the empirical ranks, denoted as ${\tilde{r}_{m}},{m \in \mathcal{M}}$, and subsequently determining the critical value. However, it is well known that this task poses substantive challenges, given that $\tilde{r}_{m}$ is an integer and is intrinsically dependent on all estimated scores, making its asymptotic behavior daunting to analyze. 

By capitalizing on the inherent interdependence between the scores and their corresponding ranks, we discern that the task of formulating confidence intervals for the ranks can be effectively converted to the construction of simultaneous confidence intervals for the pairwise differences among the population scores. It is notable that the distribution of these empirical score differences is more amenable to characterization. Consequently, we focus on the statistical properties of the estimated scores ${\tilde{\theta}_{m}},{m \in [n]}$ and present our methodology for constructing two-sided (simultaneous) confidence intervals for ranks through estimated score differences.

\begin{example}
\label{Example_confidence_interval}
We let $\mathcal{M} = \{m\}$, where $1\leq m\leq n$, to represent the item under consideration. We are interested in the construction of the $(1 - \alpha) \times 100\%$ confidence interval for the true population rank $r_{m}$, where $\alpha \in (0, 1)$ denotes a pre-specified significance level. Suppose that we are able to construct the simultaneous confidence intervals  ${[\mathcal{C}_{L}(k, m), \mathcal{C}_{U}(k, m)]},{k \neq m},(k\in [n])$ for the pairwise differences ${\theta_{k}^{*} - \theta_{m}^{*}},{k \neq m} (k\in [n])$, with the following property:
\begin{align}
\label{eq_SCIs_theta}
\PP\Big({\mathcal{C}_{L}(k, m) \leq \theta_{k}^{*} - \theta_{m}^{*} \leq \mathcal{C}_{U}(k, m) \mbox{ for all } k\neq m}\Big) \geq 1 - \alpha.
\end{align}
One observes that if $\mathcal{C}_{U}(k, m) < 0$ (respectively, $\mathcal{C}_{L}(k, m) > 0$), it implies that $\theta_{k}^{*} < \theta_{m}^{*}$ (respectively, $\theta_{k}^{*} > \theta_{m}^{*}$). Enumerating the number of items whose scores are higher than item $m$ gives a lower bound for rank $r_m$, and vice versa. 
In other words, we deduce from \eqref{eq_SCIs_theta} that
\begin{align}
\label{eq_confidence_interval_m_probability}
    \PP\left(1 + \sum_{k \neq m} 1\{\cC_{L}(k, m) > 0\} \leq r_{m} \leq n - \sum_{k \neq m} 1\{\cC_{U}(k, m) < 0\}\right) \geq 1 - \alpha. 
\end{align}
This yields a $(1 - \alpha) \times 100\%$ confidence interval for $r_{m}$, and our task reduces to construct simultaneous confidence intervals for the pairwise differences \eqref{eq_SCIs_theta}.
\qed
\end{example}

We now formally introduce the procedure to construct the confidence intervals for multiple ranks $\{r_{m}\}_{m \in \cM}$ simultaneously. Motivated by Example~\ref{Example_confidence_interval}, the key step is to construct the simultaneous confidence intervals for the pairwise score differences $\{\theta_{k}^{*} - \theta_{m}^{*}\}_{m \in \cM, k \neq m}$ such that~\eqref{eq_SCIs_theta} holds. To this end, we let
\begin{align}
\label{Statistics_max}
    T_{\mathcal{M}} = \max_{m \in \mathcal{M}} \max_{k \neq m} \left|\frac{\tilde{\theta}_{k} - \tilde{\theta}_{m} - (\theta_{k}^{*} - \theta_{m}^{*})}{\tilde{\sigma}_{km}}\right|, 
\end{align}
where $\{\tilde{\sigma}_{km}\}_{k \neq m}$ is a sequence of positive normalization given by \eqref{eq_sigma_hat_km} below. For any $\alpha \in (0, 1)$, let $Q_{1 - \alpha}$ be critical value such that $\PP(T_{\mathcal{M}} \leq Q_{1 - \alpha}) \geq 1 - \alpha$. Then, as in Example~\ref{Example_confidence_interval}, our $(1 - \alpha)\times 100\%$ simultaneous confidence intervals for $\{r_{m}\}_{m \in \mathcal{M}}$ are given by $\{[R_{mL}, R_{mU}]\}_{m \in \mathcal{M}}$, where 
\begin{align*}
    R_{mL} = 1 + \sum_{k \neq m} 1\left(\tilde{\theta}_{k} - \tilde{\theta}_{m} > \tilde{\sigma}_{km} \times Q_{1 - \alpha}\right),\quad
    R_{mU} = n - \sum_{k \neq m} 1\left(\tilde{\theta}_{k} - \tilde{\theta}_{m} < -\tilde{\sigma}_{km} \times Q_{1 - \alpha}\right).  
\end{align*}

\subsection{Multiplier bootstrap procedure}
\label{sec:bootstrap}
The key step for constructing the confidence interval of ranks of interest is to pick the critical value $Q_{1 - \alpha}$. To calculate the critical value above, we propose to use the wild bootstrap procedure. The uncertainty quantification for the spectral estimator in~\eqref{approx2} reveals that $\tilde{\theta}_{i} - \theta_{i}^{*} \approx J_{i}(\theta^{*})$ uniformly over $i \in [n]$ (see details in Section \ref{sec:ranking_inference}), which further implies that asymptotically 
\begin{align}
\label{eq_T_M_approximation}
    T_{\mathcal{M}} \approx \max_{m \in \mathcal{M}} \max_{k \neq m} \left|\frac{J_{k}(\theta^{*}) - J_{m}(\theta^{*})}{\tilde{\sigma}_{km}}\right|.
\end{align}
We focus on the fixed graph setting and leave the random graph setting in Remark \ref{Remark_ranking_inference_encoding} below. Practically, the empirical version of $J_{i}(\theta^{*})$ can be obtained via plugging in the spectral estimator $\tilde{\theta}$, namely from \eqref{eq:J_decomp_fixed},
\begin{align*}
    J_{i}(\tilde{\theta}) = \frac{1}{d} \sum_{l \in \mathcal{D}} J_{il}(\tilde{\theta}), \enspace i \in [n]. 
\end{align*}
Let $\sigma_{km}^{2} = \Var\{J_{k}(\theta^{*}) - J_{m}(\theta^{*})|\mathcal{G}\}$ for each $k \neq m$. Then our estimator for $\sigma_{km}^{2}$ is defined by
\begin{align}
\label{eq_sigma_hat_km}
    \tilde{\sigma}_{km}^{2} = \frac{e^{\tilde{\theta}_{k}} }{d^2 \tau_{k}^2(\tilde\theta)} \sum_{l \in \mathcal{D}} \frac{1(k \in A_{l})}{f^{2}(A_{l})}\left(\sum_{j \in A_{l}} e^{\tilde{\theta}_{j}} - e^{\tilde{\theta}_{k}}\right) + \frac{e^{\tilde{\theta}_{m}} }{d^2 \tau_{m}^2(\tilde\theta)} \sum_{l \in \mathcal{D}} \frac{1(m \in A_{l})}{f^{2}(A_{l})}\left(\sum_{j \in A_{l}} e^{\tilde{\theta}_{j}} - e^{\tilde{\theta}_{m}}\right)\,,
\end{align}
where $\tau_{k}(\tilde\theta)$ and $\tau_{m}(\tilde\theta)$ also plug in $\tilde\theta$; see \eqref{eq:fan1}. Let $\omega_{1}, \ldots, \omega_{|\mathcal{D}|} \in \mathbb{R}$ be i.i.d.~$N(0, 1)$ random variables. The Gaussian multiplier bootstrap statistic is then defined by 
\begin{align}
    \label{bootstrap}
    G_{\mathcal{M}} = \max_{m \in \mathcal{M}} \max_{k \neq m} \left|\frac{1}{d \tilde{\sigma}_{km}} \sum_{l \in \mathcal{D}}\{J_{kl}(\tilde{\theta}) - J_{ml}(\tilde{\theta})\} \omega_{l}\right|. 
\end{align}
Let $\PP^{*}(\cdot) = \PP(\cdot | \{(c_{l}, A_{l})\}_{l \in \mathcal{D}})$ denote the conditional probability. Then, for $\alpha \in (0, 1)$, our estimator for $Q_{1 - \alpha}$ is defined by the $(1 - \alpha)$th conditional quantile of $G_{\mathcal{M}}$, namely
\begin{align*}
    \mathcal{Q}_{1 - \alpha} = \inf\{z : \PP^{*}(G_{\mathcal{M}} \leq z) \geq 1 - \alpha\},
\end{align*}
which can be computed by the Monte Carlo simulation.
Then, our simultaneous confidence intervals $\{[\mathcal{R}_{mL}, \mathcal{R}_{mU}]\}_{m \in \mathcal{M}}$ are given by 
\begin{align}
\label{eq_rank_SCIs_feasible}
    \mathcal{R}_{mL} = 1 + \sum_{k \neq m} 1\left(\tilde{\theta}_{k} - \tilde{\theta}_{m} > \tilde{\sigma}_{km} \times \mathcal{Q}_{1 - \alpha}\right),\qquad
    \mathcal{R}_{mU} = n - \sum_{k \neq m} 1\left(\tilde{\theta}_{k} - \tilde{\theta}_{m} < -\tilde{\sigma}_{km} \times \mathcal{Q}_{1 - \alpha}\right).  
\end{align}

\begin{remark} [One-sample one-sided confidence intervals]
Now we provide details on constructing simultaneous one-sided intervals for population ranks. 
For one-sided intervals, the overall procedure is similar to constructing two-sided confidence intervals. Specifically, let
\begin{align}
\label{one_side_statistic}
    G_{\mathcal{M}}^{\circ} = \max_{m \in \mathcal{M}} \max_{k \neq m} \frac{1}{d \tilde{\sigma}_{km}} \sum_{l \in \mathcal{D}} \{J_{kl}(\tilde{\theta}) - J_{ml}(\tilde{\theta})\} \omega_{l},
\end{align}
where  $\omega_{1}, \ldots, \omega_{|\mathcal{D}|}$ are as before i.i.d.~$N(0, 1)$ random variables. Correspondingly, let $\mathcal{Q}_{1 - \alpha}^{\circ}$ be its $(1 - \alpha)$th quantile. Then the $(1 - \alpha)\times 100\%$ simultaneous lower confidence bounds for $\{r_{m}\}_{m \in \cM}$ are given by $\{[\mathcal{R}_{mL}^{\circ}, n]\}_{m \in \mathcal{M}}$, where 
\begin{align}
\label{eq_one_sided_confidence_interval}
    \mathcal{R}_{mL}^{\circ} = 1 + \sum_{k \neq m} 1\left(\tilde{\theta}_{k} - \tilde{\theta}_{m} > \tilde{\sigma}_{km} \times \mathcal{Q}_{1 - \alpha}^{\circ}\right). 
\end{align}
\end{remark}

\begin{remark}[Ranking inference for the PL model with random comparison graph] 
\label{Remark_ranking_inference_encoding}
Section \ref{Section_uncertainty_quantification_random} reveals that $\tilde{\theta}_{i} - \theta_{i}^{*} \approx J_{i}(\theta^{*})$ uniformly over $i \in [n]$, where following \eqref{eq:J_decomp_random},
\begin{align*}
    J_{i}(\theta^{*}) = \frac{1}{d}\sum_{\ell = 1}^{L} \sum_{j < s: j,s\ne i} J_{ijs\ell} (\theta^{*}). 
\end{align*}
In order to carry out ranking inference for the PL model, we need to rewrite this equation in a slightly different format. 
Let $\mathcal{N} = \sum_{i < j < k} \tilde{A}_{ijk}$ denote the total number of connected components on the random graph $\tilde{\mathcal{G}}$ and write $\{(i, j, k): i < j < k \mbox{ and } \tilde{A}_{ijk} = 1\} =: \{\tilde{A}_{q}\}_{q = 1, \ldots, \mathcal{N}}$. Let $y_{q}^{(\ell)}$ denote the $\ell$-th full-ranking comparison result for $\tilde{A}_{q}$. Then we can rewrite $P_{ij}$ as
    \begin{align*}
        P_{ij} = \frac{1}{d} \sum_{\ell = 1}^{L} \sum_{q = 1}^{\mathcal{N}} \sum_{k : k \neq i, j} 1\{(i, j, k) = \tilde{A}_{q}\} Z_{ijkq}^{(\ell)}, \enspace i \neq j, 
    \end{align*}
where for $i\neq j\neq k$ and $q \in [\mathcal{N}]$, and $Z_{ijkq}^{(\ell)} = {1\{y_{q}^{(\ell)} = (k \succ j \succ i)\}} / {f(\{i, j\})} + (1\{y_{q}^{(\ell)} = (j \succ i \succ k)\} + 1\{y_{q}^{(\ell)} = (j \succ k \succ i)\}) / {f(\{i, j, k\})}$.
It is straightforward to verify that this $P_{ij}$ is exactly the same with~\eqref{p_ij}. Therefore, we rewrite $J_{i}(\theta^{*}) = d^{-1} \sum_{\ell = 1}^{L} \sum_{q = 1}^{\mathcal{N}} J_{iq\ell}^{\diamond} (\theta^{*})$, where 
\begin{align}
\label{eq_J_diamond}
    J_{iq\ell}^{\diamond} (\theta^{*}) = \sum_{j < s: j,s \neq i} 1\{(i, j, s) = \tilde{A}_{q}\} 
    J_{ijs\ell} (\theta^{*})\,.
\end{align}
As is assumed, $\{J_{iq\ell}^{\diamond}(\theta^{*})\}_{\ell \in [L], q \in [\mathcal{N}]}$ are independent for each $i \in [n]$ conditioning on the comparison graph $\tilde{\mathcal{G}}$. Let $\{\omega_{q \ell}\}_{q, \ell \in \mathbb{N}}$ be i.i.d.~$N(0, 1)$ random variables. Then, following~\eqref{bootstrap}, the corresponding bootstrap test statistic is given by 
    \begin{align*}
        G_{\mathcal{M}}^{\diamond} = \max_{m \in \mathcal{M}} \max_{k \neq m} \left|\frac{1}{d \tilde{\sigma}_{km}^{\diamond}} \sum_{\ell = 1}^{L} \sum_{q = 1}^{\mathcal{N}} \{J_{kq \ell}^{\diamond}(\tilde{\theta}) - J_{mq \ell}^{\diamond}(\tilde{\theta})\} \omega_{q \ell}\right|,   
    \end{align*}
    where $\{\tilde{\sigma}_{km}^{\diamond}\}_{k\neq m}$ are as before the sequence of positive normalization, calculated as the sum of the variance of $J_k(\theta^*)$ and $J_m(\theta^*)$ similar to \eqref{eq_sigma_hat_km}. Consequently, the simultaneous confidence intervals for the ranks can be similarly constructed.

\end{remark}

\subsection{Ranking inference: two-sample and one-sample testing applications}
\label{Section_Application_Rank_inference}
In this section, we further illustrate how we may apply our inference methodology to a few salient testing applications, in both one-sample and two-sample testing. 

\begin{example}[Testing top-$K$ placement]
\label{example1}
Let $\cM = \{m\}$ for some $m \in [n]$ and let $K \geq 1$ be a prescribed positive integer. Our objective is to ascertain if the item $m$ is a member of the top-$K$ ranked items. Consequently, we shall examine the following hypotheses:
\begin{align}
\label{eq_top_K_test}
    H_{0} : r_{m} \leq K \enspace \mathrm{versus} \enspace H_{1} : r_{m} > K. 
\end{align}
Based on the one-sided confidence interval $[\cR_{mL}^{\circ}, n]$ in~\eqref{eq_one_sided_confidence_interval}, for any $\alpha \in (0, 1)$, a level~$\alpha$ test for~\eqref{eq_top_K_test} is simply given by $\phi_{m, K} = 1\{\cR_{mL}^{\circ} > K\}$. {Under the conditions of Theorem~\ref{Theorem_Gaussian_Approximation_T}, we have $\PP(\phi_{m, K} = 1|H_{0}) \leq \alpha$ + o(1), that is, the effective control of the Type-I error can be achieved below the significant level $\alpha$ when the null hypothesis is true.}
\end{example}

\begin{example}[Top-$K$ sure screening set]
\label{exp:candidate_ad}
Another example is on constructing a screened candidate set that contains the top-$K$ items with high probability. This is particularly useful in college candidate admission or company hiring decisions. Oftentimes, a university or a company would like to design certain admission or hiring policy with the high-probability guarantee of the sure screening of true top-$K$ candidates. 

Let $\cK = \{r^{-1}(1), \ldots, r^{-1}(K)\}$ denote the top-$K$ ranked items of the rank operator $r : [n] \to [n]$. We aim at selecting a set of candidates $\hat{\cI}_{K}$ which contains the top-$K$ candidates with a prescribed probability. Mathematically, this requirement can be expressed as $\PP(\mathcal{K} \subseteq \hat{\mathcal{I}}_{K}) \geq 1 - \alpha$, where $\alpha \in (0, 1)$. Herein, we define $\mathcal{M} = [n]$, and let $\{[\mathcal{R}_{mL}^{\circ}, n],{m \in [n]}\}$ represent the set of $(1 - \alpha) \times 100\%$ simultaneous left-sided confidence intervals, as given in \eqref{eq_one_sided_confidence_interval}. It is easy to observe that the inequality $\mathcal{R}_{mL}^{\circ} > K$ infers that $r_{m} > K$. Consequently, a selection for $\hat{\mathcal{I}}_{K}$, that satisfies the probability constraint $\PP(\mathcal{K} \subseteq \hat{\mathcal{I}}_{K}) \geq 1 - \alpha$, is given by
\begin{align*}
\hat{\mathcal{I}}_{K} = \{m \in [n] : \mathcal{R}_{mL}^{\circ} \leq K\}.
\end{align*}
\end{example}

\begin{example}[Testing ranks of two samples]
\label{Example_two_sample_item}
In many applications, we are concerned with the question of whether the ranks of certain items using two samples have been changed or preserved. For example, we may care about whether
\begin{itemize}
    \item Ranking allocation differs before and after a treatment or policy change. 
    \item Different communities such as males versus females can have different ranking preferences over the same set of products. 
    \item People's perceived preferences over the same things have changed in two time periods.
\end{itemize}

Suppose we observe two independent datasets $\mathcal{D}_{1}$ and $\mathcal{D}_{2}$ with preference scores $\theta_{[1]}^{*} = (\theta_{11}^{*}, \ldots, \theta_{1n}^{*})^{\top}$ and $\theta_{[2]}^{*} = (\theta_{21}^{*}, \ldots, \theta_{2n}^{*})^{\top}$. The associated true rankings are respectively denoted by 
\begin{align*}
    r_{[1]} = (r_{11}, \ldots, r_{1n})^{\top} \enspace \mathrm{and} \enspace r_{[2]} = (r_{21}, \ldots, r_{2n})^{\top}. 
\end{align*}
    Given any $m \in [n]$, we are interested in testing whether the same rank is preserved for item $m$ across these two samples, that is, testing the hypotheses 
    \begin{align}
    \label{eq_two_sample_rank_test}
        H_{0} : r_{1m} = r_{2m} \enspace \mathrm{versus} \enspace H_{1} : r_{1m} \neq r_{2m}
    \end{align}
    To this end, firstly we construct simultaneous confidence intervals $[R_{1mL}, R_{1mU}]$ and $[R_{2mL}, R_{2mU}]$ such that 
    \begin{align}
    \label{eq_two_sample_rank_SCI}
        \PP(r_{1m} \in [R_{1mL}, R_{1mU}] \mbox{ and } r_{2m} \in [R_{2mL}, R_{2mU}]) \geq 1 - \alpha. 
    \end{align}
    Then our $\alpha$-level test for~\eqref{eq_two_sample_rank_test} is defined by 
    \begin{align*}
        \phi_{m} = 1\{|[R_{1mL}, R_{1mU}] \cap [R_{2mL}, R_{2mU}]|= 0\}. 
    \end{align*}
    It is straightforward to verify that $\PP(\phi_{m} = 1|H_{0}) \geq 1 - \alpha$. 
\end{example}

\begin{example}[Testing top-$K$ sets of two samples] 
\label{Example_two_sample_set}
Besides testing for a single or a few ranks, one may want to evaluate whether two top-$K$ sets are identical or not, between two groups of people, two periods of time, or before and after a significant event or change. Let $\mathcal{S}_{1K} = \{r_{[1]}^{-1}(1), \ldots, r_{[1]}^{-1}(K)\}$ and $\mathcal{S}_{2K} = \{r_{[2]}^{-1}(1), \ldots, r_{[2]}^{-1}(K)\}$ denote the sets of top-\emph{K} ranked items, respectively. We consider testing the hypotheses
    \begin{align}
    \label{eq_two_sample_top_rank_test}
        H_{0} : \mathcal{S}_{1K} = \mathcal{S}_{2K} \enspace \mathrm{versus} \enspace H_{1} : \mathcal{S}_{1K} \neq \mathcal{S}_{2K}. 
    \end{align}
    For $\alpha \in (0, 1)$, we begin with constructing $(1 - \alpha)\times 100\%$ simultaneous confidence sets $\hat{\mathcal{I}}_{1K}$ and $\hat{\mathcal{I}}_{2K}$ for $\mathcal{S}_{1K}$ and $\mathcal{S}_{2K}$ such that 
    \begin{align}\label{simultaneous_set}
        \mathbb{P}\left(\mathcal{S}_{1K} \subset \hat{\mathcal{I}}_{1K} \mbox{ and } \mathcal{S}_{2K} \subset \hat{\mathcal{I}}_{2K}\right) \geq 1 - \alpha. 
    \end{align}
    Then our $\alpha$-level test for~\eqref{eq_two_sample_top_rank_test} is defined by 
    \begin{align*}
        \tilde{\phi}_{K} = 1\{|\hat{\mathcal{I}}_{1K} \cap \hat{\mathcal{I}}_{2K}| < K\}. 
    \end{align*}
\end{example}
\begin{remark}
Several methodologies, including the Bonferroni adjustment (which constructs a $(1-\alpha/2)\times 100\%$ confidence interval for each source), and Gaussian approximation (achieved by taking the maximum of the test statistics of each source), enable us to establish simultaneous confidence intervals as illustrated in equations \eqref{eq_two_sample_rank_SCI} and \eqref{simultaneous_set}. To maintain clarity and simplicity in the subsequent context, we simply employ the Bonferroni adjustment for two samples. Moreover, the framework outlined in Examples \ref{Example_two_sample_item} and \ref{Example_two_sample_set} can be extended in a straightforward way to evaluate whether the ranks of items or sets are identical across three or more sources. 
\end{remark}


\section{Theoretical Justifications}
\label{sec:ranking_inference}

In this section, we rigorously justify the conclusions in Section \ref{sec:ranking_estimation} and explicitly lay out the necessary assumptions to arrive at those conclusions. 
The first assumption is to make sure we are comparing $\theta_i^{*}$'s in the same order in a meaningful way. Otherwise, we can always group items into categories with similar qualities and then work on each sub-group or screen some extreme items. In addition, as we have discussed, we need an identifiability condition for $\theta^*$. 

\begin{assumption}
\label{Assumption_dynamic_range_bound}
There exists some positive constant $\bar{\kappa} < \infty$ such that 
\begin{align*}
    \max_{i \in [n]} \theta_{i}^{*} - \min_{i \in [n]} \theta_{i}^{*} \leq \bar{\kappa}. 
\end{align*}
In addition, for identifiability, assume $1^\top \theta^* = 0$.
\end{assumption}

{In Assumption \ref{Assumption_dynamic_range_bound}, we assume $\bar{\kappa}$ is finite, indicating we only rank items with preference scores on the same scale. If $\bar{\kappa}$ is diverging, some items will be trivially more or less favorable than others. In this case, it is typically easy in practice to separate the items into subgroups with similar preference scores, and then we can conduct ranking inference within each group. Although we assume bounded $\bar{\kappa}$, it serves as the role of a condition number whose effect has been made explicit in all our results for interested readers. However, we do not claim this dependency is optimal as our nontrivial analysis can easily encounter powers of $e^{\bar\kappa}$, say in bounding the ratio of $\pi_i^* / \pi_j^*$.}

\subsection{Estimation accuracy and asymptotic normality with fixed comparisons} 
\label{sec4.1}
To derive the asymptotic distribution of the spectral estimator, we need to rigorously justify the approximations \eqref{approx1} and \eqref{approx2}. We first take care of {approximating \eqref{approx1} using} \eqref{approx2}, where all comparisons in $\mathcal G$ are assumed to be fixed.  Note that
\begin{align*}
    P_{ij} - E[P_{ij}|\mathcal{G}] = \frac{1}{d} \sum_{l \in \mathcal{D}} 1(i, j \in A_{l}) \left[1(c_{l} = j) - \frac{\pi_{j}^{*}}{\sum_{u \in A_{l}} \pi_{u}^{*}}\right] \frac{1}{f(A_{l})}.  
\end{align*}
Let $Z_{A_l}^j = 1(c_l=j) / f(A_l)$, which is bounded from above and below as long as $f$ is bounded from above and below. Furthermore, each $Z_{A_l}^j$ is independent. Therefore,
$ P_{ij} - E[P_{ij}|\mathcal{G}] = d^{-1} \sum_{l \in \mathcal{D}} 1(i, j \in A_{l}) [Z_{A_l}^j - E(Z_{A_l}^j)]$. 
By Hoeffding's inequality, conditioning on $\mathcal G$, we have with a large probability $1 - o(1)$,
\begin{align*}
    \max_{i\neq j} \bigg|P_{ij} - E[P_{ij} | \mathcal G] \bigg| \lesssim \frac1d \sqrt{(\log n) n^\ddagger}.
\end{align*}
where
    $n^\ddagger = \max_{i\neq j} \sum_{l \in \mathcal D} 1(i, j \in A_{l})$ 
is the maximum number of cases that each pair is compared.  
Similarly, we can get the concentration bound for $\sum_{j:j\ne i} P_{ij}$. Since $Z_{A_l}^j$'s are independent, another level of summation over $j$ will lead to the following. Again by Hoeffding's inequality, with a large probability tending to $1$, we obtain
\begin{align*}
    \max_{i} \bigg|\sum_{j:j\ne i} P_{ij} - \sum_{j:j\ne i} E[P_{ij} | \mathcal G] \bigg| \lesssim \frac1d \sqrt{(\log n) n^\dagger }, 
\end{align*}
where $n^\dagger = \max_{i} \sum_{l \in \mathcal D} 1(i \in A_{l})$ is the maximum number of cases that each item is compared.
In addition, we assume 
\[
\sum_{j:j\ne i} E[P_{ij} | \mathcal G] = \tau_i e^{-\theta_i^*} \asymp \frac{1}{d} n^\dagger\,,
\]
{where $\tau_i$ defined in \eqref{definition_tau} is the denominator of $J_i^*$. This assumption makes sense as $\tau_i e^{-\theta_i^*} \lesssim \sum_{l \in \mathcal D} 1(i \in A_{l}) / d$, and it states for each $i$ the comparison graph cannot be too asymmetric.}
Note that $\sum_{l \in \mathcal D} 1(i, j \in A_{l})$ can still be widely different from $n^\ddagger$ for different pair $(i,j)$. 
Since the expectation term dominates the deviation if $n^\dagger \gtrsim \log n$, it is not hard to show that in \eqref{approx2}, changing the denominator by its expectation will only cause a small order difference, which does not affect the asymptotic distribution. 

Based on the above discussion, we impose the following assumption. 

\begin{assumption} \label{ass4.1}
In the case of a fixed comparison graph, we assume the graph is connected,
{$\tau_i e^{-\theta_i^*} \asymp n^\dagger / d$ for all $i \in [n]$, $e^{2\bar\kappa} \log n = o(n)$}
and ${e^{3\bar\kappa}} n^\ddagger n^{1/2} (\log n)^{1/2} = o(n^\dagger)$. 
\end{assumption} 

The assumption is reasonable for a fixed comparison graph. If each pair $(i,j)$ must be compared at least once, then every $\sum_{l \in \mathcal D} 1(i, j \in A_{l}) \ge 1$. If they are all in the same order, then $\sum_{l \in \mathcal D} 1(i \in A_{l}) = \sum_{j: j\ne i} \sum_{l \in \mathcal D} 1(i, j \in A_{l})$ should be indeed in the order of $n^\ddagger n$. 
Assumption \ref{ass4.1} allows some pair $(i,j)$ to be never compared directly, so we need to leverage the information from comparing $i$ and $j$ to other items separately. Moreover, we also do not require $\sum_{l \in \mathcal D} 1(i, j \in A_{l})$ to be  in the same order for any $i,j:i\neq j$ since we only require the maximum pairwise degree $n^\ddagger$ to satisfy  Assumption \ref{ass4.1}. However, in the case of a fixed graph, we do not have the randomness from the graph, and the graph must be relatively dense to make sure we have enough information to rank every item. This condition will be relaxed to $n^\dagger \gtrsim n^{\ddagger}\log n$ when we have a homogeneous random comparison graph in Section \ref{sec:rand_graph}. 

We need another technical condition on the structure of the comparison graph. Define $\Omega = \{\Omega_{ij}\}_{i\le n, j\le n}$ where $\Omega_{ij} = - P_{ji}\pi_j^*$ for $i\ne j$ and $\Omega_{ii} = \sum_{j:j\ne i} P_{ij}\pi_i^*$. Note that as we derived above, $E[\Omega_{ii} | \mathcal G]$ is in the order of $n^\dagger / (dn)$. We hope to understand the order of its eigenvalues. Since $\Omega$ has the minimal eigenvalue equal to zero, with the corresponding eigenvector ${\bf 1}$, we only focus on the space orthogonal to ${\bf 1}$. Following the notation of \cite{gao2021uncertainty}, 
\[
\lambda_{\min, \bot}(A) = \min_{\|v\|=1,v^\top {\bf 1}=0} v^\top A v\,.
\]

\begin{assumption} \label{ass4.2} There exist $C_1, C_2 > 0$ such that
\begin{equation}
\label{bound_b_1}
 C_1 {e^{-\bar\kappa}} \frac{n^\dagger}{dn} \le \lambda_{\min, \bot}(E[\Omega|\mathcal G]) \le \lambda_{\max}(E[\Omega|\mathcal G]) \le C_2  {e^{\bar\kappa}} \frac{n^\dagger}{dn}\,,
\end{equation}
\begin{equation}
\label{bound_b_2}
\|\Omega - E[\Omega|\mathcal G]\| = o_P\bigg(\frac{n^\dagger}{dn}\bigg)\,.
\end{equation}
\end{assumption}

{When $\bar\kappa = O(1)$,} Assumption \ref{ass4.2} requires that all eigenvalues (except the minimal one) of $E[\Omega|\mathcal G]$ are in the order of ${n^\dagger}/(dn)$ and $\Omega$ also shares this same eigenvalue scale as $E[\Omega|\mathcal G]$. This assumption is intuitively correct, as we have seen that $E[\Omega_{ij}] \lesssim n^\ddagger / (dn)$ for $i\ne j$ and $E[\Omega_{ii}] \asymp n^\dagger / (dn)$. We will also rigorously show that this condition can be satisfied if we consider the PL model (Theorem \ref{thm_b_rand}). 

\begin{theorem}
\label{thm_approx}
Under Assumptions \ref{Assumption_dynamic_range_bound}-\ref{ass4.2}, the spectral estimator $\tilde\theta_i$ has the following uniform approximation:
$
    \tilde\theta_i-\theta_i^*=J_i^*+\delta_i,
$
uniformly for all $i\in[n], $where $\|\delta:=(\delta_1,\cdots,\delta_n)\|_{\infty}=o({1}/{\sqrt{n^\dagger}})$ with probability $1-o(1).$

\end{theorem}
To prove Theorem \ref{thm_approx}, we need to verify \eqref{approx1}. We leave the detailed proof in the appendix. Given Theorem \ref{thm_approx}, we can easily conclude the next theorem following the properties of $J_i^*$, which lead to the rate of convergence for $\tilde\theta$ as well as its asymptotic normality. 

\begin{remark}
{The results of Theorem \ref{thm_approx} and the following Theorems are proved via Bernstein and Hoeffding type inequalities with union bound over $n$ items. Therefore, all of the high-probability terms hold with probability $(1-o(1))$ (similarly for $o_p(\cdot)$ and $O_p(\cdot)$) mentioned in the main text equivalently hold with probability in form of $1-\cO(n^{-\zeta})$ where $\zeta\ge 2$ is a positive integer (different choice of $\zeta$ will only affect constant terms in the involved concentration inequalities). }
\end{remark}

\begin{theorem}
\label{thm_consist_normality}
Under Assumptions \ref{Assumption_dynamic_range_bound}-\ref{ass4.2}, the spectral estimator \eqref{eq:spectrla_estimator} satisfies that
\begin{align}
    \|\tilde\theta-\theta^{*}\|_{\infty} & \asymp \|J^*\|_\infty \lesssim {e^{\bar\kappa}}  \sqrt{\frac{\log n}{n^\dagger}},\label{l_infty_consist}
\end{align}
with probability $1-o(1)$, where $J^*=(J_1^*,\cdots,J_n^*)$ with $J_i^*,i\in[n]$ being defined in \eqref{eq:J_decomp_fixed}.
In addition, 
\begin{align*}
   \rho_i(\theta) (\tilde\theta_i-\theta_i^*) \Rightarrow N(0,1),
\end{align*}
for all $i\in[n]$ with
\begin{align}
\label{eq_rho_theta}
 \rho_i(\theta)&=\bigg[\sum\limits_{l \in \mathcal D} 1(i \in A_{l}) \bigg(\frac{\sum_{u\in A_l} e^{\theta_u} - e^{\theta_i}}{\sum_{u\in A_l} e^{\theta_u}}\bigg) \frac{e^{\theta_i}}{f(A_l)}\bigg] / \bigg[\sum\limits_{l \in \mathcal D} 1(i \in A_{l}) \bigg(\frac{\sum_{u\in A_l} e^{\theta_u} - e^{\theta_i}}{f(A_l)}\bigg) \frac{e^{\theta_i}}{f(A_l)}\bigg]^{1/2},
\end{align}
for both $\theta = \theta^*$ and $\theta = $ any consistent estimator of $\theta^*$.
\end{theorem}

Note that Theorem \ref{thm_consist_normality} indicates that 
{the choice of $f(\cdot)>0$ does not affect the rate of convergence, but it affects the estimation efficiency. As we argued in Section \ref{Section_uncertainty_quantification_fixed}, the optimal weighting to minimize the asymptotic variance is $f(A_l) \propto \sum_{u\in A_l} e^{\theta_u^*}$ in the class of spectral estimators. In practice, however, we do not know $\theta_u^*$ beforehand. Therefore, we could implement a two-step procedure to improve the efficiency of the spectral estimator: in the first step, we obtain our initial consistent estimator $\hat\theta_u^{(\text{initial})}$ with weighting say $f(A_l) = |A_l|$, and in the second step, we estimate $f(A_l) = \sum_{u\in A_l} e^{\theta_u^*}$ by plugging $\hat\theta_u^{(\text{initial})}$ and run the spectral method again with this optimal weighting to get the final asymptotically efficient estimator $\hat\theta_u^{(\text{final})}$. Note that we do not intend to prove the theoretical properties of this two-step estimator, as the data dependency in the optional weighting of the second step makes the uniform approximation analysis highly nontrivial due to non-i.i.d. ranking outcomes. 
Nonetheless, we could circumvent this theoretical difficulty by splitting data into a very small part ($o(|\cD|)$ samples) for step 1, to achieve consistency with a worse convergence rate, and using the remaining majority ($|\cD|-o(|\cD|)$ samples) for step 2, to maintain the same asymptotic behavior. In addition, empirically, we found that directly using the same whole data in both steps achieves decent performance given a large sample size. We refer interested readers to our numerical studies.
}

\subsection{Estimation accuracy and asymptotic normality for the PL model} \label{sec:rand_graph}

In the random graph case, we have to specify the graph generation process in order to study the theoretical properties. We consider the commonly used PL model, where we sample each $M$-way comparison with probability $p$ and compare this set for $L$ times. Furthermore, we will only work with $M=3$ since we plan to focus on a transparent and intuitive discussion. We can easily generalize all the discussions to general $M$, but derivations and formulas can be more tedious.

The PL model with $3$-way comparisons has been studied in~\citet{fan2022ranking} by using MLE, where they explicitly write down the likelihood function. The proposed spectral method can work for any fixed graph, including the one generated from the PL model. In this section, we would like to compare the performance of the spectral method with that of the MLE. To make sure the spectral method works for the PL model, we need to prove the approximations \eqref{approx1} and \eqref{approx3}. 

We first take care of \eqref{approx3}.
Consider conditioning on $\tilde {\mathcal G}$, where all comparisons in $\tilde {\mathcal G}$ are independent; each $\tilde A_{ijk}$ is compared for $L$ times if $\tilde A_{ijk} = 1$. 
Now $c_l$ and $A_l$ are induced from $\tilde{\mathcal G}$, and can be dependent. In this case, we can write 
\[
P_{ij} - E[P_{ij} | \tilde{\mathcal G}] = \frac{1}{d} \sum_{\ell=1}^L \sum\limits_{k: k \ne j,i} \tilde A_{ijk} [Z_{{ijk}}^l - EZ_{{ijk}}^l]\,,
\]
where $Z_{{ijk}}^l = 1(A_l=\{i,j\}, c_l=j) / f(\{i,j\}) + 1(A_l=\{i,j,k\}, c_l=j) / f(\{i,j,k\})$, which is again bounded from above and below and independent for any given $\tilde A_{ijk}$. In this case, with a little abuse of notations, we redefine 
\[
n^\ddagger = L \max_{i\neq j} \sum\limits_{k: k \ne j,i} \tilde A_{ijk} \,,\quad n^\dagger = L \max_{i} \sum\limits_{j < k: j,k \ne i} \tilde A_{ijk}\,.
\]
Similar to Section \ref{sec4.1}, conditional on $\tilde{\cG},$ we have
\begin{equation*}
\begin{aligned}
& \max_{i} \bigg|\sum_{j:j\ne i} P_{ij} - \sum_{j:j\ne i} E[P_{ij} |\tilde{ \mathcal G}] \bigg| = \cO_P (d^{-1} \sqrt{n^\dagger \log n})\,, \\
& \max_{i \neq j} \bigg|P_{ij} - E[P_{ij} | \tilde{ \mathcal G}] \bigg| = \cO_P(d^{-1}\sqrt{n^\ddagger \log n})\,, \\
& \sum_{j:j\ne i} E[P_{ij} |\tilde{ \mathcal G}] = \tau_i^{\diamond} e^{-\theta_i^*} \asymp \frac{1}{d} n^\dagger \;\; {\text{(assumption)} }\,.
\end{aligned}
\end{equation*}

We adapt Assumption \ref{ass4.1} to the following assumption. Note that we have no assumption on $L$, so $L$ can be as low as $1$.

\begin{assumption} \label{ass4.3}
In the PL model with $M$-way complete comparisons, choose $d \asymp n^\dagger$ in the spectral ranking, and assume 
{$\tau_i^{\diamond} e^{-\theta_i^*} \asymp n^\dagger / d$ for all $i \in [n]$, $e^{4\bar\kappa} = o(n)$} and $p \gtrsim { e^{6\bar\kappa}} \textrm{poly}(\log n) / \binom{n-1}{M-1}$.  
\end{assumption}

Under Assumption \ref{ass4.3}, we can prove
\begin{equation*}
n^\dagger \asymp \binom{n-1}{M-1} p L\,, \quad\quad
\max\bigg\{ \binom{n-2}{M-2} p -\log n, 0\bigg\}L\lesssim n^\ddagger \lesssim \bigg[ \binom{n-2}{M-2} p  + \log n \bigg] L\,.
\end{equation*}
with probability $1 - o(1)$.
Note that in $n^\dagger$, by Assumption \ref{ass4.3}, we know the dominating term is  $\binom{n-1}{M-1} p L$. However, in $n^\ddagger$, we have the additional term $\log n$, which comes from the sub-exponential tail decay in Bernstein inequality, and if $p$ is really small, it could happen that $\log n$ dominates $n^\ddagger$. When $p$ is large, that is $\binom{n-2}{M-2} p \gtrsim \log n$, then $n^\dagger \asymp n n^\ddagger$ and Assumption \ref{ass4.1} holds. Therefore, we have a dense comparison graph, and the proof for this part follows in a similar vein as Theorem \ref{thm_consist_normality}. When $p$ is small, that is, $\binom{n-2}{M-2} p \lesssim \log n$, $n^\ddagger \lesssim \log n$ if $L$ is bounded. In this case, we will modify the proof of Theorem \ref{thm_consist_normality} to the random graph case in order to show Theorem \ref{thm_consist_normality_rand} below.
In addition, since $\sum_{j:j\ne i} P_{ij} = \cO_P(n^\dagger / d)$, it makes sense to choose $d \asymp n^\dagger$ in Assumption \ref{ass4.3} to make the diagonal elements of the transition matrix a constant order. Note that in the fixed graph case, we do not need to impose rate assumptions on $d$ as the comparison graph has no randomness.

Next, we verify that under the PL model, Assumption \ref{ass4.2} holds with high probability.

\begin{theorem} \label{thm_b_rand}
Under the PL model and Assumption \ref{ass4.3}, with probability $1-o(1)$, Assumption \ref{ass4.2} holds when we condition on $\tilde{\mathcal G}$ instead of $\mathcal G$.
\end{theorem}
  We next hope to show that under Assumptions \ref{ass4.3}, the spectral estimator $\tilde\theta_i$ has the uniform approximation: the differences between $\tilde\theta_i - \theta_i^*$ and $J_i^*$ for all $i \in [n]$ are $o_P(1/\sqrt{n^\dagger})$. The key step is still the verification of \eqref{approx1} under this weaker Assumptions \ref{ass4.3} for a random comparison graph.

\begin{theorem}
\label{thm_consist_normality_rand}
Under the PL model and Assumptions \ref{Assumption_dynamic_range_bound} and  \ref{ass4.3}, the spectral estimator $\tilde\theta_i$ has the uniform approximation: $\tilde\theta_i-\theta_i^*=J_i^*+o_P\big({1}/{\sqrt{n^\dagger}}\big)$, uniformly for all $i\in[n]$. Therefore, the spectral estimator~\eqref{eq:spectrla_estimator} satisfies 
\begin{align}
    \|\tilde\theta-\theta^{*}\|_{\infty} \lesssim {e^{\bar\kappa}}  \sqrt{\frac{\log n}{\binom{n-1}{M-1}pL}},\label{l_infty_consist_rand}
\end{align}
with probability $1-o(1)$.
In addition, 
\begin{align*}
   \rho_i(\theta) (\tilde\theta_i-\theta_i^*)\Rightarrow N(0,1),
\end{align*}
for all $i\in[n]$ with 
$\rho_i(\theta) = \mbox{Var}(J_i^*\given \tilde{\mathcal{G}})^{-1/2}$, where in the formula of $\mbox{Var}(J_i^*\given \tilde{\mathcal{G}})$ we can choose both $\theta = \theta^*$ and $\theta = $ any consistent estimator of $\theta^*$.
\end{theorem}

\begin{remark}
    {The two-step estimator under optimal weight $f(A_l) = \sum_{u\in A_l} e^{\theta_u^*}$ (can be consistently estimated with a small proportion of a separate dataset) achieves the same variance as the MLE estimator, which matches the Cr{a}\'mer Rao lower bound among all estimators \citep{fan2022uncertainty,fan2022ranking}.}
\end{remark}

{\begin{corollary}\label{top-k-ranking} Under the conditions of Theorem \ref{thm_consist_normality_rand}, if we have $\theta^*_{(K)}-\theta^*_{(K+1)}\ge \Delta,$ with $\theta_{(i)}^*$ denoting the underlying score of the item with true rank $i$ for $i\in[n]$, and when the sample complexity satisfies 
		\begin{align*}
	 e^{2\bar\kappa}\Delta^{-2}\cdot \log n={\cO}\bigg(\binom{n-1}{M-1}pL \bigg),
		\end{align*}
		 we have $\{i\in [n],\hat r_i\le K\}=\{i\in [n],r_i^*\le K\}$ (the selected top-K set is identical to the true top-K set), where $\hat r_i,r_i^*$ denote the empirical rank of $\hat\theta_i$ among $\{\hat\theta_i,i\in[n]\}$ and true rank of the i-th item, respectively.
\end{corollary}

We remark that when $M=2,$ and $\bar\kappa=\cO(1),$ our conclusion from Corollary \ref{top-k-ranking} reduces to the conclusion of Theorem 1 in \cite{chen2019spectral}.}

\subsection{Validity Justification for Bootstrap Procedure} \label{sec4.3}
The primary goal of this section is to justify the validity of the proposed bootstrap procedure in Section~\ref{sec:bootstrap}. Recall that the targeted quantity $T_{\mathcal{M}}$ is the maximum modulus of the random vector 
\begin{align*}
    \Delta_{\mathcal{M}} := \left\{\frac{\tilde{\theta}_{k} - \tilde{\theta}_{m} - (\theta_{k}^{*} - \theta_{m}^{*})}{\tilde{\sigma}_{km}}\right\}_{m \in \mathcal{M}, k\neq m}. 
\end{align*}
For each marginal of $\Delta_{\mathcal{M}}$, the asymptotic normality can be similarly established following Theorem~\ref{thm_consist_normality}. However, studying the asymptotic distribution of $\|\Delta_{\mathcal{M}}\|_{\infty}$ becomes quite challenging as its dimension $(n - 1)|\mathcal{M}|$ can increase with the number of items. In particular, the traditional multivariate central limit theorem for $\Delta_{\mathcal{M}}$ may no longer be valid asymptotically~\citep{Portnoy1986}. To handle the high dimensionality, we shall invoke the modern Gaussian approximation theory~\citep{CCK2017, Chernozhukov2019} in order to derive the asymptotic distribution of $\|\Delta_{\mathcal{M}}\|_{\infty}$, shown in Theorem~\ref{Theorem_Gaussian_Approximation_T}. Moreover, the validity of our multiplier bootstrap procedure is justified in the following theorem.  


\begin{theorem}
\label{Theorem_Bootstrap_Validity}
Assume {$e^{3\bar{\kappa}} (\log n)^{2} = o(n)$ and $e^{5\bar{\kappa}}n^{\ddagger} n^{1/2} (\log n)^{3} = o(n^{\dagger})$}. Then, under the conditions of Theorem~\ref{thm_approx}, we have 
\begin{align*}
    |\PP(T_{\mathcal{M}} > \mathcal{Q}_{1 - \alpha}) - \alpha| \to 0. 
\end{align*}
\end{theorem}

\begin{remark}
Theorem~\ref{Theorem_Bootstrap_Validity} indicates that the estimated critical value $\mathcal{Q}_{1 - \alpha}$ from the Gaussian multiplier bootstrap indeed controls the significance level of the simultaneous confidence intervals~\eqref{eq_rank_SCIs_feasible} for $\{r_{m}\}_{m \in \mathcal{M}}$ to the prespecified level $\alpha$, that is, 
\begin{align*}
    \PP\Big(r_{m} \in [\mathcal{R}_{mU}, \mathcal{R}_{mR}] \mbox{ for all } m \in \mathcal{M} \Big) \geq 1 - \alpha + o(1). 
\end{align*}
Recently~\citet{fan2022ranking} proposed a similar approach to construct simultaneous confidence intervals for ranks in the context of the PL model with only the top choice observed for each comparison, which, however, requires the number of comparisons $L$ for each connected item to be sufficiently large such that $L \gtrsim \mathrm{poly}(\log n)$. In contrast, our procedure in Section~\ref{sec:bootstrap} works without any constraints on the number of comparisons for each $A_{l}$ (i.e., we even allow $L=1$ for all comparisons) and is thus much more widely applicable, since in many real problems, sets of size $M$ can be compared at different times and sometimes only once.  
\end{remark}

\section{Numerical Studies} 
\label{sec:numerical}

In this section, we validate the methodology and examine the theoretical results introduced in Sections \ref{sec:ranking_estimation} and \ref{sec:ranking_inference}. We conducted comprehensive simulation studies, but due to the page limit, we relegate the results to the Appendix \ref{sec:simu} and only briefly summarize our key findings here. We first validate consistency and asymptotic distribution of the spectral estimator, check the efficacy of the Gaussian bootstrap. All simulations match with our theoretical results perfectly. We also provide examples for constructing one-sample and two-sample confidence intervals for ranks and carrying out hypothesis testing Examples \ref{example1}-\ref{Example_two_sample_set}. Finally, we empirically investigate the connections between the spectral method and MLE. In particular, we found the two-step or oracle weight spectral method behaves almost identically to MLE.

\section{Real Data Analysis}\label{sec:realdata}
We present the spectral ranking inferences for two real datasets in this section. The first one is about the ranking of statistics journals, which is based on pairwise comparisons (Journal B citing Journal A means A is preferred over B, which is consistent with the fact that good papers usually have higher citations). In particular, we can test for two periods of time, whether the journal ranking has changed significantly. The second application is about movie ranking, which is based on multiway comparisons of different sizes (3 or 4 movies are given to people to rank). The movie comparison graph shows strong heterogeneity in node degrees. Therefore, this comparison graph should be better modeled as a fixed graph without homogeneous sampling. In both cases, we will report the results from the Two-Step spectral estimator and the vanilla spectral estimator in Appendix \ref{sec:add_real_data}.

\subsection{Ranking of Statistics Journals}
In this section, we study the Multi-Attribute Dataset on Statisticians (MADStat) which contains citation information from 83,331 papers published in 36 journals between 1975-2015. The data set was collected and studied by~\citet{ji2022co, ji2023meta}.  

We follow \cite{ji2023meta}'s convention to establish our pairwise comparison data. We will use journals' abbreviations given in the data. We refer interested readers to the complete journal names on the data website 
\url{https://dataverse.harvard.edu/dataset.xhtml?persistentId=doi:10.7910/DVN/V7VUFO}. Firstly, we excluded three probability journals—AOP, PTRF, and AIHPP, due to their fewer citation exchanges with other statistics publications. Hence, our study comprises of a total of 33 journals.  Secondly, we only examine papers published between 2006-2015. For instance, if we treat 2010 as our reference year, we only count comparison results indicating `Journal A is superior to Journal B' if, and only if, a paper published in Journal B in 2010 has cited another paper that was published in Journal A between the years 2001 and 2010. This approach favors more recent citations, thus allowing the journal rankings to better reflect the most current trends and information. Finally, we chose to divide our study into two periods, 2006-2010 and 2011-2015, to detect the possible rank changes of journals. 

Utilizing the data, we showcase the ranking inference results, summarized in Table \ref{tab_journal1}. These results include two-sided, one-sided, and uniform one-sided confidence intervals for ranks within each of the two time periods (2006–2010 and 2011–2015). We calculate these intervals using the Two-Step Spectral method over the fixed comparison graph (the results based on the one-step Vanilla Spectral method are presented in Appendix \ref{vanilla_spectral_real_data}) and using the bootstrap method detailed in Sections \ref{sec:onesam_twosided_inf} - \ref{Section_Application_Rank_inference}.

\begin{table}
    \centering
     	\def\arraystretch{0.95}
      \setlength\tabcolsep{3.0pt}
    \begin{tabular}{c|r r r r r r|r r r r r r} 
    \toprule
               \multicolumn{1}{c|}{} & \multicolumn{6}{c|}{$2006-2010$} & \multicolumn{6}{c}{$2011-2015$}\\  
       Journal & \multicolumn{1}{c}{$\tilde{\theta}$} & $\tilde{r}$ & TCI & OCI & UOCI & Count & \multicolumn{1}{c}{$\tilde{\theta}$} & $\tilde{r}$ & TCI & OCI & UOCI & Count\\ \hline
JRSSB &  $1.654$   &   $1$   &   $[1, 1]$  &  $[1, n]$   & $[1, n]$       &  $5282$   &    $1.553$  &  $1$  &  $[1, 2]$  &  $[1, n]$  &  $[1, n]$ & $5513$\\
AoS   &  $1.206$   &    $3$  &  $[2, 4]$  &  $[2, n]$ &   $[2, n]$       &   $7674$   &   $1.522$  &  $2$  &  $[1, 2]$  &  $[1, n]$  &  $[1, n]$& $11316$\\
Bka   &  $1.316$   &   $2$   &   $[2, 3]$     &   $[2, n]$  &   $[2, n]$  & $5579$  &   $1.202$  &  $3$   &  $[3, 3]$  &  $[3, n]$  &   $[3, n]$ & $6399$\\
JASA  &  $1.165$   &    $4$  &    $[3, 4]$  &  $[3, n]$  &  $[3, n]$     &   $9652$   &     $1.064$  &  $4$  &  $[4, 4]$  &  $[4, n]$  &  $[4, n]$& $10862$\\
JMLR  & $-0.053$   & $20$     &   $[14, 25]$  & $[15, n]$ &  $[13, n]$      & $1100$   &    $0.721$  &  $5$  &   $[5, 7]$  &  $[5, n]$   & $[5, n]$& $2551$\\
Biost &  $0.288$   &   $13$  &  $[10, 18]$  &   $[10, n]$  &   $[9, n]$  & $2175$  &     $0.591$ &   $6$  &   $[5, 9]$  &  $[5, n]$  &   $[5, n]$  & $2727$\\  
Bcs   &  $0.820$   &   $5$   &  $[5, 7]$    &   $[5, n]$   &   $[5, n]$  & $6614$  &       $0.571$ &  $7$  &   $[5, 9]$  &  $[6, n]$  &    $[5, n]$  & $6450$\\  
StSci &  $0.668$   &    $7$  & $[5, 9]$  &   $[5, n]$  &  $[5, n]$     &  $1796$   &   $0.437$  &  $8$  &  $[6, 13]$  &  $[6, n]$   & $[6, n]$& $2461$\\
Sini  &  $0.416$   &   $10$ &   $[9, 14]$   &    $[9, n]$  &   $[8, n]$     &  $3701$    &   $0.374$  &  $9$  &  $[8, 13]$  &  $[8, n]$  &  $[8, n]$& $4915$\\
JRSSA &  $0.239$   &    $14$ &  $[10, 20]$  &    $[10, n]$  &    $[9, n]$  &  $893$    &   $0.370$  & $10$ &   $[6, 13]$  &  $[8, n]$  &  $[6, n]$& $865$\\
JCGS  &  $0.605$   &    $8$ &   $[6, 9]$    &   $[6, n]$  &  $[6, n]$     &  $2493$    &     $0.338$  & $11$  &  $[8, 13]$  &  $[8, n]$  &  $[8, n]$& $3105$\\
Bern  &  $0.793$   &    $6$  &    $[5, 8]$   & $[5, n]$  &  $[5, n]$     &    $1575$   &     $0.336$ &   $12$ &   $[8, 13]$ &  $[8, n]$  &  $[8, n]$& $2613$\\
ScaJS &  $0.528$   &    $9$  &   $[7, 12]$  &    $[7, n]$   &   $[6, n]$     & $2442$ &   $0.258$  & $13$  &  $[8, 13]$  &   $[9, n]$   & $[8, n]$& $2573$\\
JRSSC &  $0.113$   &   $15$  & $[11, 22]$  &  $[11, n]$   &     $[11, n]$    &  $1401$   &     $0.020$ &  $14$  &   $[14, 19]$ &  $[14, n]$ &   $[12, n]$& $1492$\\
AoAS  &  $-1.463$  &   $30$  &   $[30, 33]$  &  $[30, n]$ &   $[30, n]$     &  $1258$   &     $-0.017$ &  $15$ &  $[14, 20]$ &  $[14, n]$ &  $[14, n]$& $3768$\\
CanJS &  $0.101$   &  $17$   &  $[11, 22]$ &  $[11, n]$  & $[11, n]$       &   $1694$   &   $-0.033$ &  $16$  & $[14, 20]$ &  $[14, n]$  & $[14, n]$& $1702$\\
JSPI  &  $-0.327$  & $26$    &     $[24, 26]$  & $[24, n]$  & $[22, n]$    &  $6505$    &    $-0.046$ &  $17$ &  $[14, 20]$ &  $[14, n]$  & $[14, n]$& $6732$\\
JTSA  &  $0.289$   &   $12$  &  $[9, 18]$   &  $[10, n]$  &  $[8, n]$     &   $751$   &     $-0.101$ &  $18$  & $[14, 22]$ &  $[14, n]$ &  $[14, n]$& $1026$\\
JMVA  &  $-0.126$  &  $22$   &   $[17, 25]$ &  $[17, n]$ &  $[15, n]$    &  $3833$    &   $-0.148$  & $19$ &  $[14, 22]$ &  $[15, n]$ &  $[14, n]$& $6454$\\
SMed  &  $-0.131$  &   $23$  &  $[17, 25]$  &   $[18, n]$  &   $[17, n]$ & $6626$ & $-0.242$  & $20$  &  $[18, 25]$  & $[18, n]$ &  $[17, n]$  & $6857$\\    
Extrem & $-2.099$  &   $33$  &  $[30, 33]$ &  $[31, n]$ &  $[30, n]$     & $173$    &    $-0.312$ &  $21$  &   $[16, 30]$ & $[18, n]$  & $[14, n]$& $487$\\
AISM  &  $0.317$   &  $11$   &     $[9, 18]$ &  $[10, n]$  &  $[9, n]$   &  $1313$     &   $-0.359$  & $22$ &  $[19, 30]$ &  $[20, n]$ &  $[18, n]$& $1605$\\
EJS   &  $-1.717$  &  $32$   &     $[30, 33]$ &  $[30, n]$ &  $[30, n]$    &  $1366$    &    $-0.367$ &  $23$ &  $[20, 29]$ &  $[20, n]$ &  $[19, n]$& $4112$\\
SPLet & $-0.033$   &   $19$  &  $[15, 25]$  &  $[15, n]$ &  $[13, n]$  &  $3651$     &   $-0.384$ &  $24$ &  $[21, 29]$ &  $[21, n]$ &  $[19, n]$& $4439$\\
CSDA  &  $-0.975$  &   $29$  &   $[27, 30]$ &  $[27, n]$ &  $[27, n]$     & $6732$    &    $-0.467$ &  $25$ &  $[21, 30]$ &  $[21, n]$ &  $[21, n]$& $8717$\\
JNS   &  $-0.255$  &  $25$   &   $[19, 26]$  & $[20, n]$ &  $[17, n]$     &  $1286$    &    $-0.484$ &  $26$ &  $[21, 30]$ & $[21, n]$  & $[21, n]$& $1895$\\
ISRe  &  $0.082$   &   $18$   &   $[11, 25]$ &  $[11, n]$ &  $[10, n]$    &   $511$   &     $-0.491$ &  $27$ &  $[21, 30]$ &  $[21, n]$ &  $[20, n]$& $905$\\
AuNZ  &  $0.108$   &   $16$  &   $[11, 23]$  & $[11, n]$  & $[10, n]$     &  $862$     &     $-0.504$ &  $28$ &  $[21, 30]$ &  $[21, n]$ &  $[20, n]$& $816$\\
JClas & $-0.185$   &   $24$  &    $[15, 26]$ &  $[15, n]$ &  $[11, n]$    &  $260$    &     $-0.535$ &  $29$ &  $[18, 30]$  & $[20, n]$  & $[14, n]$& $224$\\
SCmp  &  $-0.096$  &    $21$ &  $[15, 25]$   &  $[15, n]$ &  $[14, n]$    &  $1309$    &     $-0.561$ &   $30$   &   $[23, 30]$ &  $[24, n]$ &  $[21, n]$& $2650$\\
Bay   &  $-1.494$  &   $31$  &  $[30, 33]$ &    $[30, n]$  &  $[27, n]$   &  $279$      &    $-1.102$  & $31$ &  $[31, 32]$  & $[31, n]$  & $[30, n]$& $842$\\
CSTM  &  $-0.843$  &   $27$  &  $[27, 29]$  &   $[27, n]$  &   $[27, n]$  & $2975$ &  $-1.296$ &  $32$ &  $[31, 32]$ &  $[31, n]$  &  $[31, n]$   & $4057$\\      
JoAS  &  $-0.912$  &  $28$   &   $[27, 30]$    &  $[27, n]$ &   $[27, n]$ &  $1055$ &   $-1.904$ &  $33$  & $[33, 33]$  &  $[33, n]$  & $[33, n]$ & $2780$\\ 
    \bottomrule
    \end{tabular}
    \caption{Ranking inference results for 33 journals in 2006-2010 and 2011-2015 based on the Two-Step Spectral estimator. For each time period, there are 6 columns. The first column $\tilde{\theta}$ denotes the estimated underlying scores. The second through the fifth columns denote their relative ranks, two-sided, one-sided, and uniform one-sided confidence intervals for ranks with coverage level $95\%$, respectively. The sixth column denotes the number of comparisons in which each journal is involved.}
    \label{tab_journal1}
\end{table}

From Table \ref{tab_journal1}, we can easily get answers to the following questions on ranking inference. 
For example, is each journal's rank maintained unchanged across the two time periods? At a significance level of $\alpha = 10\%$, we find that the ranks of the following journals in alphabetical order demonstrate significant differences between the two time frames (Example \ref{Example_two_sample_item}):
\begin{align*}
\emph{AISM, AoAS, Biost, CSTM, EJS,  JMLR, JoAS, JSPI.}
\end{align*}
This aligns with real-world observations. For instance, JMLR, EJS, and AOAS are newer journals that emerged after 2000. As a result, these journals received fewer citations in the earlier period and got recognized better in the more recent period.  

We then turn our attention to the stability of the highest-ranked journals. Referring to Table \ref{tab_journal1}, we observe that the top four journals (AoS, Bka, JASA, and JRSSB, known as the Big-Four in statistics) maintain their positions strongly across different time periods. Furthermore, with a significance level of $\alpha = 10\%$, we reject the hypothesis that the top seven ranked items remain constant across the two time periods (Example \ref{Example_two_sample_set}). Specifically, for 2006-2010, the $95\%$ confidence set for the top-$7$ items includes:  
\begin{align*}
\emph{AoS, Bern, Bcs, Bka, JASA, JCGS JRSSB, ScaJS, StSci.}
\end{align*}
And for 2011-2015, the $95\%$ confidence set for the top-$7$ items includes:
\begin{align*}
\emph{AoS, Bcs, Biost, Bka, JASA, JMLR, JRSSA, JRSSB, StSci.}
\end{align*}
Clearly, these sets intersect only at 6 items, smaller than 7, reflecting a shift in the rankings over the two periods.

\subsection{Ranking of Movies}\label{rank_movie}
In this section, we construct confidence intervals for the ranks of movies or television series featured within the \emph{The Netflix Prize} competition \citep{bennett2007netflix}, which aims to enhance the precision of the Netflix recommendation algorithm. The dataset we examine corresponds to  
100 random 3 and 4 candidate elections drawn from Data Set 1 of \cite{mattei2012empirical}, which was extracted from the original \emph{The Netflix Prize} dataset, devoid of any ties. The dataset contains $196$ movies in total and $163759$ 3-way or 4-way comparisons. For simplicity, we only use the top ranked movie, although it is straightforward to apply the multi-level breaking to use the complete ranking data. This dataset can be accessed at the website: \url{https://www.preflib.org/dataset/00004}.

We compute two-sided, one-sided, and uniform one-sided confidence intervals employing the bootstrap method as described in Sections \ref{sec:onesam_twosided_inf} - \ref{Section_Application_Rank_inference}, based on the Two-Step Spectral method. The results are shown in Table \ref{real_data_movie1}. Additionally, the results from the one-step Vanilla Spectral method are detailed in Table \ref{real_data_movie2} in Appendix \ref{vanilla_spectral_real_data2}. 

\begin{table}[h]
    \centering
         	\def\arraystretch{0.95}
      \setlength\tabcolsep{3pt}
    \begin{tabular}{r|r r r r r r r}
    \toprule
    Movie & $\tilde{\theta}$ & $\tilde{r}$ & TCI & OCI & UOCI & Count  \\ \hline
    The Silence of the Lambs                      & $3.002$  & $1$  & $[1,  1]$  & $[1, n]$  & $[1, n]$  & $19589$ \\
The Green Mile                                & $2.649$  & $2$  & $[2,  4]$ & $[2, n]$  & $[2, n]$   & $5391$ \\
Shrek (Full-screen)                           & $2.626$  & $3$  & $[2,  4]$  & $[2, n]$  & $[2, n]$  & $19447$ \\
The X-Files: Season 2                         & $2.524$  & $4$  & $[2,  7]$  & $[2, n]$  & $[2, n]$   & $1114$ \\
Ray                                           & $2.426$  & $5$  & $[4,  7]$  & $[4, n]$  & $[4, n]$   & $7905$ \\
The X-Files: Season 3                         & $2.357$  & $6$  & $[4, 10]$  & $[4, n]$  & $[2, n]$   & $1442$ \\
The West Wing: Season 1                       & $2.278$  & $7$  & $[4, 10]$  & $[4, n]$  & $[4, n]$   & $3263$ \\
National Lampoon's Animal House               & $2.196$  & $8$  & $[6, 10]$  & $[6, n]$  & $[5, n]$  & $10074$ \\
Aladdin: Platinum Edition                     & $2.154$  & $9$  & $[6, 13]$  & $[6, n]$  & $[5, n]$   & $3355$ \\
Seven                                         & $2.143$ & $10$  & $[6, 11]$  & $[7, n]$  & $[6, n]$  & $16305$ \\
Back to the Future                            & $2.030$ & $11$  & $[9, 15]$  & $[9, n]$  & $[8, n]$   & $6428$ \\
Blade Runner                                  & $1.968$ & $12$ & $[10, 16]$ & $[10, n]$  & $[9, n]$   & $5597$ \\
Harry Potter and the Sorcerer's Stone         & $1.842$ & $13$ & $[12, 22]$ & $[12, n]$  & $[11, n]$  & $7976$ \\
High Noon                                     & $1.821$ & $14$ & $[11, 25]$ & $[11, n]$  & $[10, n]$   & $1902$ \\
Sex and the City: Season 6: Part 2            & $1.770$ & $15$ & $[11, 30]$ & $[11, n]$  & $[8, n]$    & $532$ \\
Jaws                                          & $1.749$ & $16$ & $[13, 25]$ & $[13, n]$ & $[13, n]$   & $8383$ \\
The Ten Commandments                          & $1.735$ & $17$ & $[13, 28]$ & $[13, n]$  & $[12, n]$  & $2186$ \\
Willy Wonka \& the Chocolate Factory          & $1.714$ & $18$ & $[13, 26]$ & $[13, n]$  & $[13, n]$   & $9188$ \\
Stalag 17                                     & $1.697$ & $19$ & $[12, 34]$ & $[12, n]$ & $[11, n]$    & $806$ \\
Unforgiven                                    & $1.633$ & $20$ & $[14, 29]$ & $[14, n]$ & $[14, n]$   & $9422$ \\
    \bottomrule
    \end{tabular}
    \caption{Ranking inference results for top-$20$ Netflix movies or tv series based on the Two-Step Spectral estimator. The first column $\tilde{\theta}$ denotes the estimated underlying scores. The second through the fifth columns denote their relative ranks, two-sided, one-sided, and uniform one-sided confidence intervals for ranks with coverage level $95\%$, respectively. The sixth column denotes the number of comparisons in which each movie is involved.}
    \label{real_data_movie1}
\end{table}

First note that the heterogeneity of the number of comparisons (``Count'' column in Table \ref{real_data_movie1}) is more severe relative to the journal ranking data, which leads to the adaptive length of rank confidence intervals. From ``OCI'' column of Table \ref{real_data_movie1}, we can test whether each individual movie belong to the top-$10$ rated movies ($K=10$ in Example \ref{example1}). We end up failing to reject the hypothesis for the first 12 films listed in Table \ref{real_data_movie1}. 
Furthermore, we can use the uniform one-sided confidence interval (``UOCI'') column to build candidate confidence set for the true top-$10$ movies ($K=10$ in Example \ref{exp:candidate_ad}).
The result suggests that except for ``\emph{Harry Potter and the Sorcerer’s Stone}", we should include the other top 15 ranked films in our top-$10$ confidence set. The reason we cannot exclude the films ``\emph{High Noon}" and ``\emph{Sex and the City: Season 6: Part 2}", despite these movies rank lower than ``\emph{Harry Potter and the Sorcerer’s Stone}", is due to their fewer number of comparisons and thus wider confidence intervals.  Similarly, the top-$5$ confidence set is the first 9 movies in Table~\ref{real_data_movie1}.

\section{Conclusion and Discussion}
\label{sec:discussion}

In this work, we studied the performance of the spectral method in preference score estimation, quantified the asymptotic distribution of the estimated scores, and explored one-sample and two-sample inference on ranks. In particular, we worked with general multiway comparisons with fixed comparison graphs, where the size of each comparison can vary and can be as low as only one. This is much closer to real applications than the homogeneous random sampling assumption imposed in the BTL or PL models. The applications of journal ranking and movie ranking have demonstrated the clear usefulness of our proposed methodologies.
Furthermore, we studied the relationship between the spectral method and the MLE in terms of estimation efficiency and revealed that with a carefully chosen weighting scheme, the spectral method can approximately achieve the same efficiency as the MLE, which is also verified using numerical simulations. Finally, to the best of our knowledge, it is the first time that effective two-sample rank testing methods have been proposed in the literature. 

Although we have made significant improvements in relaxing conditions, the role of general comparison graphs is still not fully understood, especially in the setting of multiway comparisons. Questions like how to design a better sampling regime, either online or offline, remain open. In addition, the spectral method essentially encodes multiway comparisons into pairwise comparisons, where the encoding will break data independence. The best encoding or breaking method should be further investigated. Finally, a set of recent works on ranking inferences opens the door to many possibilities of theoretical studies on ranking inferences {and related problems such as assortment optimization, under the setting of, say, rank time series, rank change point detection, rank panel data, recommendation based on rank inferences, uncertainty quantification and inference for properties of the optimal assortment.} These may find potential application in numerous management settings.


%% file: appendix_spectral.tex
\section{Numerical Studies} \label{sec:simu}

\subsection{Validate Consistency and Asymptotic Distribution for Fixed Graph}\label{simu:consist}


Consider a scenario where the total number of items to be compared  is $n = 50$ with the true preference scores $\theta^*=(\theta_1,\theta_2, \ldots, \theta_n)$ evenly distributed on $[-2, 2]$.  The total number of comparisons is represented by $|\mathcal{D}|$, whose specific value will be defined case by case. Our comparison process employs a heterogeneous comparison graph. For the first $\frac{1}{5}|\mathcal{D}|$ comparisons, the size of each comparison set is $\{2, 3, 4, 5\}$ with equal probability of $\frac{1}{4}$, and the comparison set is sampled homogeneously among the top $20\%$ of the $n$ items. 
For the second $\frac{1}{5}|\mathcal{D}|$ comparisons, the comparisons are conducted among the top $50\%$ of the $n$ items, also utilizing the aforementioned same hypergraph sampling method.   For the remaining $\frac{3}{5}|\mathcal{D}|$ entries, we also use the same hypergraph sampling method as before but comparisons are made among all $n$ items.

The above process generates heterogeneity in item hypergraph degrees and the size of the hypergraph edge. Note that in real applications, the comparison graph may be even endogenously generated, that is, the selection of comparison sets also depends on $\theta^*$: For example, items of higher rank tend to have a higher frequency of being compared.  
Given our comparison graph generation process detailed above, it is easy to see that $n^\dagger \asymp\frac{|\mathcal{D}|}{n}$. Based on this, we select the size of the dataset, $|\mathcal{D}|$, in such a manner that $\sqrt{\frac{\log n}{n^{\dagger}}}$ is evenly spaced on a grid ranging from $0.05$ to $0.20$. 

\noindent\textbf{Consistency and Statistical Rates.} 
We employ the spectral estimator $\tilde{\theta}$ outlined in Section \ref{sec:mle_con}, where we set $f(A_l) = |A_l|$ or $ \sum_{i\in A_l}e^{\theta_i^*}$ (oracle weight). The statistical rates obtained are depicted in Figure \ref{fig:consistency}, which shows the statistical errors are proportional to their theoretical rates  linearly. This simulation lends support to the  result on rate of convergence in Theorem \ref{thm_consist_normality}.

\begin{figure}[h]
    \centering
    \begin{tabular}{c}
    \includegraphics[width=0.8\textwidth]{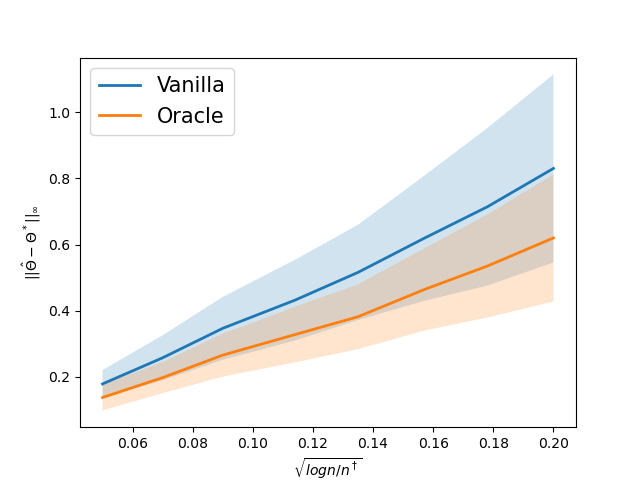}
   \end{tabular}
    \caption{$\ell_{\infty}$- statistical errors of the spectral estimator $\tilde\theta$ against the theoretical rate when $|\cD|$ varies. We let $|\cD|$ increases such that $\sqrt{\log n/n^{\dagger}}$ takes uniform grid between $[0.05,0.20].$ The solid lines represent the averaged statistical errors of 500 repetitions and  the light areas are formed by plus and minus one standard deviation curves to the average curve. The blue and yellow ones correspond to using $f(A_l) = |A_l|$ or $ \sum_{i\in A_l}e^{\theta_i^*}$ (oracle weight), respectively.}
    \label{fig:consistency}
\end{figure}

\noindent\textbf{Asymptotic Normality.} We now proceed to validate the asymptotic distribution of the entries $\tilde\theta_8$, $\tilde\theta_{20}$, and $\tilde\theta_{30}$, which are representative subjects from 3 subgroups of the $n$ items. 
From Figure \ref{fig:asym}, we observe that the empirical distribution of $\tilde\theta_8$, $\tilde\theta_{20}$, and $\tilde\theta_{30}$ are well approximated by the standard Gaussian distribution. This is consistent with the theoretical result on asymptotic distributions in Theorem \ref{thm_consist_normality}.

\begin{figure}[h]
    \centering
    \begin{tabular}{c}
    \includegraphics[width=1.0\textwidth]{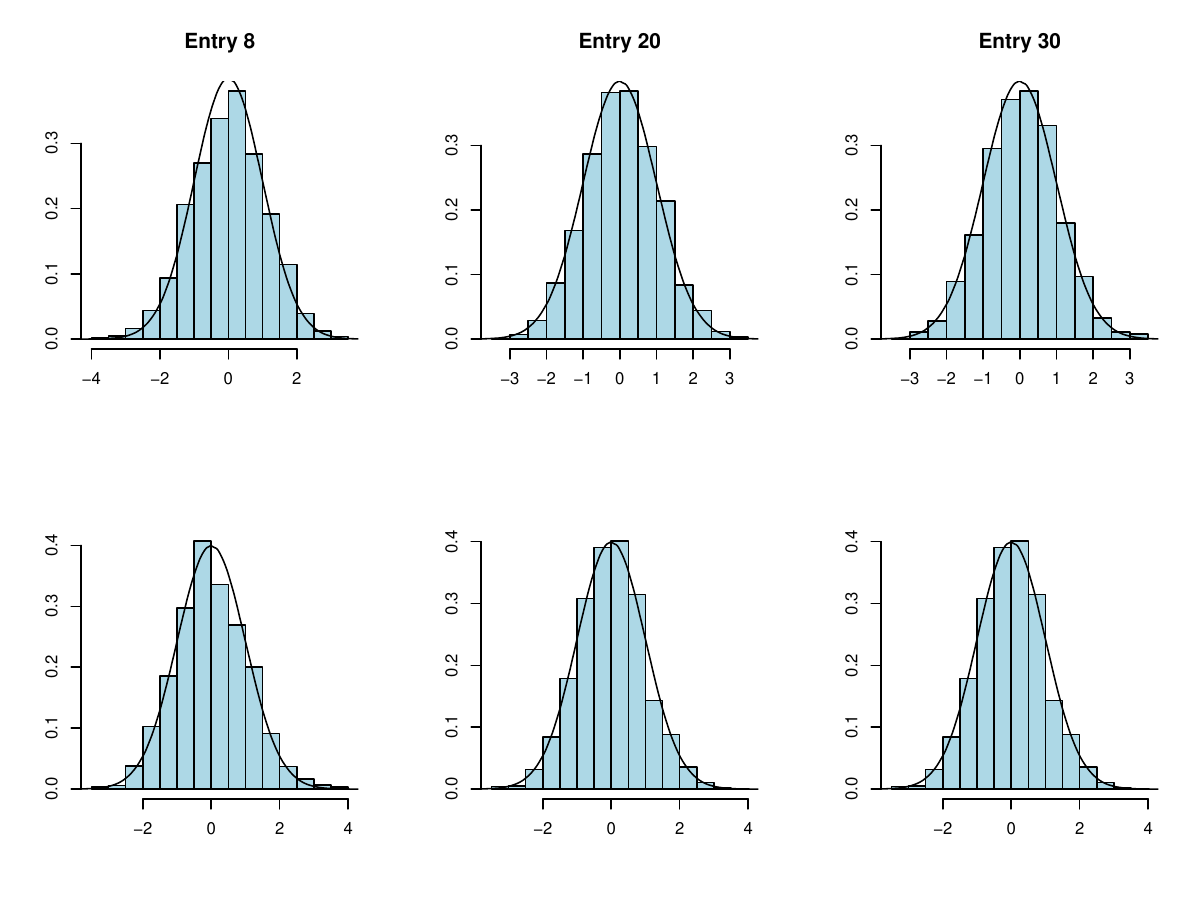}
   \end{tabular}
    \caption{Histograms for the normalized quantities $\rho_8(\tilde\theta)(\tilde\theta_8-\theta_8^*), \rho_{20}(\tilde\theta)(\tilde\theta_{20}-\theta_{20}^*),$ and $\rho_{30}(\tilde\theta)(\tilde\theta_{30}-\theta_{30}^*)$. Here, $\rho_i(\tilde\theta)$ is utilized as an estimator of the inverse of the theoretical standard deviation of $\tilde\theta_i$. The black curves denote the standard Gaussian distribution.  For this analysis, the total number of comparisons, $|\mathcal{D}|$, is set to 12,000, while the rest of the settings remain consistent with those outlined earlier. We use $f(A_l)=|A_l|$ for the three plots on the first row and use $f(A_l)=\sum_{i\in A_l}e^{\theta_i^*}$ for the three plots on the second row. }
    \label{fig:asym}
\end{figure}

\noindent\textbf{Validate Gaussian Multiplier Bootstrap using PP-Plot.} 
We finally validate the Gaussian approximation results discussed in Section~\ref{sec:onesam_twosided_inf}. We let $|\cD|=24000$ and investigate the distribution of $\cT$ in \eqref{Statistics_max} with $\cM=\{8,20,30\}$. By Theorem~\ref{Theorem_Bootstrap_Validity}, we have $\mathbb{E}[\mathbb{I}\{\cT > \cG_{1-\alpha}\}] = \PP(\cT > \cG_{1-\alpha}) \to \alpha$. We then verify this with various $\alpha \in \{0.05,0.10,\cdots,0.6\}$. 
For every $\alpha,$ we bootstrap $500$ times to compute the critical value $\cG_{1-\alpha}$, and repeat this whole procedure $2000$ times to calculate $\hat{\PP}(\cT > \cG_{1-\alpha}) = \frac{1}{2000} \sum_{b=1}^{2000} \mathbb{I}\{\cT^b > \cG_{1-\alpha}^b\}$ where $\cT^b$ and $\cG_{1-\alpha}^b$ are the pairwise score difference statistic and the corresponding bootstrap critical value of the $b$-th repetition. 
Note that $\hat{\PP}(\cT \le \cG_{1-\alpha})$ is the empirical coverage probability for the pairwise score difference statistic $\cT$.  Figure \ref{fig:ppplot} gives the so-called PP-plot, which shows $\hat{\PP}(\cT > \cG_{1-\alpha})$ against the theoretical significance level. 
From the figure, it the clear that the empirical probabilities matches well with the theoretical ones. Especially, when we apply the significant level of $\alpha=0.05$, the empirical exceptional probability is indeed around $0.05$.

\begin{figure}[h]
    \centering
    \begin{tabular}{c}
    \includegraphics[width=0.8\textwidth]{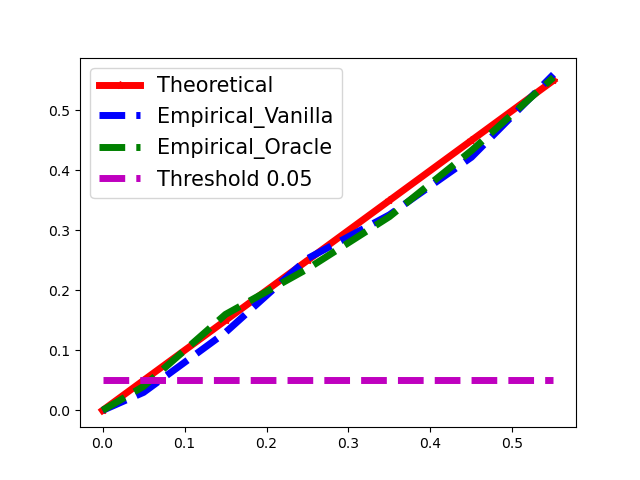}
   \end{tabular}
    \caption{PP-plot of empirical probability $\hat{\PP}(\cT > \cG_{1-\alpha})$ of $\cT$ given in \eqref{sec:onesam_twosided_inf} with $\cM=\{8,20,30\}$ against theoretical significance level $\alpha$. Here we choose $|\cD|=24000.$ The red solid and blue, and green dash-dotted lines represent theoretical and empirical probabilities using the Vanilla Spectral estimator and the Oracle Spectral estimator, respectively. The purple dotted line represents the case with a significance level $0.05$. }
    \label{fig:ppplot}
\end{figure}

\subsection{One-Sample Confidence Intervals for Ranks}
In this subsection, we provide numerical studies for validating our proposed framework for confidence interval construction in Section~\ref{sec:ranking_inference}. Throughout this subsection, we let $|\cD|$ vary among $\{12000,24000,36000\}$ and let the other data generation procedure be the same as section~\ref{simu:consist}. For each experiment, we conduct the bootstrap 500 times to compute the critical value according to $\alpha=0.05$ and further repeat the entire procedure 500 times.

\medskip
\noindent\textbf{Two-Sided Confidence Intervals.}
We construct two-sided confidence intervals (CI) for the $8,20,30$-th item ($\mathcal{M}=\{8\},\{20\},\{30\},$ respectively) using our method in Section~\ref{sec:ranking_estimation}. For the vanilla ($f(A_l) = |A_l|$) and oracle ($f(A_l) = \sum_{i \in A_l} e^{\theta_i^*}$) confidence intervals, we report: (a) EC($\theta$) -- the empirical coverage probability $\hat{\PP}(\cT \le \cG_{1-\alpha})$ for $\cT$ given in \eqref{Statistics_max}, which is also the overall empirical coverage probability for all the score differences $\theta_{k}^* - \theta_{m}^*, \forall k \ne m$ with $m\in \{8,20,30\}$, (b) EC($r$) -- the empirical coverage of the confidence intervals constructed for rank $r_{8},r_{20},r_{30}$, and furthermore (c) Length -- the length of the CIs for rank $r_{8},r_{20},r_{30}$, respectively.

    \begin{table}
	\begin{center}
		\def\arraystretch{1.2}
		\setlength\tabcolsep{4pt}
		\begin{tabular}{c|c||c|c|c||c|c|c}
			\toprule
			& & \multicolumn{3}{c||}{Vanilla Two-Sided CI} & \multicolumn{3}{c}{Oracle Two-Sided CI} \\
            \hline
	 	&$|\cD|$  & EC($\theta$) & EC($r$) & Length  & EC($\theta$) & EC($r$) & Length \\
			\hline
			\multirow{3}{*}{$\theta_8^*$	}	&$|\cD|=12000$	&0.954
			 &1.000 &6.384  &0.954 &1.000 &6.298 \\
			&$|\cD|=24000$	&0.950 &1.000 & 4.092
			 &0.968 &1.000 &4.090 \\
			&$|\cD|=36000$	  &0.956 &1.000  &3.008  &0.954 &1.000 &2.928 \\
   			\hline
			\multirow{3}{*}{$\theta_{20}^*$	}	&$|\cD|=12000$	&0.952
			 &1.000 &11.602  &0.960 &1.000 &10.082 \\
			&$|\cD|=24000$	&0.958 &1.000 & 7.450
			 &0.952 &1.000 &6.524 \\
			&$|\cD|=36000$	  &0.954 &1.000  &5.788  &0.958 &1.000 &5.068 \\
   \hline
   			\multirow{3}{*}{$\theta_{30}^*$	}	&$|\cD|=12000$	&0.950
			 &1.000 &17.502  &0.962 &1.000 &14.072 \\
			&$|\cD|=24000$	&0.952 &1.000 & 11.620
			 &0.960 &1.000 &9.528 \\
			&$|\cD|=36000$	  &0.956 &1.000  &9.262  &0.958 &1.000 &7.748 \\
			\bottomrule
		\end{tabular}
	\end{center}
	\caption{Empirical coverages and lengths of one-sample two-sided confidence intervals (CI). For the CI from the spectral method with vanilla or oracle weight, ``EC($\theta$)'' and ``EC($r$)'' denote the empirical CI coverages for score differences and ranks. ``Length'' denotes the length of the CI for $r_{8},r_{20},r_{30}$. 
	All numbers are averaged over 500 replications.}
	\label{tab1}
\end{table}

The empirical coverage probabilities for score differences, ranks, and the lengths of confidence intervals are summarized in Table \ref{tab1}. A close examination of Table \ref{tab1} reveals that the empirical coverage probabilities for score differences, are around $0.95$, aligning well with the nominal significance level. However, the CIs for the ranks of the 8th, 20th, and 30th items exhibit a more conservative nature. This is attributed to the constraint that ranks must be integers, leading to the empirical coverage probability of one across all instances. It is also noteworthy that with oracle weight, the spectral method indeed achieves narrower CIs than using the vanilla weight. 



\medskip
\noindent\textbf{One-sided Confidence Intervals.}
In this experiment, we set $K=5$ for the top-K test. We consider values of $m$ ranging from $K-2$ to $K+5$, inclusive. For each $m$, the proportion of rejections was computed. When $m \le K$, the null hypothesis holds true and the proportion approximates the size of the test. Conversely, under the alternative hypothesis, the proportion reflects the test’s power. Moreover, for cases where the null hypothesis is valid, we also assessed the test size concerning one-sided score differences.

The outcomes are presented in Table~\ref{tab2}. Evidently, when the alternative hypothesis is true and the true rank escalates, the power of our test rises fast to $1$. In contrast, when the null hypothesis is true, the test size is controlled at $\alpha=0.05$ for score differences (the values in the brackets). As for the rank test, it adopts a more conservative stance, manifesting a test size of zero. 

\begin{table}
	\begin{center}
		\def\arraystretch{1.2}
		\setlength\tabcolsep{2pt}
		\begin{tabular}{c||c|c|c||c|c|c|c|c}
			\toprule
			 & \multicolumn{3}{c||}{Null holds: $r_m \le K$} & \multicolumn{5}{c}{Alternative holds: $r_m > K$} \\
			 \hline
			$m$ & $K-2$ & $K-1$ & $K$ & $K+1$ & $K+2$ & $K+3$ &$K+4$ & $K+5$\\
			$\theta_{m}^*-\theta_{K}^*$ & $0.16$ & $0.08$ & $0$ & $-0.08$ & $-0.16$ & $-0.24$ &$-0.30$ & $-0.36$\\
			\hline 
			\multirow{2}{*}{$|\cD|=12000$	}	&0 (0.034)
			 &0 (0.032) &0 (0.046)  &0.010 &0.032 &0.222 &0.684 &0.936\\
    &0 (0.032)
			 &0 (0.040) & 0 (0.044)  &0.014 &0.058 &0.302 &0.712 &0.966\\
    \hline
			\multirow{2}{*}{$|\cD|=24000$	}	&0 (0.042) &0 (0.044)  & 0 (0.038)
			 &0.010 &0.212 &0.752 &0.990 &1.000\\
    &0 (0.046) &0 (0.038)  &0 (0.042)
			 &0.016 &0.252 &0.814 &0.998 &1.000\\
    \hline
			\multirow{2}{*}{$|\cD|=36000$	}	  &0 (0.042)  &0 (0.050) &0 (0.046) &0.034 &0.466 &0.972 &1.000 &1.000\\
   &0 (0.040)  &0 (0.048) &0 (0.044) &0.036 &0.542 &0.978 &1.000 &1.000\\
			\bottomrule
		\end{tabular}
	\end{center}
	\caption{Sizes and powers of the test $\psi_{m, K} = \mathbb{I}\{\cR_{m}^{\diamond} > K\}$ for testing hypothesis in \eqref{eq_top_K_test} with $K=5$. The numbers inside the bracket indicate the sizes of testing score differences as we discussed in the text. The numbers outside the bracket represent the sizes or powers of testing the rank, which is our goal here.  For any $|\cD|,$ the first row and second row represent the results using vanilla spectral and oracle spectral estimators, respectively.
	All displayed numbers are averaged over 500 replications. }
	\label{tab2}
\end{table}

\medskip
\noindent\textbf{Uniform One-Sided Intervals and Top-$K$ confidence set}
In the candidate admission Example~\ref{exp:candidate_ad}, to guarantee with high-confidence that all top-$K$ candidates are included, we need to build uniform one-sided coverage for all items, i.e. $\mathcal{M} = [n]$. The candidate set for top $K$ items is $\hat{\cI}_{K} = \{m \in [n] : \cR_{m L}^{\circ} \leq K\}$, all candidates that can not be rejected at significant level $\alpha$. In this simulation, we choose $K\in \{3,5,10\}$. 
In Table~\ref{tab3}, we report the average number of the top-$K$ candidate set $\hat{\cI}_{K}$ and its empirical coverages of the confidence intervals for score differences and for true top-K ranks over 500 replications. 
\begin{table}
	\begin{center}
		\def\arraystretch{1.2}
		\setlength\tabcolsep{3pt}
		\begin{tabular}{c|c|c|c|c|c|c}
		\toprule
		 &Method&EC($\theta$) & EC($r$) & $K=3$ & $K=5$ & $K=10$  \\
			\hline 
			 \multirow{2}{*}{$|\cD|=12000$	}		& Vanilla	 &0.962 &1.000	
			 & 6.394  &8.582  &14.702  \\
    	 & Oracle &0.958 &1.000	&6.206   &8.376  &14.334  \\
    \hline
			\multirow{2}{*}{$|\cD|=24000$	}	& Vanilla &0.968 &1.000
			 &5.184  &7.272  &13.090 \\
    & Oracle &0.962 &1.000
			 &5.054  &7.166  &12.830 \\
    \hline
			\multirow{2}{*}{$|\cD|=36000$	}	& Vanilla  &0.956 &1.000 &4.734 
			 &6.762  & 12.450 \\
    & Oracle &0.964 &1.000 &4.596 
			 &6.686  &12.200 \\
			\bottomrule
		\end{tabular}
	\end{center}
	\caption{Empirical coverages and lengths of sure screening confidence set for Example \ref{exp:candidate_ad}. ``EC($\theta$)'' represents the empirical coverage of confidence intervals for score differences. Note EC($\theta$) does not depend on $K$. ``EC($r$)'' denotes the empirical coverage of the confidence intervals for true top $K$ items. For all $K\in \{3,5,10\}$, all $\textrm{EC}(r) = 1$, so we collapse them into just one column. Other numbers are the lengths of the sure screening confidence set $\hat{\cI}_{K}$. 
	All displayed numbers are averaged over 500 replications.}
	\label{tab3}
\end{table}
We note from the simulation results that the lengths of the sure screening confidence set closely align with the theoretical values. This lends credence to the robustness of our methodology.
Again we see that with the oracle weight (the second row for each $|\mathcal D|$), the
spectral method generates narrower uniform on-sided intervals and thus a smaller top-$K$ candidates confidence set.


\subsection{Two-Sample Ranking Inference Test}

In this section, we perform numerical simulations for validating two-sample tests for ranks. 

\medskip
\noindent\textbf{Two-Sample Test for Individual Items.} 
We demonstrate the procedure to undertake the two-sample hypothesis testing in \eqref{eq_two_sample_rank_test}. We set the parameters $\theta_{[1]}^*$ and $\theta_{[2]}^*$ such that $\theta_{[1]}^*=\theta_{[2]}^*=(\theta_1^*,\ldots,\theta_n^*)$, uniformly distributed on the fixed grid ranging from $-2$ to $2$ with $n=50$ grid points. The null hypothesis under consideration is that the rank of the $10$th item remains constant across two samples, i.e., $r_{1,10}=r_{2,10}$. 
In an exploration of alternative hypotheses, we examine four distinct cases denoted by $i \in \{1,2,3,4\}$. For each of these cases, labelled as ``Alter $i$", we interchange $\theta^*_{2,10+3i}$ with $\theta^*_{2,10}$, while maintaining all other entries unchanged. Moreover, we permit the sample sizes $|\mathcal{D}_1|$ (corresponding to the first sample) and $|\mathcal{D}_2|$ (corresponding to the second sample) to be selected from $\{12000,24000,36000\}$.

The sizes and powers are displayed in Table \ref{tab_two1}. Same as the earlier one-sample test results, for each combination of $|\cD_1|$ and $|\cD_2|$, the results obtained through the vanilla spectral method are exhibited in the first row, while those obtained using the oracle spectral method are depicted in the second row.
\begin{table}
	\begin{center}
		\def\arraystretch{1.2}
		\setlength\tabcolsep{3pt}
		\begin{tabular}{c|c||c||c|c|c|c}
			\toprule
			\multicolumn{2}{c||}{} & \multicolumn{1}{c||}{Null holds} & \multicolumn{4}{c}{Alternative holds} \\
			 \hline
			$m$ & Method & Null  & Alter 1 & Alter 2 & Alter 3 &Alter 4 \\
			\hline 
			\multirow{2}{*}{$|\cD_1|=|\cD_2|=12000$	}	& Vanilla &0.000 
			  &0.000 &0.094 &0.652 &0.982 \\
    & Oracle &0.000 
			  &0.000 &0.162 &0.836 &0.998  \\
    \hline
			\multirow{2}{*}{$|\cD_1|=|\cD_2|=24000$	}	& Vanilla &0.000 
			  &0.022 &0.614 &0.996 &1.000 \\
    & Oracle &0.000  
			 &0.026 &0.744 &1.000 &1.000 \\
    \hline
			\multirow{2}{*}{$|\cD_1|=|\cD_2|=36000$	}	& Vanilla  &0.000    &0.048 &0.894 &1.000 &1.000 \\
   & Oracle &0.000   &0.108 &1.000 &1.000 &1.000  \\
			\bottomrule
		\end{tabular}
	\end{center}
	\caption{Sizes and powers of the two-sample hypothesis testing with significance level $\alpha=0.05$ for $r_{1,10}^*=r_{2,10}^*$ in Example \ref{Example_two_sample_item}.	To this end, we first construct confidence intervals for $r_{1,10}$ and $r_{2,10}$ respectively with a coverage level of $1-\frac{\alpha}{2}$. Subsequently, we combine these two confidence intervals by employing the Bonferroni adjustment. The values presented are averaged over 500 replications.}
	\label{tab_two1}
\end{table}
We observe from Table \ref{tab_two1} that when the null holds, the sizes of the tests are controlled well (actually conservative due to the integer nature of ranks), and under the alternatives, if we increase the sample sizes or the distances of $r_{2,10}-r_{1,10},$ the power of the test grows rapidly to $1$. Oracle spectral estimator performs better than vanilla spectral estimator as expected.

\begin{table}
	\begin{center}
		\def\arraystretch{1.2}
		\setlength\tabcolsep{3pt}
		\begin{tabular}{c|c||c||c|c|c|c}
			\toprule
			\multicolumn{2}{c||}{} & \multicolumn{1}{c||}{Null holds} & \multicolumn{4}{c}{Alternative holds} \\
			 \hline
			$m$ & Method & Null  & Alter 1 & Alter 2 & Alter 3 &Alter 4 \\
			\hline 
			\multirow{2}{*}{$|\cD_1|=|\cD_2|=12000$	}	& Vanilla &0.000 
			  &0.010 &0.114 &0.518 &0.866 \\
    & Oracle & 0.000 &0.014 &0.256 &0.638 &0.962  \\
    \hline
			\multirow{2}{*}{$|\cD_1|=|\cD_2|=24000$	}	& Vanilla &0.000 
			  &0.268 &0.754 &0.992 &1.000 \\
    & Oracle &0.000  &0.340 &0.892 &0.996 &1.000 \\
    \hline
			\multirow{2}{*}{$|\cD_1|=|\cD_2|=36000$	}	& Vanilla  &0.000    &0.678  &0.984 &1.000 &1.000 \\
   & Oracle &0.000   &0.764 &0.992 &1.000 &1.000  \\
			\bottomrule
		\end{tabular}
	\end{center}
	\caption{Sizes and powers of the two-sample hypothesis testing for the top-$10$ set $S_{1,10}^*=S_{2,10}^*$ in Example \ref{Example_two_sample_set}. 	Here we establish top-10 candidate confidence sets for $S_{1,10}$ and $S_{2,10}$ separately with a coverage level of $1-\frac{\alpha}{2}$ and combine them via Bonferroni adjustment. All displayed numbers are averaged over 500 replications.  }
	\label{tab_two2}
\end{table}

\medskip
\noindent\textbf{Two-Sample Test for Top-K Sets.} In this experiment, a series of simulations are conducted in correspondence with Example \ref{Example_two_sample_set}. 
The hypothesis under investigation here is $H_0: S_{1,10}=S_{2,10}$, which states that the top-$10$ sets of the two samples are identical. 
This experiment also considers the situation when alternative hypotheses are genuine. Specifically, for each alternative, denoted as ``Alter $i$'', we keep values of $\theta_1^*$, and swap $(\theta^*_{2,9-i},\ldots,\theta^*_{2,10})$ with $(\theta^*_{2,20},\ldots, \theta^*_{2,22+i})$. This swap involves selecting $i+2$ elements from the original top-10 set and interchanging them with elements outside this set, and keeping the remaining entries remain untouched. For simplicity, the experiment only accommodate equal sample size for two datasets, so $|\cD_1| = |\cD_2|$. And the sample sizes are selected from $\{12000, 24000, 36000\}$.

From Table \ref{tab_two2}, we observe that when the null holds, the sizes of the tests are controlled well although conservative, and under the alternatives, if we increase the sample sizes or the differences between $S_{1,10},$ and $S_{2,10}$ the power of the test grows rapidly to $1$. Consistent with testing individual ranks, oracle weight achieves better efficiency and outperforms vanilla weight in spectral method.

    \begin{table}[h]
	\begin{center}
		\def\arraystretch{1.2}
		\setlength\tabcolsep{3pt}
		\begin{tabular}{c||c|c|c|c|c|c}
			\toprule
			&Estimator  & $p=0.02$ & $p=0.05$ & $p=0.08$  & $p=0.11$ & $p=0.14$   \\
			\hline
   \multirow{4}{*}{$\ell_2$} & Vanilla &1.092 (0.140) &0.688 (0.086)  &0.543 (0.061) &0.301 (0.052) &0.181 (0.047) \\ 
   & Oracle &0.902 (0.102) &0.561 (0.061) &0.447 (0.043) &0.248 (0.040) &0.150 (0.037) \\
    &Two Step &0.906 (0.103) &0.562 (0.061) &0.447 (0.043) &0.248 (0.040) &0.150 (0.037) \\
    & MLE  &0.902 (1.102) &0.562 (0.061) &0.447 (0.043) &0.248 (0.040) &0.150 (0.037) \\
      &Two Step$\;-\;$MLE &0.046 (0.012) &0.018 (0.004) &0.011 (0.002) &0.008 (0.002) &0.006 (0.001)\\
    \hline
   \multirow{4}{*}{$\ell_{\infty}$} & Vanilla &0.427 (0.081) &0.259 (0.059) &0.206 (0.041) &0.116 (0.039) &0.070 (0.037) \\
   & Oracle &0.338 (0.063) &0.204 (0.034) &0.162 (0.030) &0.091 (0.027) &0.054 (0.022) \\
    & Two Step &0.337 (0.063) &0.204 (0.034) &0.162 (0.030) &0.091 (0.027) &0.054 (0.022) \\
    & MLE &0.337 (0.063) &0.204 (0.034) &0.162 (0.030) &0.091 (0.027) &0.054 (0.022) \\
    & Two Step$\;-\;$MLE &0.021 (0.007) &0.008 (0.002) &0.005 (0.002) &0.003 (0.001) &0.002 (0.001) \\
			\bottomrule
		\end{tabular}
	\end{center}
	\caption{Statistical errors for Vanilla Spectral, Oracle Spectral, Two-Step Spectral, and MLE when sampling probability $p\in \{0.02,0.05,0.08,0.11,0.14\}$. The first four rows are the comparisons of $\|\tilde{\theta}-\theta^*\|_2$ and the last four rows are results for  $\|\tilde{\theta}-\theta^*\|_{\infty}.$ The rows ``Two Step$\;-\;$MLE'' shows the difference between the two-step estimator and MLE. The numbers inside the bracket are standard errors. The  displayed numbers are averages among 500 replications.}
	\label{tab_encoding_consist}
\end{table}

\subsection{Validate Connection Between MLE and Spectral Method via Random Graph}

In this subsection, we explore and validate the relationship between the MLE and the spectral estimators in terms of their statistical errors and asymptotic variances, under the PL model with a random comparison graph. 
We set $n=50$ and take the uniform grids on $[-2, 2]$ as the score vector $\theta^* = (\theta_1^*, \theta_2^*, \ldots, \theta_n^*)$. Comparisons are sampled from a hypergraph with size $M=3$. Each sampled comparison tuple $(i_1, i_2, i_3)$ is compared independently for $L=10$ times if they are chosen. We calculate various estimators and examine their $\ell_2$ and $\ell_{\infty}$ statistical rates and asymptotic normality.

    \begin{table}[h]
	\begin{center}
		\def\arraystretch{1.2}
		\setlength\tabcolsep{3pt}
		\begin{tabular}{c|c||c|c|c|c|c}
			\toprule
			&Estimator  & $\tilde{\theta}_1$ & $\tilde{\theta}_2$ & $\tilde{\theta}_3$  & $\tilde{\theta}_4$ & $\tilde{\theta}_5$   \\
			\hline
    \multirow{4}{*}{$p=0.02$}&Vanilla &0.140 &0.152 &0.134 &0.153 &0.133 \\ 
    &Oracle &0.136 &0.149 &0.133 &0.149 &0.131 \\
   & Two Step &0.137 &0.149 &0.132 &0.149 &0.131 \\
    & MLE  &0.137 &0.149 &0.133 &0.149 &0.131 \\
    \hline
   \multirow{4}{*}{$p=0.05$}&Vanilla &0.095 &0.090 &0.082 &0.090 &0.087 \\
    &Oracle &0.092 &0.088 &0.080 &0.089 &0.083 \\
     &Two Step &0.092 &0.088 &0.080 &0.089 &0.083 \\
     &MLE &0.092 &0.088 &0.080 &0.089 &0.083 \\
     \hline
        \multirow{4}{*}{$p=0.08$}&Vanilla &0.077 &0.065 &0.072 &0.074 &0.072 \\
    &Oracle &0.075 &0.062 &0.069 &0.068 &0.071 \\
     &Two Step &0.075 &0.062 &0.069 &0.067 &0.071 \\
     &MLE &0.075 &0.063 &0.069 &0.067 &0.071 \\
			\bottomrule
		\end{tabular}
	\end{center}
	\caption{Asymptotic variances of Vanilla Spectral, Oracle Spectral, Two-Step Spectral, and MLE estimators when sampling probability $p\in \{0.02,0.05,0.08\}$. Here we compare the first five entries $\{\tilde{\theta}_1,\cdots,\tilde{\theta}_5\}$. The  displayed numbers are averages among 500 replications.}
	\label{tab_variance_encoding}
\end{table}

We consider the following estimators:
(i) Vanilla Spectral estimator (by setting $f(A_l) = |A_l|$),
(ii) Oracle Spectral estimator (by setting $f(A_l) = \sum_{i \in A_l} e^{\theta_i^*}$),
(iii) Two-Step Spectral estimator (by initially setting $f(A_l) = |A_l|$ to obtain $\hat\theta$, followed by a re-run of the spectral method with $f(A_l) = \sum_{i \in A_l} e^{\hat\theta_i}$), and
(iv) the MLE estimator as presented in \cite{fan2022ranking}.

We vary the sampling probability $p$ among the values ${0.02, 0.05, 0.08, 0.11, 0.14}$ and summarize the $\ell_2$ and $\ell_{\infty}$ errors in Table \ref{tab_encoding_consist}. Additionally, we investigate the asymptotic variances of these four estimators by examining the first 5 entries. The findings are consolidated in Table \ref{tab_variance_encoding}.
Table \ref{tab_encoding_consist} and Table \ref{tab_variance_encoding} reveal that the Oracle Spectral and Two-Step Spectral methods exhibit performance closely aligned with that of the MLE, with respect to both statistical rates and efficiency. In contrast, the Vanilla Spectral method demonstrates a marginally inferior performance compared to the aforementioned trio. This observation corroborates our assertion that the asymptotic distribution of Two-Step (or Oracle Weight) Spectral method converges to that of the MLE.

\newpage
\section{Additional Real Data Results}\label{sec:add_real_data}

\subsection{Real Data Analysis: Vanilla Spectral Method for Journal Ranking} 
\label{vanilla_spectral_real_data}
In this section, we present the ranking inference outcomes for the 33 journals, derived using the Vanilla Spectral method (or the one-step spectral method). These results, corresponding to the two previously mentioned time periods, are compiled and displayed in Table \ref{Table_one_spectral}.
\begin{table}
    \centering
         	\def\arraystretch{0.95}
      \setlength\tabcolsep{3.0pt}
    \begin{tabular}{c|r r r r r r|r r r r r r}  
    \toprule
               \multicolumn{1}{c|}{} & \multicolumn{6}{c|}{$2006-2010$} & \multicolumn{6}{c}{$2011-2015$}\\
       Journal & \multicolumn{1}{c}{$\tilde{\theta}$} & $\tilde{r}$ & TCI & OCI & UOCI & Count & \multicolumn{1}{c}{$\tilde{\theta}$} & $\tilde{r}$ & TCI & OCI & UOCI & Count\\ \hline
       JRSSB  & $1.635$ & $1$ & $[1, 1]$ & $[1, n]$ & $[1, n]$ & $5282$ &  $1.536$ & $1$ & $[1, 2]$ & $[1, n]$ & $[1, n]$ & $5513$ \\ 
       AoS  & $1.231$ & $3$ & $[2, 4]$ & $[2, n]$ & $[2, n]$ & $7674$ & $1.533$ & $2$ & $[1, 2]$ & $[1, n]$ & $[1, n]$ &$11316$ \\ 
       Bka  & $1.288$ & $2$ & $[2, 4]$ & $[2, n]$ & $[2, n]$ & $5579$ & $1.163$ & $3$ & $[3, 4]$ & $[3, n]$ & $[3, n]$ & $6399$ \\ 
       JASA  & $1.161$ & $4$ & $[2, 4]$ & $[3, n]$ & $[2, n]$ & $9652$ & $1.060$ & $4$ & $[3, 4]$ & $[3, n]$ & $[3, n]$ & $10862$\\ 
       JMLR  & $0.020$ & $19$ & $[12, 25]$ & $[12, n]$ & $[11, n]$ & $1100$ & $0.779$ & $5$ & $[5, 6]$ & $[5, n]$ & $[5, n]$ & $2551$ \\ 
       Biost  & $0.293$ & $11$ & $[10, 20]$ & $[10, n]$ & $[9, n]$ &  $2175$ & $0.619$ & $6$ & $[5, 10]$ & $[5, n]$ & $[5, n]$  & $2727$\\  
       Bcs  & $0.742$ & $6$ & $[5, 8]$ & $[5, n]$ & $[5, n]$ & $6614$ & $0.509$ & $7$ & $[6, 10]$ & $[6, n]$ & $[6, n]$  & $6450$ \\ 
       StSci  & $0.729$ & $7$ & $[5, 9]$ & $[5, n]$ & $[5, n]$ & $1796$ & $0.455$ & $8$ & $[6, 11]$ & $[6, n]$ & $[6, n]$  & $2461$ \\ 
       Bern  & $0.838$ & $5$ & $[5, 7]$ & $[5, n]$ & $[5, n]$ & $1575$ & $0.424$ & $9$ & $[6, 13]$ & $[7, n]$ & $[6, n]$ & $2613$ \\
       JRSSA  & $0.237$ & $14$ & $[10, 20]$ & $[10, n]$ & $[8, n]$ & $893$ & $0.418$ & $10$ & $[6, 13]$ & $[6, n]$ & $[6, n]$ & $865$ \\ 
       JCGS  & $0.575$ & $8$ & $[6, 10]$ & $[7, n]$ & $[5, n]$ & $2493$ & $0.306$ & $11$ & $[8, 13]$ & $[8, n]$ & $[8, n]$ & $3105$ \\ 
       Sini  & $0.384$ & $10$ & $[9, 16]$ & $[9, n]$ & $[8, n]$ & $3701$ & $0.263$ & $12$ & $[9, 13]$ & $[9, n]$ & $[8, n]$ & $4915$ \\ 
       ScaJS  & $0.504$ & $9$ & $[8, 14]$ & $[8, n]$ & $[7, n]$ & $2442$ & $0.240$ & $13$ & $[9, 14]$ & $[9, n]$ & $[8, n]$ & $2573$ \\
       AoAS  & $-1.229$ & $30$ & $[27, 32]$ & $[28, n]$ & $[27, n]$ & $1258$ & $0.066$ & $14$ & $[13, 20]$ & $[13, n]$ & $[13, n]$  & $3768$  \\  
       JRSSC  & $0.068$ & $17$ & $[11, 24]$ & $[12, n]$ & $[11, n]$  & $1401$ & $-0.008$ & $15$ & $[14, 22]$ & $[14, n]$ & $[13, n]$ & $1492$ \\
       JSTA  & $0.236$ & $15$ & $[9, 21]$ & $[10, n]$ & $[8, n]$ & $751$ & $-0.061$ & $16$ & $[14, 23]$ & $[14, n]$ & $[13, n]$ & $1026$ \\ 
       JSPI  & $-0.320$ & $26$ & $[22, 26]$ & $[22, n]$ & $[22, n]$ & $6505$ & $-0.073$ & $17$ & $[14, 22]$ & $[14, n]$ & $[14, n]$ & $6732$ \\ 
       CanJS  & $0.089$ & $16$ & $[11, 23]$ & $[11, n]$ & $[11, n]$  &  $1694$ & $-0.074$ & $18$ & $[14, 23]$ & $[14, n]$ & $[14, n]$ & $1702$ \\ 
       JMVA  & $-0.081$ & $22$ & $[16, 25]$ & $[16, n]$ & $[13, n]$ & $3833$ & $-0.130$ & $19$ & $[15, 23]$ & $[15, n]$ & $[15, n]$ & $6454$ \\
       EJS  & $-1.476$ & $31$ & $[30, 33]$ & $[30, n]$ & $[30, n]$ & $1366$ & $-0.188$ & $20$ & $[15, 25]$ & $[15, n]$ & $[15, n]$  & $4112$\\ 
       Extrem  & $-2.416$ & $33$ & $[32, 33]$ & $[32, n]$ & $[29, n]$ & $173$ & $-0.242$ & $21$ & $[14, 30]$ & $[15, n]$ & $[14, n]$ & $487$ \\ 
       SMed  & $-0.201$ & $24$ & $[19, 26]$ & $[20, n]$ & $[18, n]$ & $6626$ & $-0.248$ & $22$ & $[17, 26]$ & $[17, n]$ & $[16, n]$ & $6857$ \\ 
       SPLet  & $-0.029$ & $20$ & $[13, 24]$ & $[14, n]$ & $[13, n]$ & $3651$ & $-0.346$ & $23$ & $[20, 28]$ & $[21, n]$ & $[19, n]$  & $4439$\\ 
       JClas  & $-0.034$ & $21$ & $[10, 26]$ & $[11, n]$ & $[9, n]$  & $260$ & $-0.360$ & $24$ & $[14, 30]$ & $[15, n]$ & $[12, n]$  & $224$ \\
       ISRe  & $0.283$ & $12$ & $[8, 21]$ & $[8, n]$ & $[8, n]$ & $511$ & $-0.449$ & $25$ & $[20, 30]$ & $[21, n]$ & $[17, n]$ & $905$ \\
       AISM  & $0.275$ & $13$ & $[9, 20]$ & $[10, n]$ & $[8, n]$ & $1313$ & $-0.466$ & $26$ & $[22, 30]$ & $[22, n]$ & $[20, n]$ & $1605$ \\ 
       CSDA  & $-0.953$ & $28$ & $[27, 31]$ & $[27, n]$ & $[27, n]$ & $6732$ & $-0.513$ & $27$ & $[23, 30]$ & $[23, n]$ & $[22, n]$ & $8717$ \\ 
       JNS  & $-0.209$ & $25$ & $[17, 26]$ & $[18, n]$ & $[16, n]$ & $1286$ & $-0.546$ & $28$ & $[22, 30]$ & $[23, n]$ & $[22, n]$ & $1895$ \\ 
       AuNZ  & $0.057$ & $18$ & $[11, 24]$ & $[11, n]$ & $[10, n]$ & $862$ & $-0.559$ & $29$ & $[22, 30]$ & $[22, n]$ & $[21, n]$ & $816$ \\ 
       SCmp  & $-0.126$ & $23$ & $[16, 26]$ & $[16, n]$ & $[13, n]$ & $1309$ & $-0.568$ & $30$ & $[23, 30]$ & $[23, n]$ & $[22, n]$ & $2650$ \\  
       Bay  & $-1.700$ & $32$ & $[30, 33]$ & $[30, n]$ & $[27, n]$ & $279$ & $-1.144$ & $31$ & $[31, 32]$ & $[31, n]$ & $[31, n]$ & $842$ \\ 
       CSTM  & $-0.967$ & $29$ & $[27, 31]$ & $[27, n]$ & $[27, n]$ & $2975$ & $-1.343$ & $32$ & $[31, 32]$ & $[31, n]$ & $[31, n]$ & $4057$ \\ 
       JoAS  & $-0.906$ & $27$ & $[27, 30]$ & $[27, n]$ & $[27, n]$  & $1055$ & $-2.051$ & $33$ & $[33, 33]$ & $[33, n]$ & $[33, n]$ & $2780$ \\ 
       \bottomrule
    \end{tabular}
    \caption{Ranking inference results for 33 journals in 2006-2010 and 2011-2015 based on the Vanilla Spectral estimator. For each time period, there are 6 columns. The first column $\tilde{\theta}$ denotes the estimated underlying scores. The second through the fifth columns denote their relative ranks, two-sided, one-sided, and uniform one-sided confidence intervals for ranks with coverage level $95\%$, respectively. The sixth column denotes the number of comparisons in which each journal is involved.}
    \label{Table_one_spectral}
\end{table}

From our analysis, we can draw parallels to the findings from the two-step regime described in the main text. With a significance level of $\alpha = 10\%$, we note that the ranks of the following journals have shown significant differences between the two time periods:
\begin{align*}
\emph{ AISM, AoAS, EJS, Extrem, JMLR, JoAS.}
\end{align*}
Further, at a significance level of $\alpha = 10\%$, we observe that the top five ranked journals' sets exhibit notable differences across the two periods. Specifically, for the 2006-2010 period, the 95\% confidence set for the top-5 ranked items includes: \emph{AoS, Bcs, Bern, Bka, JASA, JRSSB, JCGS, StSci.}
And for the 2011-2015 period, the 95\% confidence set for the top-5 ranked items includes: \emph{AoS, Biost, Bka, JASA, JMLR, JRSSB.}

\subsection{Real Data Analysis: Vanilla Spectral Method for Movie Ranking} \label{vanilla_spectral_real_data2}
In this section, we carry out the ranking inference on the movies and TV series featured in the \emph{Netflix Prize} dataset \citep{bennett2007netflix}. We follow all details described in Section \ref{rank_movie} of the main text. The only distinction is that here, we employ the Vanilla Spectral estimator (one-step spectral estimator) in lieu of the Two-Step Spectral estimator. Similar results are presented in Table \ref{real_data_movie2}.

\begin{table}[H]
    \centering
             	\def\arraystretch{0.95}
      \setlength\tabcolsep{3pt}
    \begin{tabular}{r|r r r r r r r}
    \toprule
    Movie & $\tilde{\theta}$ & $\tilde{r}$ & TCI & OCI & UOCI & Count \\ \hline
    The Silence of the Lambs               & $2.423$ & $1$ & $[1, 1]$ & $[1, n]$ &  $[1, n]$ &  $19589$ \\
The Green Mile                         & $2.171$ & $2$ & $[2, 5]$ & $[2, n]$ &  $[2, n]$ &   $5391$ \\
Shrek (Full-screen)                    & $2.142$ & $3$ & $[2, 4]$ & $[2, n]$ &  $[2, n]$ &  $19447$ \\
The X-Files: Season 2                  & $2.058$ & $4$ & $[2, 9]$ & $[2, n]$ &  $[1, n]$ &   $1114$ \\
Ray                                    & $1.998$ & $5$ & $[3, 7]$ & $[4, n]$ &  $[2, n]$ &   $7905$ \\
The X-Files: Season 3                  & $1.872$ & $6$ & $[4, 11]$ &  $[4, n]$ &  $[2, n]$ &   $1442$ \\
The West Wing: Season 1                & $1.870$ & $7$ & $[4, 11]$ & $[4, n]$ &  $[4, n]$ &   $3263$ \\
National Lampoon's Animal House        & $1.770$ & $8$ & $[5, 13]$ & $[5, n]$ &  $[5, n]$ &  $10074$ \\
Aladdin: Platinum Edition              & $1.744$ & $9$ & $[5, 17]$ & $[5, n]$ &  $[4, n]$ &   $3355$ \\
Seven                                  & $1.713$ & $10$ &  $[6, 16]$ & $[6, n]$ &  $[5, n]$ &  $16305$ \\
Back to the Future                     & $1.672$ & $11$ & $[6, 18]$ & $[6, n]$ &  $[5, n]$ &   $6428$ \\
Blade Runner                           & $1.574$ & $12$ & $[9, 20]$ & $[9, n]$ &  $[7, n]$ &   $5597$ \\
Harry Potter and the Sorcerer's Stone  & $1.518$ & $13$ & $[10, 24]$ & $[10, n]$ & $[10, n]$ &   $7976$ \\
High Noon                              & $1.490$ & $14$ & $[8, 30]$ & $[9, n]$ &  $[6, n]$ &  $1902$ \\
Jaws                                   & $1.436$ & $15$ & $[12, 32]$ & $[12, n]$ & $[12, n]$ &   $8383$ \\
The Ten Commandments                   & $1.423$ & $16$ & $[11, 34]$ & $[12, n]$ &  $[9, n]$ &  $2186$ \\
Sex and the City: Season 6: Part 2     & $1.393$ & $17$ & $[8, 39]$ & $[9, n]$ &  $[6, n]$ &    $532$ \\
Stalag 17                              & $1.377$ & $18$ & $[9, 38]$ & $[11, n]$ &  $[6, n]$ &    $806$ \\
Willy Wonka \& the Chocolate Factory    & $1.370$ & $19$ & $[13, 34]$ & $[13, n]$ & $[12, n]$ &   $9188$ \\
Blazing Saddles                        & $1.367$ & $20$ & $[13, 34]$ & $[13, n]$ & $[12, n]$ &   $7096$ \\
    \bottomrule
    \end{tabular}
    \caption{Ranking inference results for top-$20$ Netflix movies or TV series based on the Vanilla Spectral estimator. The first column $\tilde{\theta}$ denotes the estimated underlying scores. The second through the fifth columns denote their relative ranks, two-sided, one-sided, and uniform one-sided confidence intervals for ranks with coverage level $95\%$, respectively. The sixth column denotes the number of comparisons in which each movie is involved.}
    \label{real_data_movie2}
\end{table}

\newpage

\section{Proofs for Fixed Graph (Section \ref{sec4.1})}
{We would also like to emphasize that in the derivation of the non-asymptotic statistical rates of all proofs in the following sections, although we write some high-probability terms as $1-o(1)$ or $o_p(\cdot),O_p(\cdot)$, as we use Bernstein and Hoeffding-type concentration inequalities through our analysis, these probability terms should be regarded as $1-\cO(n^{-\zeta})$ with some $\zeta\ge 2$ (different choice of $\zeta$ will only affect constant terms in the involved concentration inequalities). }

\subsection{Proof of Theorem \ref{thm_approx}}

We need to verify \eqref{approx1}, where we hope to show $\hat\pi_i$ and $\bar\pi_i = \frac{\sum_{j:j\ne i} P_{ji}\pi_j^*}{\sum_{j:j\ne i} P_{ij}}$ are close. Let 
\begin{equation}
\delta_i = \frac{\hat\pi_i - \bar\pi_i}{\pi_i^*} = \frac{\sum_{j:j\ne i} P_{ji} (\hat\pi_j - \pi_j^*)}{\pi_i^* \sum_{j:j\ne i} P_{ij}} = \frac{\sum_{j:j\ne i} P_{ji} \pi_j^*\delta_j}{\pi_i^* \sum_{j:j\ne i} P_{ij}} + \frac{\sum_{j:j\ne i} P_{ji} (\bar\pi_j - \pi_j^*)}{\pi_i^* \sum_{j:j\ne i} P_{ij}} \,.
\label{eq_delta_sp}
\end{equation}
Define $R = (R_1,\dots,R_n)^\top$ where $R_{i} = \sum_{j:j\ne i} P_{ji} (\bar\pi_j - \pi_j^*)$ and recall $\Omega = \{\Omega_{ij}\}_{i\le n, j\le n}$ where $\Omega_{ij} = - P_{ji}\pi_j^*$ for $i\ne j$ and $\Omega_{ii} = \sum_{j:j\ne i} P_{ij}\pi_i^*$. The above equation can be written as 
\[
\Omega \delta = R\,.
\]
Next we will work on bounding $\|R\|_{\infty} = \cO_P(a_{n,p})$, leading to the bound $\|R\|_{2} = \cO_P(\sqrt{n} a_{n,p})$. Then we hope to bound the minimal eigenvalue of $\Omega$ and invert $\Omega$. However, since $E[\Omega |\mathcal G \text{ or }
\tilde{\mathcal G}]$ has minimal eigenvalue equal to zero, with eigenvector ${\bf 1}$. So we only work on spaces orthogonal to ${\bf 1}$. Following the notation and derivation of \cite{gao2021uncertainty}, 
\[
\lambda_{\min, \bot}(\Omega) = \min_{\|v\|=1,v^\top {\bf 1}=0} v^\top \Omega v\,.
\]
Then, letting $\bar\delta = n^{-1}\delta^\top{\bf 1}$, we have
\begin{equation*}
\begin{aligned}
\|R\|_2 &= \|\Omega\delta\|_2 \ge \|E[\Omega|\mathcal G \text{ or } \tilde{\mathcal G} ] \delta\|_2 - \|\Omega - E[\Omega|\mathcal G \text{ or } \tilde{\mathcal G} ]\| \|\delta\|_2 \\
& = \|E[\Omega|\mathcal G \text{ or } \tilde{\mathcal G} ] (\delta - \bar\delta {\bf 1}) \|_2 - \|\Omega - E[\Omega|\mathcal G \text{ or } \tilde{\mathcal G} ]\| \|\delta\|_2 \\
& \ge \lambda_{\min\bot}(E[\Omega|\mathcal G \text{ or } \tilde{\mathcal G} ]) (\|\delta\|_2 - \sqrt{n} |\bar\delta|) - \|\Omega - E[\Omega|\mathcal G \text{ or } \tilde{\mathcal G} ]\| \|\delta\|_2 \,.
\end{aligned}
\end{equation*}
Therefore, 
\begin{equation}
\label{bound_delta_l2}
(\lambda_{\min\bot}(E[\Omega|\mathcal G \text{ or } \tilde{\mathcal G} ]) - \|\Omega - E[\Omega|\mathcal G \text{ or } \tilde{\mathcal G} ]\|) \|\delta\|_2 \le \|R\|_2 + \lambda_{\min\bot}(E[\Omega|\mathcal G \text{ or } \tilde{\mathcal G} ]) \sqrt{n} |\bar\delta| \,.
\end{equation}
We will find $b_{n,p},c_{n,p}$ such that ${e^{-\bar\kappa}} b_{n,p} \lesssim \lambda_{\min, \bot}(E[\Omega|\mathcal G \text{ or } \tilde{\mathcal G} ]) \le \lambda_{\max}(E[\Omega|\mathcal G \text{ or } \tilde{\mathcal G} ]) \lesssim { e^{\bar\kappa}} b_{n,p}$ and $\|\Omega - E[\Omega|\mathcal G \text{ or } \tilde{\mathcal G} ]\| = o_P(b_{n,p})$, and $|\bar\delta|=\cO_P(c_{n,p})$. Then we conclude $\|\delta\|_2 = \cO_P(\sqrt{n} ({ e^{\bar\kappa}} a_{n,p} b_{n,p}^{-1}+c_{n,p}))$. This $\ell_2$ bound of $\delta$, together with \eqref{eq_delta_sp}, actually generates an $\ell_\infty$ bound of $\delta$, that is $\|\delta\|_\infty = \cO_P(d_{n,p})$, which is of a smaller order comparing to the leading term that dominates the asymptotic distribution. Specifically, 
\[
\frac{\hat\pi_i - \pi_i^*}{\pi_i^*} = \frac{\bar\pi_i - \pi_i^*}{\pi_i^*} + \delta_i ={ \frac{\sum_{j:j\ne i} (P_{ji}\pi_j^* - P_{ij}\pi_i^*)}{\pi_i^* \sum_{j:j\ne i} P_{ij}} }+ \delta_i\,.
\]
As we have seen in { Section \ref{sec4.1} that conditioning on $\mathcal G$, $\frac{\bar\pi_i - \pi_i^*}{\pi_i^*}$ in \eqref{approx1} can be well approximated by $J_i^*$ defined in \eqref{approx2}, that is $\frac{\bar\pi_i - \pi_i^*}{\pi_i^*} = \cO_P(J_i^*)$. In addition, $J_i^*$ is asymptotically normally distributed with variance given in \eqref{eq:fan1}. By Assumption \ref{ass4.1}, this variance term is in the order of $O(n^\dagger (e^{\theta_i + \max_k\theta_k^*}) / (n^\dagger e^{\theta_i})^2 = O(e^{\bar\kappa}/n^\dagger)$. Therefore,
}
\[
\frac{\bar\pi_i - \pi_i^*}{\pi_i^*} = \cO_P(J_i^*) = \cO_P({ e^{\bar\kappa/2}}/\sqrt{n^\dagger})\,.
\]
So it suffices to confirm that $\|\delta\|_\infty = O_P(d_{n,p}) = o_p(1/\sqrt{n^\dagger})$.

Now let us get into the details of finding out $a_{n,p}, b_{n,p}, c_{n,p}, d_{n,p}$ to prove \eqref{approx1}. For $j\ne i$, define 
\begin{equation} \label{eq:pi_bar_leave_one_out}
\bar\pi_j^{(-i)} = \frac{\sum_{k:k\ne j} \tilde P_{kj}\pi_k^*}{\sum_{k:k\ne j} \tilde P_{jk}} \,,
\end{equation}
where $\tilde P_{jk}$ changes any randomness in $(A_l,c_l)$ to its expectation if $(A_l,c_l)$ involves item $i$, that is, $\tilde P_{jk}$ removes any randomness in comparison with item $i$. In the case of fixed $\mathcal G$, { following the notation $Z_{A_l}^k$ in Section \ref{sec4.1},}
\[
\tilde P_{jk} = \frac1d \sum\limits_{l \in \mathcal D} 1(j,k\in A_l) \tilde Z_{A_l}^k \,,
\]
where $\tilde Z_{A_l}^k = E[Z_{A_l}^k]$ if $i \in A_l$ and $\tilde Z_{A_l}^k = Z_{A_l}^k$ otherwise. With the definition of $\bar\pi_j^{(-i)}$, we have
\begin{equation*}
\begin{aligned}
R_i &= \sum_{j:j\ne i} P_{ji} (\bar\pi_j - \pi_j^*) = \sum_{j:j\ne i} P_{ji} (\bar\pi_j - \bar\pi_j^{(-i)}) + \sum_{j:j\ne i} P_{ji} (\bar\pi_j^{(-i)} - \pi_j^*) = I_{i,1} + I_{i,2} \,.
\end{aligned}
\end{equation*}
Note that we expect $I_1$ to be small as $\bar\pi_j$ and $\bar\pi_j^{(-i)}$ only differs for comparisons involving $i$, and we expect $I_2$ to be small since now $P_{ji}$, which only involves comparison including item $i$, and $\bar\pi_j^{(-i)} - \pi_j^*$ are independent. 

Next we provide bounds for $I_1$ and $I_2$ separately. Let us first work on $I_1$.
\begin{equation*}
\begin{aligned}
\max |I_{i,1}| &= \max_i |\sum_{j:j\ne i} P_{ji} (\bar\pi_j - \bar\pi_j^{(-i)})| \\
& = \max_i \bigg| \sum_{j:j\ne i} P_{ji} \bigg( \frac{\sum_{k:k\ne j} P_{kj}\pi_k^*}{\sum_{k:k\ne j} P_{jk}} - \frac{\sum_{k:k\ne j} \tilde P_{kj}\pi_k^*}{\sum_{k:k\ne j} \tilde P_{jk}}\bigg)\bigg| \\
& = \max_i \bigg| \sum_{j:j\ne i} P_{ji} \bigg( \frac{\sum_{k:k\ne j} P_{kj}\pi_k^*}{\sum_{k:k\ne j} P_{jk}} - \frac{\sum_{k:k\ne j} P_{kj}\pi_k^*}{\sum_{k:k\ne j} \tilde P_{jk}} - \frac{\sum_{k:k\ne j} (\tilde P_{kj} - P_{kj}) \pi_k^*}{\sum_{k:k\ne j} \tilde P_{jk}}\bigg) \bigg| \\
& \lesssim \frac{n n^{\ddagger}}{d} \max_j \Big| \frac{1}{\sum_{k:k\ne j} \tilde P_{jk}} - \frac{1}{\sum_{k:k\ne j} P_{jk}} \Big| \sum_{k:k\ne j} P_{kj}\pi_k^* + \max_i \bigg| \sum_{j:j\ne i} P_{ji} \frac{\sum_{k:k\ne j} (\tilde P_{kj} - P_{kj}) \pi_k^*}{\sum_{k:k\ne j} \tilde P_{jk}} \bigg| \\
& =: \Delta_{1} + \Delta_{2} \,,
\end{aligned}
\end{equation*}
{ where the inequality uses $|P_{ji}| \lesssim n^{\ddagger} / d$ and holds with probability $1-o(1)$}.
It is not hard to show that $\max_j | \sum_{k:k\ne j} (\tilde P_{kj} - P_{kj})\pi_k^*| = \cO_P({ e^{\bar\kappa}}(dn)^{-1} \sqrt{n^\ddagger \log n})$. Furthermore, when $n^\dagger \gtrsim \log n$, $\max_{i} \big|\sum_{j:j\ne i} P_{ij} - \sum_{j:j\ne i} E[P_{ij} | \mathcal G] \big| = \cO_P(d^{-1}\sqrt{n^\dagger \log n})$ and $\sum_{j:j\ne i} E[P_{ij} | \mathcal G] \asymp \frac{1}{d} n^\dagger$ imply that $\max_j (\sum_{k:k\ne j} \tilde P_{jk})^{-1}$ and $\max_j (\sum_{k:k\ne j} P_{jk})^{-1}$ are both $\cO_P(d/n^\dagger)$. Therefore, the first term $\Delta_1$ has the rate of convergence of 
\[
\Delta_1 = \cO_P\bigg( {e^{\bar\kappa}} \frac{n n^{\ddagger}}{d} \bigg( \frac{d}{n^
\dagger}\bigg)^2 \frac{\sqrt{n^\ddagger \log n}}{dn} \frac{n^\dagger}{dn}\bigg) = \cO_P\bigg( {e^{\bar\kappa}} \frac{n^\ddagger \sqrt{n^\ddagger \log n}}{dn n^\dagger} \bigg) \,.
\]
To deal with $\Delta_2$, we need to further expand the expression.
\begin{equation*}
\begin{aligned}
\Delta_2 &= \max_i \bigg| \sum_{j:j\ne i} P_{ji} \frac{\sum_{k:k\ne j} (\tilde P_{kj} - P_{kj}) \pi_k^*}{\sum_{k:k\ne j} \tilde P_{jk}} \bigg| \\
& \lesssim \frac{{ e^{\bar\kappa}}}{nn^\dagger} \max_i \bigg| \sum_{j:j\ne i} P_{ji} \bigg(\sum_{k:k\ne j} \sum_{h \in \mathcal D} 1(i,j,k \in A_h) (Z_{A_h}^j - E[Z
_{A_h}^j]) \bigg) \bigg| \\
& \lesssim \frac{{ e^{\bar\kappa}}}{nn^\dagger} \max_i \bigg| \sum_{j:j\ne i} P_{ji} \bigg(\sum_{h \in \mathcal D} 1(i,j \in A_h) (Z_{A_h}^j - E[Z
_{A_h}^j]) \bigg) \bigg| \,,
\end{aligned}
\end{equation*}
{with probability $1-o(1)$}. In the above, the second inequality is due to $ (\sum_{k:k\ne j} \tilde P_{jk})^{-1} = \cO_P(d/n^\dagger), \forall j$, $\pi^* \asymp {e^{\bar\kappa}}/n$ and $\tilde P_{kj}$ only differs with $P_{kj}$ when $i$ is in the comparison set. And the third inequality is because when $k\ne i$, if $i,j\in A_h$, $A_h$ is at most counted for $M-2$ times in $i,j,k \in A_h$ for different $k \in A_h$, while when $k=i$, we get another copy of the same term. Next notice that $P_{ji}$ and $Z_{A_h}^j$ are independent over both $j$ and $h$ since $P_{ji}$ only involves comparisons ending with winner $i$ and $Z_{A_h}^j$ only involves comparisons ending with winner $j$. So by Hoeffding's inequality, first conditioning on $P_{ji}$, we have
\begin{equation*}
\begin{aligned}
\Delta_2 &= \cO_P\bigg(\frac{{ e^{\bar\kappa}}}{nn^\dagger} \sqrt{\log n \cdot \max_i \sum_{j: j\ne i} \sum_{h\in \mathcal D} 1(i,j\in A_h) P_{ji}^2}\bigg) = \cO_P\bigg(\frac{{e^{\bar\kappa}}}{dnn^\dagger} \sqrt{\log n \cdot n^\dagger (n^\ddagger)^2} \bigg)\,.
\end{aligned}
\end{equation*}
Here we use the fact that $|P_{ji}| \le n^\ddagger / d$. Combining the bounds for $\Delta_1$ and $\Delta_2$, we obtain
\begin{equation*}
\begin{aligned}
\max_i |I_{i,1}| &\lesssim \frac{{ e^{\bar\kappa}} n^\ddagger}{dnn^\dagger} \sqrt{n^\dagger \log n} \,,
\end{aligned}
\end{equation*}
with probability $1-o(1)$.

\medskip
\medskip

Now we bound $\max_i |I_{i,2}|$ for fixed $\mathcal G$. With probability $1-o(1)$, 
\begin{equation*}
\begin{aligned}
\max_i |I_{i,2}| &= \max_i \bigg| \sum_{j:j\ne i} P_{ji} (\bar\pi_j^{(-i)} - \pi_j^*) \bigg|
\\
&\lesssim \frac{d}{n^\dagger} \max_i \bigg| \sum_{j:j\ne i}  \sum_{k:k\ne j} P_{ji} (\tilde P_{kj}\pi_k^* - \tilde P_{jk}\pi_j^*) \bigg| \\
& \lesssim \frac{{e^{-\min_{i\in[n]} \theta_i^*}}}{n n^\dagger} \max_i \bigg| \sum_{j:j\ne i}  \sum_{k:k\ne j} P_{ji} \bigg( \sum_{h\in\mathcal D} 1(j,k\in A_h) (\tilde Z_{A_h}^j e^{\theta_k^*} - \tilde Z_{A_h}^k e^{\theta_j^*}) \bigg) \bigg| \,.
\end{aligned}
\end{equation*}
In the right-hand side of the above bound, when $k=i$, the term is exactly zero because when $i\in A_h$, both $\tilde Z_{A_h}^j e^{\theta_k^*}$ and $\tilde Z_{A_h}^k e^{\theta_j^*}$ are constants in expectation and they cancel out. When $k\ne i$, $P_{ji}$ and $(\tilde Z_{A_h}^j e^{\theta_k^*} - \tilde Z_{A_h}^k e^{\theta_j^*})$ are independent over both $h$ and $j,k$. We separate $1(j,k\in A_h)(\tilde Z_{A_h}^j e^{\theta_k^*} - \tilde Z_{A_h}^k e^{\theta_j^*})$ into two terms $1(j,k\in A_h)(\tilde Z_{A_h}^j e^{\theta_k^*} - E[\tilde Z_{A_h}^j] e^{\theta_k^*}) - 1(j,k\in A_h)(\tilde Z_{A_h}^k e^{\theta_j^*} - E[\tilde Z_{A_h}^k] e^{\theta_j^*})$. For the former term, summing over $k$ makes it $1(j \in A_h)(\tilde Z_{A_h}^j - E[\tilde Z_{A_h}^j]) \sum_{k: k\in A_h, k\ne i,j} e^{\theta_k^*}$, { where $\sum_{k: k\in A_h, k\ne i,j} e^{\theta_k^*} \lesssim e^{\max_{i\in[n]} \theta_i^*} \lesssim e^{\bar\kappa} e^{\min_{i\in[n]} \theta_i^*}$}, we then apply Hoeffding's inequality with average over $h$ and $j$. For the latter term, we sum over $j$ first. It is not hard to see that
\begin{equation*}
\begin{aligned}
\max_i |I_{i,2}| & = \cO_P\bigg(\frac{{ e^{\bar\kappa}}}{nn^\dagger} \sqrt{\log n \cdot \max_i \sum_{j: j\ne i} \sum_{h\in \mathcal D} 1(j\in A_h) P_{ji}^2}\bigg) \\ 
& = \cO_P\bigg(\frac{{e^{\bar\kappa}}}{nn^\dagger} \sqrt{n^\dagger \log n \cdot \max_i \sum_{j: j\ne i} P_{ji}^2}\bigg) \\
& = \cO_P\bigg(\frac{{ e^{\bar\kappa}}}{dnn^\dagger} \sqrt{\log n \cdot n^\ddagger (n^\dagger)^2} \bigg)\,.
\end{aligned}
\end{equation*}
Here the last equation is because $\max_i \sum_{j: j\ne i} P_{ji}^2 =\cO_P(n^\ddagger \max_i \sum_{j: j\ne i} P_{ji} / d) = \cO_P(n^\ddagger n^\dagger / d^2)$.
Combining the bounds for $I_{i,1}$ and $I_{i,2}$, we conclude that with probability $1-o(1)$,
\[
\|R\|_\infty \lesssim \frac{{e^{\bar\kappa}}}{d}\sqrt{\frac{n^\ddagger \log n}{n^2}} =: a_{n,p}\,.
\]

\medskip
\medskip

Next, we need
\begin{equation}
{e^{-\bar\kappa}} b_{n,p} \lesssim \lambda_{\min, \bot}(E[\Omega|\mathcal G \text{ or } \tilde{\mathcal G} ]) \le \lambda_{\max}(E[\Omega|\mathcal G \text{ or } \tilde{\mathcal G} ]) \lesssim { e^{\bar\kappa}} b_{n,p}\,.
\end{equation}
\begin{equation}
\|\Omega - E[\Omega|\mathcal G \text{ or } \tilde{\mathcal G} ]\| = o_P(b_{n,p})\,.
\end{equation}
By Assumption \ref{ass4.2}, we can simply choose $b_{n,p} := \frac{n^\dagger}{dn}$. 
Then, let us find out $c_{n,p}$. We will show $|\bar\delta|=\cO_P({e^{\bar\kappa}} / \sqrt{nn^\dagger})$, that is $c_{n,p} := { e^{\bar\kappa}} / \sqrt{nn^\dagger}$. 
{From \eqref{eq_delta_sp}, it can be shown that $|\bar\delta| = \cO_P( e^{-\min_{i} \theta_i^*} |\sum_i e^{\theta_i^*}\delta_i|)$. Also from \eqref{eq_delta_sp}, }
$\pi_i^*\delta_i = \hat\pi_i - \bar\pi_i$. Since $\sum_{i=1}^n \hat\pi_i = 1$, we have 
\[
|\sum_i \pi_i^*\delta_i| = |\sum_i \bar\pi_i - 1| = |\sum_i (\bar\pi_i - \pi_i^*)| = \cO_P\bigg( \frac{d}{n^\dagger} \bigg| \sum_{i=1}^n \sum_{j:j\ne i} (P_{ji} \pi_j^* - P_{ij}\pi_i^*) \bigg| \bigg)\,.
\]
{Equivalently, $|\sum_i e^{\theta_i^*}\delta_i| = \cO_P( \frac{d}{n^\dagger} \big| \sum_{i=1}^n \sum_{j:j\ne i} (P_{ji} e^{\theta_j^*} - P_{ij} e^{\theta_i^*}) \big|)$.}
Similar to the proof of $I_{i,2}$ without taking the additional max over $i$, again by Hoeffding's inequality, we conclude that 
\[
|\bar\delta| \lesssim \frac{{ e^{\bar\kappa}}}{nn^\dagger} \sqrt{\sum_{i=1}^n \sum_{l\in\mathcal D} 1(i\in A_l)} \lesssim  \frac{{e^{\bar\kappa}} }{\sqrt{nn^\dagger}} \,,
\]
with probability $1-o(1)$.

Finally, we achieve the goal of bounding $\|\delta\|_{\infty}$ as follows. From \eqref{bound_delta_l2}, we have 
\[
\|\delta\|_2 = \cO_P(\sqrt{n}({ e^{\bar\kappa}} a_{n,p}b_{n,p}^{-1}+c_{n,p})) = \cO_P\bigg({ e^{2\bar\kappa}} \sqrt{\frac{n n^\ddagger \log n}{ (n^\dagger)^2}} + { e^{\bar\kappa}} \sqrt{\frac{1}{ n^\dagger}} \bigg).
\]
Then from \eqref{eq_delta_sp}, with probability $1-o(1)$, 
\begin{equation} \label{eq:delta_bound_fixed_graph}
\begin{aligned}
\|\delta\|_\infty & \lesssim \frac{d\, { e^{-\min_i\theta_i^*}}}{n^\dagger} \bigg(\|\delta\|_2 \sqrt{\sum_{j: j\ne i} P_{ji}^2 e^{2\theta_j^*}} + \|(\sum_i e^{\theta_i^*}) R\|_\infty\bigg)  \\ 
& \lesssim \frac{d}{n^\dagger} \bigg( { e^{\bar\kappa}} \bigg({ e^{2\bar\kappa}}\sqrt{\frac{n n^\ddagger \log n}{ (n^\dagger)^2}} + { e^{\bar\kappa}}\sqrt{\frac{1}{ n^\dagger}}\bigg) \sqrt{\frac{1}{d^2} (n^{\ddagger} n^\dagger)} + \frac{{ e^{\bar\kappa}}}{d} \sqrt{n^\ddagger\log n}\bigg) \\
& \lesssim { e^{3\bar\kappa}} \sqrt{\frac{n (n^\ddagger)^2\log n}{(n^\dagger)^3}} + { e^{2\bar\kappa}} \sqrt{\frac{n^\ddagger \log n}{(n^\dagger)^2}} =: d_{n,p} \,.
\end{aligned}
\end{equation}
We have $d_{n,p} = o(1/\sqrt{n^\dagger})$ if  ${ e^{3\bar\kappa}} n^\ddagger n^{1/2} (\log n)^{1/2} = o(n^\dagger)$ {and $e^{2\bar\kappa} \log n = o(n)$}.

Recall $\frac{\hat\pi_i - \pi_i^*}{\pi_i^*} = \frac{\bar\pi_i - \pi_i^*}{\pi_i^*} + \delta_i$, and we have shown $\delta_i$ is a small term negligible for deriving the asymptotic distribution of $\frac{\hat\pi_i - \pi_i^*}{\pi_i^*}$, which is entirely driven by $\frac{\bar\pi_i - \pi_i^*}{\pi_i^*}$.   

Finally, we go from the inference of $\hat\pi_i$ to the inference of $\tilde\theta_i$. Recall $\tilde\theta_i = \log \hat\pi_i - \frac{1}{n} \sum_{k=1}^n \log \hat\pi_k$ and $\log \pi_i^* = \theta_i^* - \log \sum_{k=1}^n e^{\theta_k^*}$, which together with our assumption that $1^\top \theta^* = 0$ gives $\theta_i^* = \log \pi_i^* - \frac{1}{n} \sum_{k=1}^n \log \pi_k^*$. Therefore
\[
\tilde\theta_i - \theta_i^* =  \log \frac{\hat\pi_i}{\pi_i^*} - \frac{1}{n} \sum_{k=1}^n \log \frac{\hat\pi_k}{\pi_k^*} \,.
\]
By Delta method, $\log \frac{\hat\pi_i}{\pi_i^*}$ has the same asymptotic distribution as $\frac{\hat\pi_i - \pi_i^*}{\pi_i^*}$. It suffices to show that the additional term $\frac{1}{n} \sum_{k=1}^n \log \frac{\hat\pi_k}{\pi_k^*} = o_P(1/\sqrt{n^\dagger})$. By Taylor's expansion, we only need to show $\frac{1}{n} \sum_{k=1}^n \frac{\hat\pi_k-\pi_k^*}{\pi_k^*} = \frac{1}{n} \sum_{k=1}^n \frac{\bar\pi_k-\pi_k^*}{\pi_k^*} + \bar\delta = \cO_P(\frac{1}{n} \sum_{k=1}^n J_k^* + \bar\delta) = o_P(1/\sqrt{n^\dagger})$. From the above, we already derived that $\bar\delta = \cO_P({ e^{\bar\kappa}} / \sqrt{n n^\dagger}) = o_P(1/\sqrt{n^\dagger})$ {when $e^{2\bar\kappa} \log n = o(n)$}. So it only remains to prove $\bar J:= \frac{1}{n} \sum_{k=1}^n J_k^* = o_P(1/\sqrt{n^\dagger})$. In fact, this proof is almost identical to that of $\bar \delta$. 
\[
|\bar J| = \bigg| \frac{1}{n} \sum_{i=1}^n \frac{\sum_{j:j\ne i} (P_{ji} \pi_j^* - P_{ij}\pi_i^*)}{\sum_{j:j\ne i} \EE[P_{ij}|\mathcal G \text{ or } \tilde{\mathcal G} ] \pi_i^*} \bigg| = \cO_P\bigg( \frac{d \, {e^{-\min_i\theta_i^*}}}{n^\dagger} \bigg| \sum_{i=1}^n \sum_{j:j\ne i} (P_{ji} e^{\theta_j^*} - P_{ij} e^{\theta_i^*}) \bigg| \bigg) = \cO_P\bigg( \frac{{ e^{\bar\kappa}}}{\sqrt{nn^\dagger}} \bigg)\,,
\]
which is indeed $o_P(1/\sqrt{n^\dagger})$.
So we conclude that $\tilde\theta_i - \theta_i^*$ has the same asymptotic distribution as $\frac{\bar\pi_i - \pi_i^*}{\pi_i^*}$ and $J_i^*$ for all $i \in [n]$.
\qed

\subsection{Proof of Theorem \ref{thm_consist_normality}}

From the proof of Theorem \ref{thm_approx}, we know for all $i\in[n]$, $\tilde{\theta}_i-\theta_i^*=J_i^*+o_P(1/\sqrt{n^\dagger})$.
According to the definition of $J_i^*,$ given in \eqref{approx2}. For the denominator, we have
${\sum_{j\neq i}}\EE[P_{ij}\given \cG]e^{\theta_i^*} \asymp e^{\theta_i^*} (n^\dagger / d)$ { from Assumption \ref{ass4.1}}. For the numerator, it can be written into an independent sum over $\cD$. By Bernstein inequality, we obtain the order of its maximum to be ${ e^{\max_i\theta_i^*}} \sqrt{n^{\dagger}\log n}/d$. Therefore, by combining the orders of the numerator and denominator, we have
$\|\tilde{\theta}-\theta^*\|_{\infty}\lesssim {e^{\bar\kappa}} \sqrt{\frac{\log n}{n^\dagger}}$ with probability $1-o(1)$. This concludes the proof of \eqref{l_infty_consist}.

Also since the numerator is an independent sum, by central limit theorem, it is asymptotically normal. It is easy to verify the asymptotic variance leads to the form of $\rho_i(\theta^*)$ in \eqref{eq_rho_theta}. Next we prove that the numerator and denominator of $\rho_i(\theta)$ can be calculated with any consistent estimator $\hat\theta$. To be specific, for the denominator, we need to show
\begin{align*}
    \sum\limits_{l \in \mathcal D} 1(i \in A_{l}) \bigg(\frac{\sum_{u\in A_l} e^{\hat \theta_u} - e^{\hat \theta_i}}{f(A_l)}\bigg) \frac{e^{\hat \theta_i}}{f(A_l)}\bigg{/} \bigg[\sum\limits_{l \in \mathcal D} 1(i \in A_{l}) \bigg(\frac{\sum_{u\in A_l}  e^{\theta_u^*} - e^{\theta_i^*}}{f(A_l)}\bigg) \frac{ e^{\theta_i^*}}{f(A_l)}\bigg]\rightarrow 1.
\end{align*}
When we have $\theta_i^*$ being bounded for all $i\in[n],$ and $\hat \theta_i-\theta_i^*=o_P(1)$, it is straightforward to verify that 
\begin{align*}
    \frac{\bigg(\sum\limits_{l \in \mathcal D} 1(i \in A_{l}) \bigg(\frac{\sum_{u\in A_l} e^{\hat \theta_u} - e^{\hat \theta_i}}{f(A_l)}\bigg) \frac{e^{\hat \theta_i}}{f(A_l)}-\sum\limits_{l \in \mathcal D} 1(i \in A_{l}) \bigg(\frac{\sum_{u\in A_l}  e^{\theta_u^*} - e^{\theta_i^*}}{f(A_l)}\bigg) \frac{ e^{\theta_i^*}}{f(A_l)}\bigg)}{\bigg[\sum\limits_{l \in \mathcal D} 1(i \in A_{l}) \bigg(\frac{\sum_{u\in A_l}  e^{\theta_u^*} - e^{\theta_i^*}}{f(A_l)}\bigg) \frac{ e^{\theta_i^*}}{f(A_l)}\bigg]}\rightarrow 0,
\end{align*}
by using Taylor's expansion.
The numerator component of $\rho_i(\theta)$ can be proved similarly. Therefore, we conclude that by plugging in any consistent estimator of $\theta_i^*,i\in[n],$ we obtain consistent estimator for the asymptotic variance.
\qed

\newpage

\section{Proofs for the PL Model with Random Graph (Section \ref{sec:rand_graph})}
\subsection{Proof of Theorem \ref{thm_b_rand}}
We hope to show 
\begin{equation}
\label{bound_b_1}
{e^{-\bar\kappa}} b_{n,p} \lesssim \lambda_{\min, \bot}(E[\Omega|\tilde{\mathcal G} ]) \le \lambda_{\max}(E[\Omega|\tilde{\mathcal G} ]) \lesssim {e^{\bar\kappa}} b_{n,p}\,,
\end{equation}
\begin{equation}
\label{bound_b_2}
\|\Omega - E[\Omega|\tilde{\mathcal G} ]\| = o_P(b_{n,p})\,,
\end{equation}
hold for $b_{n,p} = \frac{n^\dagger}{dn}$. We will only focus on conditioning on $\tilde{\mathcal G}$ with $M=3$. The case with general $M$ can be shown similarly. 

Let us first look at \eqref{bound_b_1} conditioning on $\tilde{\mathcal G}$ with $M=3$. Let $\Omega^* = E[\Omega|\tilde{\mathcal G}]$. $\Omega_{ij}^* = -P_{ji}^* \pi_j^* = -P_{ij}^* \pi_i^* = \Omega_{ji}^*$ and $\Omega_{ii}^* = \sum_{j:j\ne i} P_{ij}^* \pi_i^*$, where recall that $P_{ij}^* = E[P_{ij} | \tilde{\mathcal G}] = \frac{L}{d} \sum_{k:k\ne j,i} \tilde A_{ijk} E[Z_{ijk}^l]$. So,
\[
{e^{-\bar\kappa}} \frac{L}{dn} v^\top L_{\tilde{\mathcal G}} v \lesssim v^\top \Omega^* v = \sum_{i \ne j} \frac{1}{2} (v_i-v_j)^2 \Omega_{ij}^* \lesssim { e^{\bar\kappa}} \frac{L}{dn} v^\top L_{\tilde{\mathcal G}} v \,,
\]
with probability $1-o(1)$ where $(L_{\tilde{\mathcal G}})_{ij} = -\sum_{k:k\ne j,i} \tilde A_{ijk}$ and $(L_{\tilde{\mathcal G}})_{ii} = \sum_{j:j\ne i} \sum_{k:k\ne j,i} \tilde A_{ijk}$. 
It is easy to see all eigenvalues of $E[L_{\tilde{\mathcal G}}]$ are of the order $n^2 p \asymp n^{\dagger}/L$, except for the zero minimal eigenvalue. Again, here $E$ is with respect to the randomness in comparison sets. So we only need to show 
\[
\max_j |\lambda_j(L_{\tilde{\mathcal G}}) - \lambda_j(E[L_{\tilde{\mathcal G}}])| \le \|L_{\tilde{\mathcal G}} - E[L_{\tilde{\mathcal G}}]\| = o_P(n^2 p)\,.
\]
The first inequality is by Weyl's inequality. To prove the second bound, write $L_{\tilde{\mathcal G}} = \sum_{i < j<k} \Delta_{ijk}$ where
\[
\Delta_{ijk} = \tilde A_{ijk} \Big[2(e_i e_i^\top+e_j e_j^\top+e_k e_k^\top) - (e_i e_j^\top+e_j e_i^\top+e_i e_k^\top + e_k e_i^\top+e_j e_k^\top+e_k e_j^\top) \Big]
\]
$\Delta_{ijk}$'s are symmetric and independent. Hence, by matrix Bernstein inequality (Theorem 1.4 of \cite{tropp2012user}), we have
\[
P\Big(\|L_{\tilde{\mathcal G}} - E[L_{\tilde{\mathcal G}}] \ge t\Big) \le n \exp\Big(-\frac{t^2/2}{\sigma^2 + Mt/3}\Big)\,,
\]
where $\max_{i,j,k} \|\Delta_{ijk}\| \le M = O(1)$ and $\sigma^2 = \|\sum_{i<j<k} E[\Delta_{ijk}^2]\| = O(n^2 p)$. Note that although the summation over $i,j,k$ sums over 3 indices, for diagonal summations only 2 indices have nonzero elements,  while for off-diagonal summations, only 1 index have nonzero elements. Therefore, we have shown that
\[
\|L_{\tilde{\mathcal G}} - E[L_{\tilde{\mathcal G}}]\| = \cO_P\bigg(\sqrt{n^2 p \log n} + \log n\bigg)\,,
\]
which is $o_P(n^2 p)$ if $n^2 p \gg \log n$. 

Now let us verify \eqref{bound_b_2} conditioning on $\tilde{\mathcal G}$, when $n^\dagger \gg \log n$. We will still only consider $M=3$ for simplicity. Plugging in the definition of $P_{ij}$ into $\Omega$, we have
\begin{equation*}
\begin{aligned}
\Omega_{ij} &= -\frac{1}{d} \sum_{l=1}^L \sum_{k: k\ne j,i} \tilde A_{ijk} Z_{ijk}^l \pi_j^* \,, \\
\Omega_{ii} &= \frac{1}{d} \sum_{l=1}^L \sum_{j < k: j,k\ne i} \tilde A_{ijk} (Z_{ijk}^l + Z_{ikj}^l) \pi_i^* \,.
\end{aligned}
\end{equation*}
We can write $\Omega = \sum_{l=1}^L \sum_{i<j<k} S_{ijk}^l$ where 
\begin{equation*}
\begin{aligned}
S_{ijk}^l = d^{-1}\tilde A_{jik} \Big[
&-Z_{ijk}^l \pi_j^* e_i e_j^\top + Z_{ijk}^l \pi_i^* e_i e_i^\top -Z_{kij}^l \pi_k^* e_i e_k^\top + Z_{ikj}^l \pi_i^* e_i e_i^\top \\
&-Z_{ijk}^l \pi_i^* e_j e_i^\top + Z_{jik}^l \pi_j^* e_j e_j^\top -Z_{kji}^l \pi_k^* e_j e_k^\top + Z_{jki}^l \pi_j^* e_j e_j^\top \\
&-Z_{ikj}^l \pi_i^* e_k e_i^\top + Z_{kij}^l \pi_k^* e_k e_k^\top -Z_{jki}^l \pi_j^* e_k e_j^\top + Z_{kji}^l \pi_k^* e_k e_k^\top \Big]
\end{aligned}
\end{equation*}
{ $S_{ijk}^l$ is independent for $i,j,k,l$ but asymmetric.} We leverage the Bernstein inequality for asymmetric matrices (Theorem 1.6 of \cite{tropp2012user}), we have
\[
P\Big(\|\Omega - \Omega^* \|\ge t | \tilde{\mathcal G}\Big) \le 2n 
\exp\Big(-\frac{t^2/2}{\sigma^2 + Mt/3}\Big)\,,
\]
where $$\sigma^2 = \max\{\|\sum_l \sum_{i<j<k} E[S_{ijk}^l (S_{ijk}^l)^\top |\tilde{\mathcal G}]\|, \|\sum_l \sum_{i<j<k} E[(S_{ijk}^l)^\top S_{ijk}^l |\tilde{\mathcal G}]\|$$ 
and $\max_{i,j,k,l} \|S_{ijk}^l\| \le M = O({ e^{\bar\kappa}}(dn)^{-1})$. It is not hard to see $\sigma^2 = O(L { e^{2\bar\kappa}} \cdot \max_{i} \sum_{j < k: j,k\ne i} \tilde A_{ijk} / (dn)^2) = O({e^{2\bar\kappa}} n^\dagger / (dn)^2)$, where rate ${e^{2\bar\kappa}}/n^2$ comes from the order of $\pi_i^*$ and rate $n^\dagger$ comes from summer over $M-1$ indices of the $M$ items being compared. So we get
\[
\|\Omega - E[\Omega|\tilde{\mathcal G}]\| = \cO_P\bigg({e^{\bar\kappa}} (dn)^{-1}(\sqrt{n^\dagger \log n} + \log n)\bigg)\,,
\]
which is $o_P(n^\dagger / dn)$ if $n^\dagger \gg {e^{2\bar\kappa}} \log n$. So we can indeed choose $b_{n,p} = n^\dagger / dn$ to make \eqref{bound_b_1} and \eqref{bound_b_2} hold.  
\qed

\subsection{Proof of Theorem \ref{thm_consist_normality_rand}}

In the case of random comparison graph i.e. conditioning on $\tilde {\mathcal G}$ and $M=3$, recall that 
\[
P_{jk} = \frac{1}{d} \sum_{l=1}^L \sum\limits_{s: s \ne j,k} \tilde A_{jks} Z_{{jks}}^l\,,
\]
where we have $Z_{ijk}^l = 1(A_l=\{i,j\}, c_l=j) / f(\{i,j\}) + 1(A_l=\{i,j,k\}, c_l=j) / f(\{i,j,k\})$. In addition, we have the leave-one-out version of $P_{jk}$, in which we only replace any randomness related to the comparison of item $i$ with its expectation and denote the new transition probability as $\tilde P_{jk}$.

As we discussed after Assumption \ref{ass4.3} in Section \ref{sec:rand_graph}, $n^\dagger \asymp \binom{n-1}{M-1}pL$ and $n^\ddagger \asymp (\binom{n-2}{M-2}p + \log n)L$. If we are in a dense regime where $\binom{n-2}{M-2}p \ge C \log n$, then the first term dominates $n^\ddagger$ and we have $n^\dagger \asymp n n^\ddagger$. So Assumption \ref{ass4.1} holds. We can exactly follow the procedures in fixed graph proof to show the theorem. The nontrivial proof is really for the case of more sparse graph where $\binom{n-2}{M-2}p \le C \log n$ so that $n^\ddagger \lesssim L \log n$. Next we will only work on this sparse graph case,  i.e. $n^{\ddagger}\lesssim L\log n$ and $n^\dagger\gtrsim L\textrm{poly}(\log n)$, which  makes random graph achieve the weaker assumption of Assumption \ref{ass4.3}.

The inequality that makes the proof for the fixed graph case invalid for the random graph case is the first inequality bound in \eqref{eq:delta_bound_fixed_graph}, where because $\delta_j$ and $P_{ji}$ in the numerator of the first term of \eqref{eq_delta_sp} are not independent, we can only bound this using Cauchy-Schwarz in \eqref{eq:delta_bound_fixed_graph}. But if we define the leave-one-out version of $\hat\pi$ and $\delta$ corresponding to $\tilde P$ rather than $P$, denoted as $\hat\pi^{(-i)}$ and $\delta^{(-i)}$, we can have independence between $\delta_j^{(-i)}$ and $P_{ji}$. 

In the scenario of $n^{\ddagger}\lesssim L\log n$ and $n^\dagger\gtrsim L\textrm{poly}(\log n)$, i.e. $\textrm{poly}(\log n) n^\ddagger / n^\dagger = o(1)$, we first obtain
\begin{equation} \label{eq_delta_sp2}
\begin{aligned}
\delta_i &= \frac{\hat\pi_i - \bar\pi_i}{\pi_i^*} = \frac{\sum_{j:j\ne i} P_{ji} (\hat\pi_j - \pi_j^*)}{\pi_i^* \sum_{j:j\ne i} P_{ij}} \\
&= \frac{\sum_{j:j\ne i} P_{ji} (\hat\pi_j - \hat \pi_j^{(-i)})}{\pi_i^* \sum_{j:j\ne i} P_{ij}}+\frac{\sum_{j:j\ne i} P_{ji} \pi_j^*\delta_j^{(-i)}}{\pi_i^* \sum_{j:j\ne i} P_{ij}} + \frac{\sum_{j:j\ne i} P_{ji} (\bar\pi_j^{(-i)} - \pi_j^*)}{\pi_i^* \sum_{j:j\ne i} P_{ij}} \,.
\end{aligned}
\end{equation}
where $\delta_j^{(-i)}$ is the leave-one-out version of $\delta_j$ and we also define $R^{(-i)}_j:= \sum_{j:j\ne i} P_{ji} (\bar\pi_j^{(-i)} - \pi_j^*).$ 

\begin{lemma} \label{lem:pi_perturb_leave_one_out} Under the conditions of Theorem \ref{thm_consist_normality_rand}, we have 
\[
\|\hat\pi_j-\hat\pi_j^{(-i)}\|_2 \lesssim {e^{\bar\kappa}} \frac{\sqrt{\textrm{poly}(\log n) n^\dagger}}{nd} \asymp { e^{\bar\kappa}} \frac{\sqrt{\textrm{poly}(\log n)}}{n\sqrt{n^\dagger}}\,.
\]
with probability $1-o(1)$.
\end{lemma}

We first take care of the first term and show it is of negligible order. Specifically, we have \begin{align*}
\frac{\sum_{j:j\ne i} P_{ji} (\hat\pi_j - \hat \pi_j^{(-i)})}{\pi_i^* \sum_{j:j\ne i} P_{ij}}&\lesssim \frac{nd}{n^\dagger}(\sum_{j}P_{ji}^2)^{1/2}\|\hat\pi_j-\hat\pi_j^{(-i)}\|_2\\
&\lesssim {e^{\bar\kappa}} \frac{nd}{n^\dagger}\cdot \sqrt{n^\dagger n^\ddagger}/d \cdot \frac{\sqrt{\textrm{poly}(\log n)}}{n\sqrt{n^{\dagger}}}\\
&\lesssim {e^{\bar\kappa}} \frac{\sqrt{\textrm{poly}(\log n) n^\ddagger}}{n^\dagger}=o\bigg(\frac{1}{\sqrt{n^\dagger}}\bigg),
\end{align*} 
with probability $1-o(1)$.

Next, we upper bound $\|\delta\|_{\infty}$ by the second and third terms of \eqref{eq_delta_sp2}. According to Bernstein's inequality (conditioning on $\delta^{(-i)}$), we obtain
\begin{align}
    \|\delta\|_\infty & \lesssim \frac{d {e^{\bar\kappa}}}{n^\dagger}  \bigg( \max_{i,j}\EE[P_{ji}] \sqrt{n}\|\delta^{(-i)}\|_2 + \sqrt{\log n \cdot \max_{i,j}\textrm{Var}(P_{ji})\|\delta^{(-i)}\|_2^2} + \frac{\log n}{d} \bigg) + \frac{d n}{n^\dagger} \|R^{(-i)}\|_\infty \nonumber\\
& \lesssim \frac{d {e^{\bar\kappa}}}{n^\dagger} \bigg( \frac{(n^\dagger/n) \sqrt{n}\|\delta^{(-i)}\|_2}{d}+ \frac{\sqrt{(n^\dagger / n) \log n}\|\delta^{(-i)}\|_{2}}{d} + \frac{\log n}{d}\bigg) + \frac{dn}{n^\dagger} \|R^{(-i)}\|_\infty\,, \label{delta_bound2}
\end{align}
with probability $1-o(1)$, where $\EE[P_{ji}]$ and Var$(P_{ji})$, by definition of the random graph, are of orders $\binom{n-2}{M-2}p L / d \asymp n^{\dagger}/(dn)$ and $\binom{n-2}{M-2}p L / d^2 \asymp n^{\dagger}/(d^2n)$ respectively. 

The upper bound of $\|\delta^{(-i)}\|_{2}$ can be similarly achieved like $\|\delta\|_2$ (can be proved using similar techniques in the previous section) with order 
\[
\|\delta^{(-i)}\|_{2} = \cO_P\bigg( { e^{2\bar\kappa}} \sqrt{\frac{n n^\ddagger \log n}{ (n^\dagger)^2}} + {e^{\bar\kappa}} \sqrt{\frac{1}{ n^\dagger}}\bigg)\,.
\]
And the upper bound of $\|R^{(-i)}\|_{\infty}$ can also be achieved similarly with $\|R\|_{\infty}$, with order 
\[
\|R^{(-i)}\|_{\infty} = \cO_P\bigg(\frac{{ e^{\bar\kappa}}}{d}\sqrt{\frac{n^\ddagger \log n}{n^2}}\bigg)\,.
\]
Plugging these bounds into \eqref{delta_bound2}, we only need {$e^{4\bar\kappa} = o(n)$, } $n^{\ddagger}\lesssim L\log n$ and $n^\dagger\gtrsim {e^{6\bar\kappa}} L\textrm{poly}(\log n)$ in order to obtain $\|\delta\|_{\infty}=o_P(1/\sqrt{n^\dagger}).$

Once we have shown the uniform approximation, similar to the proof of Theorem \ref{thm_consist_normality}, the rate of convergence and asymptotic normality follow from the definition of $J_i^*$. Now the proof is complete. 
\qed

\subsection{Proof of Lemma \ref{lem:pi_perturb_leave_one_out}}

To prove Lemma \ref{lem:pi_perturb_leave_one_out}, we need to apply the following perturbation bound given by Lemma \ref{lem:perturbation} for $\hat\pi$, the stationary distribution of the Markov chain with transition probability $P$. 

\begin{lemma}[{Theorem 5.1 of \cite{chen2019spectral}, eigenvector perturbation}] 
\label{lem:perturbation}
Suppose that $P$, $\tilde P$ and $P^*$ are probability transition matrices with stationary distributions $\hat\pi$, $\tilde\pi$ and $\pi^*$ respectively. Also assume that $P^*$ represents a reversible Markov chain. When $\|P - \tilde P\|_2 < 1 - \max\{\lambda_2(P^*), -\lambda_n(P^*)\}$, it holds that
\[
\|\hat\pi - \tilde\pi\|_2 \le C \frac{\|\hat\pi^\top (P - \tilde P)\|_2}{1 - \max\{\lambda_2(P^*), -\lambda_n(P^*)\} - \|P - \tilde P\|_2} \,.
\]
\end{lemma}

\medskip
\medskip
We will use this result to prove Lemma \ref{lem:pi_perturb_leave_one_out}. Choose $\tilde P$ as our leave-one-out version of $P$, then $\tilde\pi = \hat\pi^{(-i)}$. 
Therefore, we only need to verify two things, firstly $\|\hat\pi^\top (P - \tilde P)\|_2 = \cO_P(\textrm{poly}({e^{\kappa}}\log n) \sqrt{n^\dagger}/(dn))$ gives the desired rate of convergence, and secondly $\|P - \tilde P\|_2 = o_P(1)$. 

Note that $P_{jk} - \tilde P_{jk}$ only involves terms that compare item $i$. We have
\begin{equation*}
\begin{aligned}
\|\hat\pi^\top (P - \tilde P)\|_2 & \lesssim \frac{{e^{\kappa}}}{dn} \bigg| \sum_{m=1}^n \sum_{k:k\ne i,m}  \sum_{l=1}^L \tilde A_{mik} (Z_{mik}^l - E[Z_{mik}^l]) \bigg| + \frac{{e^{\kappa}}}{n} \sqrt{\sum_{j:j\ne i} (P_{jj} - \tilde P_{jj})^2} \\
& \quad\quad + \frac{{e^{\kappa}}}{dn} \sqrt{\sum_{j: j\ne i} \bigg( \sum_{m: m\ne i,j} \sum_{l=1}^L \tilde A_{mji} (Z_{mji}^l - E[Z_{mji}^l]) + \sum_{k:k\ne i,j} \sum_{l=1}^L \tilde A_{ijk} (Z_{ijk}^l - E[Z_{ijk}^l]) \bigg)^2} \\
& =: H_1 + H_2 + H_3\,.
\end{aligned}
\end{equation*}

It is easy to see that $H_1 = { e^{\bar\kappa}}\sqrt{n^\dagger} / (dn)$ by Bernstein inequality directly. Then note that $P_{jj} - \tilde P_{jj} = \sum_{k: k\ne j} (P_{jk} - \tilde P_{jk})$, which actually makes $H_2$ the same as $H_3$. So we only need to find out the order of $H_3$. Define $Z_{mji} = \sum_{l=1}^L Z_{mji}^l$, we have with probability $1-o(1)$,
\begin{equation*}
\begin{aligned}
H_3 &\asymp \frac{{e^{\kappa}}}{dn} \sqrt{\sum_{j: j\ne i} \bigg( \sum_{m: m\ne i,j} \tilde A_{mji} (Z_{mji} - E[Z_{mji}]) \bigg)^2} \\
&\lesssim \frac{{e^{\kappa}}}{dn} \sqrt{\sum_{j: j\ne i} \sum_{m: m\ne i,j} \tilde A_{mji} (Z_{mji} - E[Z_{mji}])^2} \\
& \quad\quad + \frac{{ e^{\kappa}}}{dn} \sqrt{\sum_{j: j\ne i} \sum_{m_1: m_1 \ne i,j} \sum_{m_2: m_2\ne i,j,m_1} \tilde A_{m_1 ji} \tilde A_{m_2 ji} |Z_{m_1 ji} - E[Z_{m_1 ji}]| |Z_{m_2 ji} - E[Z_{m_2 ji}]|} \\
& \lesssim \frac{{ e^{\kappa}}\sqrt{n^\dagger}}{dn} + \frac{{ e^{\kappa}}\sqrt{\textrm{poly}(\log n) n^\dagger}}{dn} \,,
\end{aligned}
\end{equation*}
where the last inequality is because $\max_{m,j,i} |Z_{mji} - E[Z_{mji}]| = \cO_P(\sqrt{L \log n})$, $\sum_{m_2: m_2\ne i,j,m_1} \tilde A_{m_2 ji} = \cO_P(\log n)$ given that $p \le C \log n/\binom{n-2}{M-2}$, and $L \sum_{j: j\ne i} \sum_{m_1: m_1 \ne i,j} \tilde A_{m_1 ji} = \cO_P(n^\dagger)$. Therefore, combining the rates for $H_1, H_2, H_3$, we have shown that $\|\hat\pi^\top (P - \tilde P)\|_2 = \cO_P({ e^{\kappa}}\textrm{poly}(\log n) \sqrt{n^\dagger}/(dn))$.

Next, we show $\|P - \tilde P\|_2 \le \|P - \tilde P\|_F = o_P(\sqrt{n^\dagger} / d) = o_P(1)$. Actually the proof is similar to the above bound for $\|\hat\pi^\top (P - \tilde P)\|_2$, except for the rate $1/n$ from $\hat\pi$. $\|P - \tilde P\|_F^2 = \sum_{j=1}^n \sum_{m=1}^n (P_{mj} - \tilde P_{mj})^2$. We separate the summations over $j,m$ into 3 cases of $m=j$, $m \ne j$ but $m=i$ or $j=i$, and $m \ne j, m,j \ne i$. In the case of $m=j$, we bound $\sqrt{\sum_{j=1}^n (P_{jj} - \tilde P_{jj})^2}$, which is essentially the same as bounding the second case. In the second case of $m \ne j$ but $m=i$ or $j=i$, we bound $\sqrt{\sum_{j: j\ne i} (P_{ij} - \tilde P_{ij})^2 + \sum_{m:m\ne i} (P_{mi} - \tilde P_{mi})}$. Following the bound for $H_3$ above, this term is of order $\sqrt{n^\dagger} / d$. Finally in the third case $m \ne j, m,j \ne i$, we need to bound $\sqrt{\sum_{j: j\ne i}  \sum_{m:m\ne j,i} (P_{mj} - \tilde P_{mj})^2}$, where $P_{mj} - \tilde P_{mj} =\frac{1}{d} \tilde A_{mji} (Z_{mji} - E[Z_{mji}])$. Again following the proof of $H_3$, this term is also $\sqrt{n^\dagger}/d$. Therefore, we conclude that $\|P - \tilde P\|_2 \le \|P - \tilde P\|_F = o_P(1)$.
\qed

\newpage

\section{Proofs for Ranking Inference (Section \ref{sec4.3})}
For simplicity of notations, we write $V_{ijl} = 1(i, j \in A_{l}) 1(c_{l} = j) /f(A_{l})$ for any $i\neq j$ and $l \in \mathcal{D}$. For each $m \in [n]$, define 
\begin{align*}
    \xi_{m}^{*} = \sum_{j : j\neq m} (P_{jm} e^{\theta_{j}^{*}} - P_{mj} e^{\theta_{m}^{*}}) \enspace \mathrm{and} \enspace \nu_{m}^{*} = \sum_{j : j\neq m} \mathbb{E}[P_{mj}|\mathcal{G}] e^{\theta_{m}^{*}}.  
\end{align*}
Note that $\nu_m^*,m\in[n]$ has the same meaning as $\tau_m(\theta^*)$ defined in the main text. However, for simplicity of notation, we keep $\nu_m^*$ in the proof of this section.

Then $J_{m}^{*} = \xi_{m}^{*}/\nu_{m}^{*}$ for each $m \in [n]$ in view of~\eqref{approx2} and 
\begin{align}
\label{eq_xi_population_version}
    \xi_{m}^{*} = \frac{1}{d} \sum_{l \in \mathcal{D}} \sum_{j: j\neq m} (V_{jml} e^{\theta_{j}^{*}} - V_{mjl} e^{\theta_{m}^{*}}) =: \frac{1}{d} \sum_{l \in \mathcal{D}} \mathcal{V}_{ml},
\end{align}
where $\mathcal{V}_{ml} = \sum_{j : j \neq m} (V_{jml} e^{\theta_{j}^{*}} - V_{mjl} e^{\theta_{m}^{*}})$. Under Assumption~\ref{Assumption_dynamic_range_bound}, it follows that for all $m \in [n]$, with probability $1-o(1)$,
\begin{align}
\label{eq_tau_m_order}
    \frac{\nu_{m}^{*}}{e^{\theta_{m}^{*}}} \asymp \frac{1}{d} \sum_{l \in \mathcal{D}} 1(m \in A_{l}) \asymp \frac{n^{\dagger}}{d}. 
\end{align}


\subsection{Gaussian Approximation}
Our goal of this section is to derive the asymptotic distribution of $T_{\mathcal{M}}$. Recall $\sigma_{km}^{2} = \Var(J_{k}^{*} - J_{m}^{*}|\mathcal{G})$ for each $k \neq m \in [n]$. Define 
\begin{align*}
    \mathcal{J}_{\mathcal{M}}^{*} = \max_{m \in \mathcal{M}} \max_{k \neq m} \left|\frac{J_{k}^{*} - J_{m}^{*}}{\sigma_{km}}\right|.
\end{align*}
We first derive a Gaussian approximation result~\citep{CCK2017} for $\mathcal{J}_{\mathcal{M}}^{*}$. Let $Z = (Z_{km})_{m \in [n], k\neq m} \in \mathbb{R}^{n(n - 1)}$ be a centered Gaussian random vector such that
\begin{align*}
    \Cov(Z_{km}, Z_{k'm'}) = \frac{\Cov(J_{k}^{*} - J_{m}^{*}, J_{k'}^{*} - J_{m'}^{*})}{\sigma_{km} \sigma_{k'm'}}, 
\end{align*}
for all $k\neq m \in [n]$ and $k'\neq m' \in [n]$. Then the Gaussian analogue of $\mathcal{J}_{\mathcal{M}}^{*}$ is defined by $\max_{m \in \mathcal{M}} \max_{k \neq m} |Z_{km}|$. In the following lemma, we establish an upper bound for the Kolmogorov distance between the distribution functions of $\mathcal{J}_{\mathcal{M}}^{*}$ and its Gaussian analogue. 
\begin{lemma}
\label{Lemma_Gaussian_approximation}
Under the conditions of Theorem~\ref{thm_approx}, we have
{
\begin{align}
\label{eq_GA_J_M}
    \sup_{z \geq 0} \left|\mathbb{P}(\mathcal{J}_{\mathcal{M}}^{*} \leq z|\mathcal{G}) - \mathbb{P}\left(\max_{m \in \mathcal{M}} \max_{k \neq m} |Z_{km}|\leq z\right)\right| \leq \frac{C e^{3\bar{\kappa}/4}\log^{5/4}(n |\mathcal{D}|)}{(n^{\dagger})^{1/4}}, 
\end{align}
}
where $C < \infty$ is a positive constant independent of $n$. 
\end{lemma}


\noindent\textbf{Proof of Lemma~\ref{Lemma_Gaussian_approximation}.}
Elementary calculation implies that for any $m \in [n]$,  
\begin{align*}
    \Var(J_{m}^{*}|\mathcal{G}) = \frac{e^{\theta_{m}^{*}}}{(\nu_{m}^{*} d)^{2}} \sum_{l \in \mathcal{D}} \frac{1(m \in A_{l})}{f^{2}(A_{l})} \left(\sum_{j \in A_{l}} e^{\theta_{j}^{*}} - e^{\theta_{m}^{*}}\right), 
\end{align*}
and for any $k \neq m$, 
\begin{align*}
    \Cov(J_{k}^{*}, J_{m}^{*}|\mathcal{G}) = -\frac{1}{d^{2} \nu_{m}^{*} \nu_{k}^{*}} \sum_{l \in \mathcal{D}} \frac{1(k, m \in A_{l})}{f^{2}(A_{l})} \frac{e^{\theta_{k}^{*} + \theta_{m}^{*}} (e^{\theta_{k}^{*}} + e^{\theta_{m}^{*}})}{\sum_{u \in A_{l}} e^{\theta_{u}^{*}}}. 
\end{align*} 
Since $f(A_{l})$ is uniformly bounded from below and above for $l \in \mathcal{D}$, it follows from~\eqref{eq_tau_m_order} that 
{
\begin{align*}
    \Var(J_{m}^{*}|\mathcal{G}) &\asymp \frac{e^{\theta_{m}^{*}}}{(n^{\dagger} e^{\theta_{m}^{*}})^{2}} \sum_{l \in \mathcal{D}} \frac{1(m \in A_{l})}{f^{2}(A_{l})} \left(\sum_{j \in A_{l}} e^{\theta_{j}^{*}} - e^{\theta_{m}^{*}}\right) \cr 
    &= \frac{1}{(n^{\dagger})^{2}} \sum_{l \in \mathcal{D}} 1(m \in A_{l}) \left(1 - \frac{e^{\theta_{m}^{*}}}{\sum_{j \in A_{l}} e^{\theta_{j}^{*}}}\right) \frac{1}{f(A_{l})} \frac{\sum_{j \in A_{l}} e^{\theta_{j}^{*} - \theta_{m}^{*}}}{f(A_{l})} \cr 
    &\gtrsim \frac{1}{n^{\dagger} e^{\bar{\kappa}}}. 
\end{align*}
}
Similarly, it follows that
{
\begin{align*}
    |\Cov(J_{k}^{*}, J_{m}^{*}|\mathcal{G})| &\asymp \frac{1}{(n^{\dagger})^{2}} \sum_{l \in \mathcal{D}} \frac{1(k, m \in A_{l})}{f^{2}(A_{l})} \frac{(e^{\theta_{k}^{*}} + e^{\theta_{m}^{*}})}{\sum_{u \in A_{l}} e^{\theta_{u}^{*}}} \cr 
    &\lesssim \frac{e^{\bar{\kappa}}}{(n^{\dagger})^{2}} \sum_{l \in \mathcal{D}} 1(k, m \in A_{l}) \cr
    &\asymp \frac{e^{\bar{\kappa}} n^{\ddagger}}{(n^{\dagger})^{2}}, 
\end{align*}
}
uniformly for $k \neq m \in [n]$. By Assumption~\ref{ass4.1}, we have~{$e^{2 \bar{\kappa}} n^{\ddagger}/n^{\dagger} = o(1)$} and hence $\sigma_{km}^{2} \gtrsim 1/(n^{\dagger} e^{\bar{\kappa}})$ for all $k\neq m \in [n]$. Consequently, we have with probability $1-o(1)$, 
{
\begin{align}
\label{eq_uniformly_bounded}
    \max_{l \in \mathcal{D}} \max_{m \in \mathcal{M}} \max_{k \neq m} \left|\frac{\mathcal{V}_{kl}/\nu_{k}^{*} - \mathcal{V}_{ml}/\nu_{m}^{*}}{d \sigma_{km}}\right|^{2} \lesssim \frac{e^{3\bar{\kappa}}}{n^{\dagger}}.
\end{align}
}
Putting all these pieces together, we obtain~\eqref{eq_GA_J_M} by applying Theorem 2.1 in~\citet{Chernozhukov2019}. 
\qed


Combining Lemma~\ref{Lemma_Gaussian_approximation} with Theorem~\ref{thm_approx}, we obtain the asymptotic distribution for $T_{\mathcal{M}}$. 
\begin{theorem}
\label{Theorem_Gaussian_Approximation_T}
Under the conditions of Theorem~\ref{Theorem_Bootstrap_Validity}, we have 
\begin{align}
\label{eq_GA_T_M}
    \sup_{z \in \mathbb{R}} \left|\mathbb{P}(T_{\mathcal{M}} \leq z | \mathcal{G}) - \mathbb{P}\left(\max_{m \in \mathcal{M}} \max_{k \neq m} |Z_{km}|\leq z\right)\right| \to 0. 
\end{align}
\end{theorem}

\noindent\textbf{Proof of Theorem~\ref{Theorem_Gaussian_Approximation_T}.}

Recall~\eqref{eq_sigma_hat_km} for $\{\tilde{\sigma}_{km}\}_{k \neq m}$. It is straightforward to verify that under the conditions of Theorem~\ref{thm_approx} and $\|\tilde{\theta} - \theta^{*}\|_{\infty} \leq 1$, 
{
\begin{align}
\label{eq_ratio_consistency_variance}
    \max_{m \in [n]} \max_{k \neq m} \left|\frac{\tilde{\sigma}_{km}}{\sigma_{km}} - 1\right| & = \cO_P\bigg( \frac{e^{2 \bar{\kappa}} n^{\ddagger}}{n^{\dagger}} + \|\tilde{\theta} - \theta^{\star}\|_{\infty}\bigg) = o_P\left(\frac{1}{\log^{3/2} n}\right).  
\end{align}
}
Therefore, with probability at least $1 - o(1)$, we have
\begin{align*}
    |T_{\mathcal{M}} - \mathcal{J}_{\mathcal{M}}^{*}| \lesssim \max_{m \in \mathcal{M}} \max_{k \neq m} \left|\frac{\vartheta_{m} - \vartheta_{k}}{\sigma_{km}}\right| + \mathcal{J}_{\mathcal{M}}^{*} \max_{m \in \mathcal{M}} \max_{k \neq m} \left|\frac{\sigma_{km}}{\tilde{\sigma}_{km}} - 1\right|,  
\end{align*}
where $\vartheta_{m} = \tilde{\theta}_{m} - \theta_{m}^{*} - J_{m}^{*}$ for each $m \in [n]$. By the proof of Theorem~\ref{thm_approx} and~{the condition $e^{5\bar{\kappa}} n^{\ddagger} n^{1/2} (\log n)^{3} = o(n^{\dagger})$}, 
it follows that
{
\begin{align*}
    \max_{m \in \mathcal{M}} \max_{k \neq m} \left|\frac{\vartheta_{m} - \vartheta_{k}}{\sigma_{km}}\right| = o_P\left(\frac{1}{\log n}\right). 
\end{align*}
}
By Lemma~\ref{Lemma_Gaussian_approximation}, we have $\mathcal{J}_{\mathcal{M}}^{*} = \cO_P(\sqrt{\log n})$. Therefore $|T_{\cM} - \cJ_{\cM}^{*}| = o_P(1/\log n)$. Combining this with~\eqref{eq_GA_J_M} yields~\eqref{eq_GA_T_M}.



\qed

\subsection{Bootstrap Calibration, Proof of Theorem~\ref{Theorem_Bootstrap_Validity}}

For simplicity of notation, we write 
\begin{align*}
    \hat{X}_{mkl} = \frac{1}{d\tilde{\sigma}_{km}}\bigg(\frac{\hat{\mathcal{V}}_{kl}}{\hat{\nu}_{k}} - \frac{\hat{\mathcal{V}}_{ml}}{\hat{\nu}_{m}}\bigg) \enspace \mathrm{and} \enspace X_{mkl} = \frac{1}{d\sigma_{km}}\bigg(\frac{\mathcal{V}_{kl}}{\nu_{k}^{*}} - \frac{\mathcal{V}_{ml}}{\nu_{m}^{*}}\bigg). 
\end{align*}
Here $\hat\nu_m,m\in [n],$ have the same meaning as $\tau_m(\tilde{\theta}),m\in[n]$ defined in the main text.
Define 
\begin{align*}
    \bar{G}_{\mathcal{M}} = \max_{m \in \mathcal{M}} \max_{k \neq m} \left|\sum_{l \in \mathcal{D}} X_{mkl} \omega_{l}\right| 
\end{align*}
and let $\bar{Q}_{1 - \alpha}$ denote its $(1 - \alpha)$th conditional quantile. Following the proof of Lemma~\ref{Lemma_Gaussian_approximation}, by Theorem 2.2 in~\citet{Chernozhukov2019}, we have 
\begin{align*}
    \sup_{\alpha \in (0, 1)} \bigg|\mathbb{P}(\mathcal{J}_{\mathcal{M}}^{*} > \bar{Q}_{1 - \alpha}|\mathcal{G}) - \alpha \bigg| \leq \frac{C e^{3\bar{\kappa}/4}\log^{5/4}(n |\mathcal{D}|)}{(n^{\dagger})^{1/4}}.
\end{align*}
Hence it suffices to upper bound the Kolmogorov distance between the conditional distribution functions of $\bar{G}_{\mathcal{M}}$ and $G_{\mathcal{M}}$. Note that 
\begin{align*}
    |G_{\mathcal{M}} - \bar{G}_{\mathcal{M}}| \leq \max_{m \in \mathcal{M}} \max_{k \neq m} \bigg|\sum_{l \in \mathcal{D}} (\hat{X}_{mkl} - X_{mkl}) \omega_{l}\bigg|. 
\end{align*}
For any $m \in [n]$ and $l \in \mathcal{D}$, it follows that $|\hat{\cV}_{ml}/\hat{\nu}_{m} - \cV_{ml}/\nu_{m}^{*}| \lesssim \cV_{ml} \|\tilde{\theta} - \theta^{*}\|_{\infty}/\nu_{m}^{*}$. Consequently, it follows from~\eqref{eq_uniformly_bounded} and~\eqref{eq_ratio_consistency_variance} that for any $k \neq m$, we have with probability $1-o(1)$,
\begin{align*}
    \sum_{l \in \mathcal{D}} |\hat{X}_{mkl} - X_{mkl}|^{2} \lesssim \left(\|\tilde{\theta} - \theta^{*}\|_{\infty} + \frac{e^{2\bar{\kappa}} n^{\ddagger}}{n^{\dagger}}\right)^{2} \sum_{l \in \mathcal{D}} |X_{mkl}|^{2} \lesssim {e^{3\bar{\kappa}}} \left(\|\tilde{\theta} - \theta^{*}\|_{\infty} + \frac{e^{2\bar{\kappa}} n^{\ddagger}}{n^{\dagger}}\right)^{2}. 
\end{align*}
Consequently, we obtain $|\PP(T_{\mathcal{M}} > \mathcal{Q}_{1 - \alpha}) - \alpha| \to 0$. 

\qed